\documentclass[journal,twocolumn]{IEEEtran}
\usepackage{multirow}
\usepackage{booktabs}
\usepackage{amsfonts}
\usepackage{placeins}
\usepackage{array}
\usepackage{amsmath,cases,amssymb}
\usepackage{amsmath}
\usepackage{mathtools}
\usepackage{graphicx}
\usepackage{subcaption}
\usepackage{cite}
\usepackage{epsfig}
\usepackage{todonotes}
\usepackage{url}
\usepackage{algorithm}
\usepackage{algorithmic}
\usepackage{epstopdf}
\usepackage{balance}
\usepackage{color}
\usepackage{algorithm}
\DeclareGraphicsExtensions{.eps,.pdf}

\usepackage{scalerel}
\usepackage{multicol}
\usepackage{amssymb}
\usepackage{pifont}
\usepackage{cases}

\usepackage{scalerel}
\usepackage{orcidlink}

\let\labelindent\relax
\usepackage{enumitem}

\usepackage{mathbbol}
\usepackage{cite}
\usepackage[T1]{fontenc} 
\usepackage{amsthm} 
\usepackage{eqparbox}
\usepackage{textcomp}
\usepackage[eulergreek]{sansmath}
\usepackage{color}
\usepackage[cmintegrals]{newtxmath}

\usepackage[open,openlevel=1]{bookmark}
\usepackage{ctable}
\usepackage{threeparttablex}
\usepackage{makecell}

\newtheorem{remark}{Remark}
\newtheorem{proposition}{Proposition}

\allowdisplaybreaks

\newcommand{\cmark}{\ding{51}}%
\newcommand{\xmark}{\ding{55}}%

\newcommand{\bseq}{\begin{subequations}}
\newcommand{\eseq}{\end{subequations}}
\newcommand{\baln}{\begin{align}}
\newcommand{\ealn}{\end{align}}
\newcommand{\balnd}{\begin{aligned}}
\newcommand{\ealnd}{\end{aligned}}
\newcommand{\beq}{\begin{equation}}
\newcommand{\eeq}{\end{equation}}
\newcommand{\beqn}{\begin{eqnarray}}
\newcommand{\eeqn}{\end{eqnarray}}
\newcommand{\beqno}{\begin{eqnarray*}}
\newcommand{\eeqno}{\end{eqnarray*}}
\newcommand{\bma}{\begin{displaymath}}
\newcommand{\ema}{\end{displaymath}}
\newcommand{\bnu}{\begin{enumerate}}
\newcommand{\enu}{\end{enumerate}}
\newcommand{\bce}{\begin{center}}
\newcommand{\ece}{\end{center}}
\newcommand{\btb}{\begin{tabular}}
\newcommand{\etb}{\end{tabular}}
\newcommand{\bieq}{\begin{IEEEeqnarray}}
\newcommand{\eieq}{\end{IEEEeqnarray}}

\newcommand{\st}{{\mathrm{s.t.}}}
\newcommand{\subnum}{\IEEEyessubnumber}

\makeatletter
\newcommand{\linebreakand}{%
\end{@IEEEauthorhalign}
\hfill\mbox{}\par
\mbox{}\hfill\begin{@IEEEauthorhalign}
}
\makeatother

\begin{document}
\title{Digital-Twin-Aided Dynamic Spectrum Sharing and Resource Management in Integrated Satellite-Terrestrial Networks \vspace{-2mm} }

\author{ \IEEEauthorblockN{Hung Nguyen-Kha\orcidlink{0000-0002-5956-4279}, \IEEEmembership{Graduate Student Member,~IEEE,} Vu Nguyen Ha\orcidlink{0000-0003-1325-3480}, \IEEEmembership{Senior Member,~IEEE,}\\ Ti Nguyen\orcidlink{}, \IEEEmembership{Member,~IEEE,} Eva Lagunas\orcidlink{0000-0002-9936-7245}, \IEEEmembership{Senior Member,~IEEE}, Joel Grotz\orcidlink{0000-0002-4095-4015}}, \IEEEmembership{Senior Member,~IEEE}, \\
Symeon Chatzinotas\orcidlink{0000-0001-5122-0001}, \IEEEmembership{Fellow,~IEEE}, and Bj\"orn Ottersten\orcidlink{0000-0003-2298-6774}, \IEEEmembership{Fellow,~IEEE}

\IEEEcompsocitemizethanks{ 
This work was supported by the Luxembourg National Research Fund (FNR) under the project INSTRUCT (IPBG19/14016225/INSTRUCT).
The preliminary result of this manuscript is presented in IEEE GLOBECOM'25 \cite{Hung_GLOBECOM25}.
H. Nguyen-Kha, V. N. Ha, T. Nguyen, E. Lagunas, S. Chatzinotas, B. Ottersten are with the Interdisciplinary Centre for Security, Reliability and Trust (SnT), University of Luxembourg, 1855 Luxembourg Ville, Luxembourg.  (e-mail: \{khahung.nguyen; vu-nguyen.ha; titi.nguyen; eva.lagunas; Symeon.Chatzinotas, bjorn.ottersten\}@uni.lu). J. Grotz is with SES, Chateau de Betzdorf, Betzdorf 6815, Luxembourg (e-mail: Joel.Grotz@ses.com). }
}

\maketitle

\begin{abstract} 
The explosive growth in wireless service demand has prompted the evolution of integrated satellite-terrestrial networks (ISTNs) to overcome the limitations of traditional terrestrial networks (TNs) in terms of coverage, spectrum efficiency, and deployment cost. Particularly, leveraging 
LEO satellites and dynamic spectrum sharing (DSS), ISTNs offer promising solutions but face significant challenges due to diverse terrestrial environments, user and satellite mobility, and long propagation LEO-to-ground distance. 
To address these challenges, digital-twin (DT) has emerged as a promising technology to offer virtual replicas of real-world systems, facilitating prediction for resource management.
In this work, we study a time-window-based DT-aided DSS framework for ISTNs, enabling joint long-term and short-term resource decisions to reduce system congestion. Based on that, two optimization problems are formulated, which aim to optimize resource management using DT information and to refine obtained solutions with actual real-time information, respectively. To efficiently solve these problems, we proposed algorithms using compressed-sensing-based and successive convex approximation techniques. 
Simulation results using actual traffic data and the London 3D map demonstrate the superiority in terms of congestion minimization of our proposed algorithms compared to benchmarks. Additionally, it shows the adaptation ability and practical feasibility of our proposed solutions.
\end{abstract}
\vspace{-1mm}

\vspace{-6mm}

\section{Introduction}
\vspace{-2mm}






\IEEEPARstart{I}{n} recent years, the explosive growth in mobile traffic and emerging service demands, such as seamless and ubiquitous connectivity, have driven the evolution of wireless communication. 
While next-generation (NextG) networks are envisioned to meet these requirements \cite{SurveyTut_Roadto6G, ETSI_NTN6G}, relying solely on denser 
TN deployment faces significant challenges, including inefficient spectrum utilization in sparsely populated areas and high infrastructure costs. 
To overcome these limitations, non-terrestrial networks (NTNs), notably satellite networks (SatNet), can offer complementary coverage and traffic offloading. 
Leveraging both TN and SatNet capabilities, 
ISTNs have attracted much attention as a viable architecture to meet NextG requirements, advancing the goal of global, uninterrupted coverage.
Backed by regulatory and institutional bodies such as the US Federal Communications Commission (FCC) and the European Space Agency (ESA), ISTNs are recognized as enablers of direct-to-device (D2D) connectivity and TN support \cite{FCC2322_Supplemental_Space_Coverage, ESA_D2D_conf}.
Among the SatNet options, low Earth orbit (LEO) satellites (LSats) offer lower latency and higher channel gain, making them especially suited for tight TN-NTN integration and native support of NTNs in 6G \cite{ETSI_NTN6G}.
 
To realize ubiquitous connectivity in 6G with ISTNs, efficient use of radio spectrum across TNs and SatNets becomes critical.
Traditionally, these two networks have operated in separate radio frequency bands (RFBs), with TNs utilizing lower bands and SatNet relying on higher ones. 
However, due to limited spectrum availability, this static allocation strategy poses limitations, especially in densely deployed or heterogeneous traffic environments.
To address these constraints, dynamic spectrum sharing (DSS) between TNs and SatNets appears as an important strategy within ISTNs to improve spectral efficiency.  
Recent proposals, such as MediaTek’s spectrum reuse framework, highlight the potential of allowing SatNet to opportunistically access terrestrial spectrum \cite{3gpp.Mediatek.Tdoc_RWS230110}. This vision is further reinforced by regulatory bodies like the FCC and ESA, which support co-primary spectrum use between TN and NTN systems, particularly for D2D connectivity \cite{FCC2322_Supplemental_Space_Coverage, ESA_D2D_conf}.
Nevertheless, enabling spectrum sharing in ISTNs introduces critical challenges, especially inter-system interference (ISyI) between coexisting TN and SatNet transmissions. 


Effectively managing ISyI is essential to realize the full potential of DSS in 6G-ISTN architectures. This, however, demands advanced technologies to address key challenges stemming from the wide coverage areas, high mobility of LSats, and the long LSat-to-ground (Sat2G) distance. Moreover, LSats often serve regions with diverse terrestrial conditions, especially in urban areas, where dense buildings introduce reflection, diffraction, and signal blockage, further complicating interference patterns. In addition, the movement of both LSats and users leads to rapidly varying link conditions, while the long Sat2G distance increases the overhead for signaling and real-time channel estimation.
To address these challenges, digital-twin (DT) has emerged as a promising technology to offer virtual replicas of real-world systems, which enables capturing and emulating the evolution of the physical environments, network dynamics, and entities in real time \cite{Spirent_Whitepaper5GDT, Huan_ComMag21_DT} through continuous monitoring. In the context of ISTNs, utilizing DT to replicate traffic demand, environmental conditions \cite{Zhu_OJCOM24_DTRT}, and channel state information (CSI) can significantly enhance resource management (RM), support proactive interference mitigation, and reduce Sat2G control signaling by using DT information. This work leverages DT to address the key challenges of ISyI, DSS coordination, and RM in practical ISTN deployments.

\vspace{-3mm}

\subsection{Related Works}
\vspace{-1mm}

Spectrum sharing (SS) designs in ISTNs have been considered in literature \cite{du_JSAC2018, zhang_JSAC2022, Lee_TVT23, Li_TWC2024, martikainen_WoWMoM2023, Zhu_TWC24_BeamManage, Hung_TCOM25}. In \cite{du_JSAC2018, zhang_JSAC2022, Lee_TVT23, Li_TWC2024}, snapshot-based SS solutions for ISTNs were considered.
In particular, \cite{du_JSAC2018} designed the second-price auction mechanism to achieve an equilibrium SS solution for TNs and SatNet-based traffic offloading. \cite{zhang_JSAC2022} considered a backhaul and access SS framework, where one satellite served both base stations (BSs) and users (UEs). Limited backhaul link capacity, UE association, allocates bandwidth (BW) assignment, and power control were addressed through the corresponding sub-problems. 
\cite{Lee_TVT23} and \cite{Li_TWC2024} studied the RA problems for throughput maximization in reverse SS scenarios \cite{Lee_TVT23} and in underlay/overlay sharing cases \cite{Li_TWC2024}.
However, these snapshot-based works have not considered ubiquitous and seamless connectivity and struggle to capture the dynamic behavior of ISTNs.


The time-window-based SS systems have been examined in \cite{martikainen_WoWMoM2023, Zhu_TWC24_BeamManage, Hung_TCOM25}.
Specifically, \cite{martikainen_WoWMoM2023} studied DSS for ISTNs, where a centralized coordinator allocates BW to each network based on traffic load.
In \cite{Zhu_TWC24_BeamManage}, the TN spectrum was temporarily shared with SatNet. The study focused on beam management, scheduling, and DSS for only SatNet, while TNs were considered under interference constraints. Both \cite{martikainen_WoWMoM2023, Zhu_TWC24_BeamManage} have not considered interference and service requirements at the user side. 
Additionally, their designs mainly rely on the statistical channel model. 
Filling these gaps, our previous work \cite{Hung_TCOM25} has utilized 3D map and the ray-tracing (RayT) mechanism to examine the seamless connectivity and coverage enhancement in urban areas. {However, \cite{Hung_TCOM25} did not consider BW and RB allocation.} 


Recently, the DT technology has been utilized in RM \cite{Huynh_TCOM22_DT, Sun_IoT22_DT, Gong_TVT23_DT, Wu_IoT25_DT, Yu_IoT25_DT, Tang_TITS25_DT}.  However, existing DT models are predominantly used to replicate the computation information \cite{Huynh_TCOM22_DT, Sun_IoT22_DT, Gong_TVT23_DT, Wu_IoT25_DT}, traffic/load patterns \cite{Sun_IoT22_DT, Tang_TITS25_DT}, and node positions \cite{Sun_IoT22_DT, Gong_TVT23_DT, Yu_IoT25_DT}. Besides, given the crucial role of CSI in wireless RM, \cite{Tang_TITS25_DT} incorporated CSI into the DT model using the statistical channel. Nevertheless, CSI is strongly affected by the surrounding environment, particularly in complex areas such as urban scenarios \cite{Hung_MeditCom24, Hung_TCOM25}.
Therefore, modeling the environment and its influence on CSI is essential, and should be incorporated into DT frameworks using 3D maps and RayT mechanisms \cite{Zhu_OJCOM24_DTRT}. 
To the best of our knowledge, the joint design of BW allocation, traffic steering, AP/LSat-UE association, RB assignment, and power control in DT-aided DSS ISTNs, incorporating a realistic channel model based on actual 3D maps, has not been thoroughly investigated. This motivates us to fill this research gap in this work. The comparison between our work and the literature is summarized in Table~\ref{tab:works}.

\begin{table*}[t]
    \captionsetup{font=small}
    \centering
    \caption{Comparison of related works.}
    \label{tab:works}
    \vspace{-2mm}
    \scalebox{0.72}{
    \begin{tabular}{|p{0.55cm}|p{0.5cm}|p{1.55cm}|p{2.4cm}|p{0.25cm}|p{0.25cm}|*{15}{c|}}
        \hline
        \multirow{2}{*}{\textbf{Works}} &
        \multirow{2}{*}{\textbf{Year}} &
        \multirow{2}{*}{\textbf{Mode}} &
        \multicolumn{9}{c|}{\textbf{Resource/Network Management}} &
        \multicolumn{5}{c|}{\textbf{Digital Twin Repication}} & \multirow{2}{*}{\textbf{Channel model}}
        \\
        \cline{4-17}
        & & & \textbf{SS mode} & \textbf{PC} & \textbf{UA} & \textbf{RB/SC} & \textbf{QoS} 
        & \textbf{Intra SI} & \textbf{Inter SI} & \textbf{BWA} & \textbf{T/L/C} 
        & \textbf{T/L} & \textbf{Computing} & \textbf{Position} & \textbf{3D Env} & \textbf{CSI} & \\
        \hline
        \cite{du_JSAC2018} & 2018 & Snapshot & Primary-secondary 
        & \xmark & \xmark & \cmark & \xmark 
        & \xmark & \cmark & \xmark & \xmark 
        & \xmark & \xmark & \xmark & \xmark 
        & \xmark& Statistic \\
        \hline
        \cite{zhang_JSAC2022} & 2022 & Snapshot & Non-overlap 
        & \cmark & \cmark & \xmark & Rate 
        & \cmark & \cmark & \cmark & \xmark
        & \xmark & \xmark & \xmark & \xmark 
        & \xmark & Statistic \\
        \hline
        \cite{Lee_TVT23} & 2023 & Snapshot & Overlap, non-overlap 
        & \xmark & \xmark & \cmark & \xmark 
        & \cmark & \cmark & \xmark & \xmark
        & \xmark & \xmark & \xmark & \xmark
        & \xmark & Statistic \\
        \hline
        \cite{Li_TWC2024} & 2024 & Snapshot & Primary-secondary 
        & \cmark & \xmark & \xmark & Rate 
        & \cmark & \cmark & \xmark & \xmark 
        & \xmark & \xmark & \xmark & \xmark
        & \xmark & Statistic \\
        \hline
        \cite{martikainen_WoWMoM2023} & 2023 & Time window & Non-overlap 
        & \cmark & \xmark & \xmark & \xmark 
        & \xmark & \xmark & \cmark & \cmark
        & \xmark & \xmark & \xmark & \xmark
        & \xmark & Statistic \\
        \hline
        \cite{Zhu_TWC24_BeamManage} & 2024 & Time window & Overlap 
        & \xmark & \cmark & \xmark & \xmark
        & \cmark & \cmark & \xmark & \cmark
        & \xmark & \xmark & \xmark & \xmark
        & \xmark & Statistic \\
        \hline
        \cite{Hung_TCOM25} & 2025 & Time window & Co-primary 
        & \cmark & \cmark & \xmark & rate
        & \xmark & \cmark & \xmark & \xmark
        & \xmark & \xmark & \xmark & \xmark
        & \xmark & 3D map-based \\
        \hline
        \cite{Huynh_TCOM22_DT} & 2022 & Snapshot & None 
        & \cmark & \cmark & \xmark & Latency 
        & \cmark & \xmark & \xmark & \cmark
        & \xmark & \cmark & \xmark & \xmark
        & \xmark & Statistic \\
        \hline
        \cite{Sun_IoT22_DT} & 2022 & Snapshot & None 
        & \xmark & \xmark & \xmark & \xmark 
        & \xmark & \xmark & \xmark & \cmark
        & \cmark & \cmark & \cmark & \xmark
        & \xmark & None \\
        \hline
        \cite{Gong_TVT23_DT} & 2023 & Snapshot & None 
        & \xmark & \xmark & \xmark & Latency 
        & \xmark & \xmark & \cmark & \cmark
        & \xmark & \cmark & \cmark & \xmark
        & \xmark & Statistic \\
        \hline
        \cite{Wu_IoT25_DT} & 2025 & Snapshot & None 
        & \xmark & \cmark & \xmark & \xmark 
        & \xmark & \xmark & \xmark & \cmark
        & \xmark & \cmark & \xmark & \xmark
        & \xmark & Statistic \\
        \hline
        \cite{Yu_IoT25_DT} & 2025 & Time window & None 
        & \xmark & \cmark & \cmark & Latency 
        & \cmark & \xmark & \xmark & \xmark
        & \xmark & \xmark & \cmark & \xmark
        & \xmark & Statistic \\
        \hline
        \cite{Tang_TITS25_DT} & 2025 & Time window & None 
        & \xmark & \cmark & \xmark & Rate/Latency 
        & \xmark & \xmark & \cmark & \cmark
        & \cmark & \xmark & \xmark & \xmark
        & \cmark & Statistic \\
        \hline
        \multicolumn{2}{|c|}{Our work} & Time window & Slice, co-primary 
        & \cmark & \cmark & \cmark & Latency 
        & \cmark & \cmark & \cmark & \cmark 
        & \cmark & \xmark & \cmark  &  \cmark
        & \cmark & 3D map-based \\
        \hline
        \multicolumn{18}{|l|}{PC: Power control, UA: User association, RB/SC: Resource-block/Sub-channel assignment, Inter/Intra SI: Inter/Intra system interference,  BWA: Bandwidth allocation, T/L/C: Traffic/Load/Computing-resource} \\
        \hline
    \end{tabular}
    }
\end{table*}

\vspace{-3mm}
\subsection{Research Contributions}
\vspace{-1mm}
According to the discussed research gap as well as D2D scenarios driven by ESA and FCC, this work studies novel DT-aided DSS mechanisms for ISTNs supporting D2D connectivity, where TNs and SatNets share the same 5G-NR RFBs. Regarding uneven traffic demand among services and systems, the work aims to minimize system congestion by improving spectrum usage through DSS and RM under practical constraints. This work's main contributions are summarized as:
\begin{itemize}
    \item We investigate the 5G-NR transmission model for ISTNs under a practical system architecture and develop a novel DT model to replicate key components of the network. Beyond replicating traffic demand and terminal positions, as done in prior works, our model incorporates realistic terminal antenna patterns, an actual 3D map of the target area, and a RayT tool to accurately emulate LSat/TN access point (TAP)-to-UE CSIs. 
    \item We formulate two resource-allocation (RA) optimization problems, ``\textit{joint-RA}'' and ``\textit{refinement}'', which respectively rely on DT-predicted information and real-time system feedback, such as CSIs and traffic arrival. The joint-RA problem leverages DT-predicted information to optimize both long-term RA decisions (i.e., traffic steering and BW allocation across services and TN/SatNet systems) and preliminary estimated short-term RA decisions (i.e., TAP/LSat–UE association, RB assignment, and power control). The refinement problem adjusts the TN short-term decisions based on real-time feedback to align the DT-optimized strategy with current environmental conditions. Both problems are challenging due to non-convex constraints; notably, the joint-RA problem is formulated as a mixed-integer non-linear program (MINLP) and relies on predicted information.
    \item By leveraging the DT model for predictive information and applying compressed sensing and successive convex approximation (SCA) techniques, we propose two efficient solution methods, named DT-based joint RA and Real-time refinement algorithms (\textit{DT-JointRA} and \textit{RT-Refine}), that address the problems in alignment with the system architecture.
    \item For practical evaluation, we utilize a realistic 3D map of London, mobility of UEs such as vehicle UEs extracted from UE route in Google Navigator, and actual data traffic datasets. Numerical results demonstrate that our proposed algorithms significantly outperform benchmarks in minimizing queue lengths. Moreover, with DT-assisted prediction and pre-optimization in DT-JointRA, the RT-Refine scheme achieves fast convergence within a few iterations, highlighting the practical feasibility and efficiency of the proposed solutions.
\end{itemize}
The rest of this paper is organized as follows. Section~\ref{sec: sysmodel} presents the system model and problem formulation. Section~\ref{sec: solution} and \ref{sec: benchmark-complexity} describe the proposed solutions, benchmarks, and complexity analysis, respectively. Section~\ref{sec: result} and \ref{sec: conclusion} provide numerical results and the conclusion.

\vspace{-2mm}
\section{5G-NR and Digital Twin System Model}\label{sec: sysmodel}
\vspace{-1mm}
\begin{figure*}
	\centering
 	\includegraphics[width=1\linewidth]{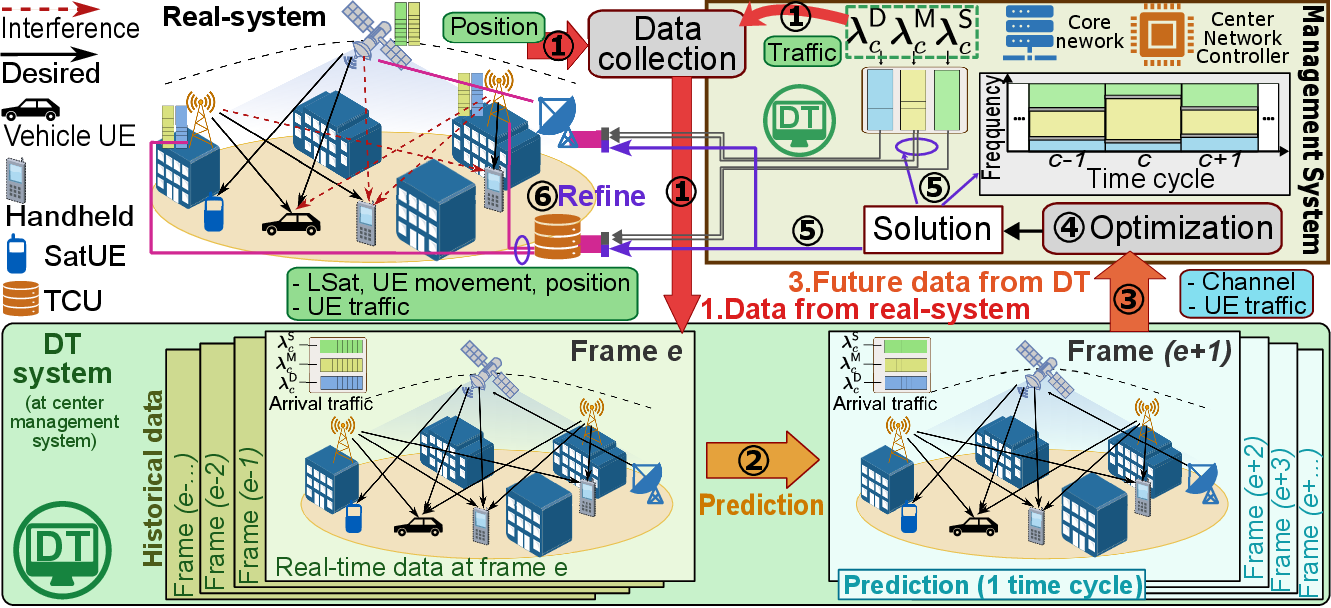}
    \captionsetup{font=footnotesize}
    \vspace{-4mm}
	\caption{Digital-twin-aided ISTNs.}
	\label{fig:sysmodel}
    \vspace{-6mm}
\end{figure*}
This paper considers a 5G-NR-based downlink transmission in an ISTN consisting of $N$ TAPs and one LSat jointly serving $K$ UEs, as depicted in Fig.~\ref{fig:sysmodel}. 
It is assumed that SatNet employs the FDD strategy to manage downlink (DL) and uplink (UL) transmissions in distinct RFBs, whereas TN operates under the TDD mechanism for allocating radio resources between DL and UL in the same RFB \cite{6GNTN_EU_SpectrumPolicies}.
In this context, we focus on utilizing the RFB, denoted $W^{\sf{tot}}$~(Hz), which is assumedly accessible to TN for both DL and UL with TDD management as well as the SatNet DL transmission. 
We assume that synchronization is achieved based on network information while handover procedures follow user association.
The key notations used in this work are described in Table~\ref{tab:notation}

\subsection{Services and DT-based Network Management}
\vspace{-1mm}
\begin{figure}
	\centering
	\includegraphics[width=1 \linewidth]{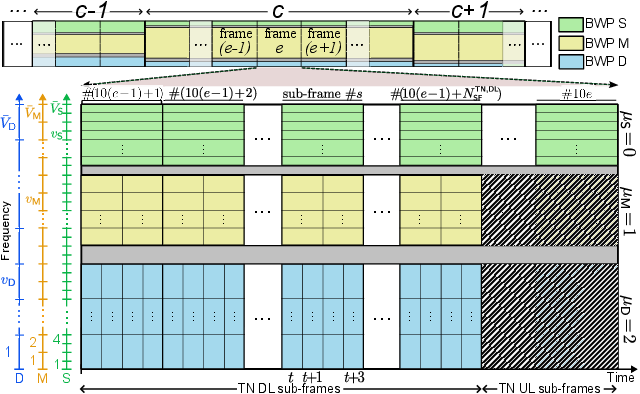}
    \vspace{-6mm}
    \captionsetup{font=footnotesize}
	\caption{Resource block grid.}
	\label{fig:RBGrid}
\end{figure}
The system supports three types of UEs, each associated with a specific communication service: (i) Delay-sensitive (${\sf{D}}$), such as uRLLC, exclusively served by TAPs; (ii) SatCom (${\sf{S}}$), provided solely by the LSat; and (iii) Multinet (${\sf{M}}$) supported by both TAPs and LSat.
The sets of UEs associated with these three services and set of all UEs are denoted by $\mathcal{K}_{\sf{D}}$, $\mathcal{K}_{\sf{S}}$, $\mathcal{K}_{\sf{M}}$, and $\mathcal{K} = \mathcal{K}_{\sf{D}} \cup \mathcal{K}_{\sf{S}} \cup \mathcal{K}_{\sf{M}}$ while the corresponding numbers of UEs are $K_{\sf{D}}=|\mathcal{K}_{\sf{D}}|$, $K_{\sf{S}}=|\mathcal{K}_{\sf{S}}|$, and $K_{\sf{S}}=|\mathcal{K}_{\sf{M}}|$, respectively; ${\sf{UE}}_{k}$ is the $k$-th UE.
The total BW is dynamically divided into three BW parts (BWPs) for three services. 
Each BWP employs a 5G-NR numerology, i.e., $\mu_{\sf{D}}$, $\mu_{\sf{M}}$, and $\mu_{\sf{S}}$, tailored to its respective latency perspective.
Let $\mathcal{L}$, $\mathcal{S}$, $L=|\mathcal{L}|$, and ${\sf{AP}}_{\ell}$ be the sets of TAPs, services, number of TAPs, and $\ell$-th TAP, respectively. 

As depicted in Fig.~\ref{fig:sysmodel}, this paper aims to jointly optimize \textit{``long-term decisions''}, i.e., traffic steering from the core network (CN) to the LSat and TAPs, BW allocation (BWA) between TN and SatNet systems and across BWPs--executed at the central network management system (NMS)--as well as \textit{``short-term decisions''}, i.e., UE association, RB assignment, and power control--executed at TN center unit (TCU) and gateway. 
To facilitate such optimization, we assume the availability of a DT model with the virtual representation of network components to enable the prediction of network information. \textit{The DT model will be presented in Section~\ref{sec:DTmodel}}. The proposed optimization framework operates over multiple time scales, including cycle, frame, and sub-frame (SF), which will be detailed later.

This paper considers a time window of $N_{\sf{Cy}}$ cycles. As depicted in Fig.~\ref{fig:sysmodel}, at the end of each cycle, the DT model updates its state (\textbf{1. Data from real-system}), predicts (\textbf{2. Prediction}), and provides predicted system-level information of the subsequent cycle to the NMS (\textbf{3. Future data from DT}). Leveraging this foresight, the NMS performs joint resource management (\textbf{Optimization step}), generating long-term decisions along with preliminary estimates of short-term ones (at \textbf{Solution block}). This stage is referred as the \textit{\textbf{``DT-based Joint RA Stage''}}. 
Each cycle is composed of multiple frames, with each frame consisting of $10$ sub-frames (SFs).
We assume that traffic rate information for $\sf{D}$ UEs can be updated at the SF level  \cite{kavehmadavani_TWC23}. Hence, \textit{``short-term decisions''} for the TN is further refined at the SF granularity (at \textbf{TCU}). This SF-level refinement process constitutes the \textit{\textbf{``RT-Refine Stage''}}.

\begin{table*}
    \centering
    \caption{Notation description.}
    \vspace{-2mm}
    \label{tab:notation}
    \scalebox{0.92}{
    \begin{tabular}{|l|l|l|l|}
        \hline 
        \multicolumn{4}{|c|}{\textbf{System notation}} \\
        \hline
        BWP, NMS, TAP, TCU & \multicolumn{3}{l|}{Bandwidth part, network management system, TN access point, and central unit} \\
        $\mathcal{K}$, $\mathcal{L}$, $\mathcal{S}$ & \multicolumn{3}{l|}{Set of all users, TAPs, and services} \\
        $\mathcal{K}_{\!\sf{D}}$, $\mathcal{K}_{\!\sf{S}}$, $ \mathcal{K}_{\!\sf{M}}$, $K_{\!\sf{D}}$, $K_{\!\sf{S}}$, $K_{\!\sf{M}}$  & \multicolumn{3}{l|}{Sets and numbers of users using delay-sensitive, SatCom, and multinet services} \\
        $\tilde{h}_{\ell,k}$, $\tilde{g}_{k}$ & \multicolumn{3}{l|}{Channel coefficients from ${\sf{TAP}}_{\ell}$ and LSAT to user $k$} \\
        $\{\boldsymbol{h}_{c}, \boldsymbol{g}_{c} \}$, $\{\boldsymbol{h}_{e}, \boldsymbol{g}_{e} \}$ & \multicolumn{3}{l|}{Channel gain sets in cycle $c$ and in frame $e$} \\
        $\tilde{h}_{\ell,k}^{{\sf{ls}}}$, $\tilde{h}_{\ell,k}^{{\sf{nl}}}$, $\tilde{g}_{k}^{{\sf{ls}}}$, $\tilde{g}_{k}^{{\sf{nl}}}$ & \multicolumn{3}{l|}{LoS and NLoS components in the real-system} \\
        $\mu_{\sf{D}}=2$, $\mu_{\sf{M}}=1$, $\mu_{\sf{S}}=0$ & \multicolumn{3}{l|}{Numerlogoies used in BWPs for corresponding services} \\
        \hline 
        \multicolumn{4}{|c|}{\textbf{Time-Frequency domain notation}} \\
        \hline
        $[v_{\sf{x}}, n_{\sf{x}}]$ & Frequency-time RB index in BWP $\sf{x}$ (used for service $\sf{x}$) &
        $c$, $e$, $s$ & Cycle, frame, and sub-frame indices \\
        $N_{\sf{TF}}$ and $N_{\sf{SF}}$ & Number of frames and sub-frames (SF) in one cycle &
        $N_{\sf{Cy}}$ & Number of time cycles \\
        $N_{\sf{x}}$ & Number of time indices of RBs for service $\sf{x}$ in one SF &
        $T_{\sf{x}}$ & Time duration of RB in BWP $\sf{x}$ \\
        $\mathcal{T}_{\sf{x}}^{{\sf{Cy}},c}$, $\mathcal{T}_{\sf{x}}^{{\sf{TF}},e}$, $\mathcal{T}_{\sf{x}}^{{\sf{SF}},s}$ & Time index sets of RBs for service $\sf{x}$ in cycle $c$, frame $e$, sub-frame $s$ &
        $N_{\sf{SF}}^{\sf{TN,DL}}$ & Number of TN DL SFs in each frame \\
        \hline 
        $[f^{\sf{min}},f^{\sf{max}}]$ & Operational frequency range &
        $W^{\sf{total}} $ & System bandwidth \\
        $W_{{\sf{x}},c} = [f^{\sf{min}}_{{\sf{x}},c},f^{\sf{max}}_{{\sf{x}},c}]$ & Frequency ranges of BWPs sliced for service $\sf{x}$ in cycle $c$ &
        $w_{\sf{x}}$ & RB channel spacing in BWP $\sf{x}$ \\
        $\bar{V}_{\sf{x}} = \lfloor W^{\sf{tot}} / w_{{\sf{x}}} \rfloor$ & Maximum SC numbers that can be assigned to service ${\sf{x}}$ &
        $W_{\sf{G1}}$, $W_{\sf{G2}}$ & Guard bands between BWPs \\
        $\mathcal{V}_{\sf{x}} \triangleq \{ 1,\dots, \bar{V}_{\sf{x}} \}$ & Set of all SCs that can be assigned to service $\sf{x}$ &
        $v_{\sf{x}} \in \mathcal{V}_{\sf{x}}$ & index of the SCs in BWP ${\sf{x}}$ \\
        \hline
        \multicolumn{4}{|c|}{\textbf{Digital-twin system}} \\
        \hline
        $\widehat{\sf{AP}}_{\ell}$, ${\sf{A}}_{\ell}^{\sf{ap}}$, $\boldsymbol{u}_{\ell}^{\sf{ap}}$ & TAP, TAP's position, and TAP's antenna property in TN cell $\ell$ &
        $\widehat{\sf{UE}}_{k,[e]}^{\sf{x}}$ & User $k$ using service $\sf{x}$ at frame $e$ \\
        $\widehat{\sf{Sat}}_{[e]}$, ${\sf{A}}^{\sf{sat}}$, $\hat{\boldsymbol{u}}^{\sf{sat}}_{[e]}$, ${\sf{TLE}}_{[e]}\!\!$ & LSAT, antenna, position, and TLE data in frame $e$ &
        $\widehat{\sf{env}}_{\ell}$, ${\sf{map}}_{\ell}$ & DT environment, 3D map in TN cell $\ell$ \\
        ${\sf{A}}_{k}^{\sf{ue}}$, $\hat{\boldsymbol{u}}_{k,[e]}^{{\sf{ue}}}$, $\hat{\lambda}_{k,[e]}^{\sf{TF,x}}$, $\hat{\lambda}_{k,[s]}^{\sf{SF,D}}\!\!$ & User antenna, position, traffic in frame $e$; traffic $\sf{D}$ in sub-frame $s$ &
        ${\sf{chan}}_{[e]}$;$\bar{h}_{\ell,k}^{{\sf{nl}}}$, $\bar{g}_{k}^{{\sf{nl}}}\!\!\!$ & Channel set; NLoS components \\
        \hline
        \multicolumn{4}{|c|}{\textbf{Transmission model notation}} \\
        \hline
        $\gamma_{\ell,k}^{{\sf{x}},[v_{\sf{x}},n_{\sf{x}}]}$ & SINR of signal from ${\sf{TAP}}_{\ell}$ at ${\sf{UE}}_{k}$ via ${\sf{RB}}_{[v_{\sf{x}}, n_{\sf{x}}]}$ &
        $\Psi_{\ell,k}^{[v_{\sf{x}},n_{\sf{x}}]}(\! \boldsymbol{p}_{\!c},\! \boldsymbol{\alpha}_{\!c}\!)$ & Inter-cell-interference (ICI) power \\
        $\Theta_{k}^{{\sf{TN}},[v_{\sf{M}},n_{\sf{M}}]}\!(\! \boldsymbol{P}_{\!c},\! \boldsymbol{\beta}_{\!c} \!)$ &  Inter-system-interference (ISyI) caused by the LSat &
        $\Theta_{k}^{{\sf{Sat}},[v_{\sf{M}},n_{\sf{M}}]}\!(\! \boldsymbol{P}_{\!c},\! \boldsymbol{\alpha}_{\!c} \!)\!\!\!\!$ & ISyI caused by TAPs \\
        $R_{\ell,k}^{{\sf{M}},[n_{\sf{M}}]}\!(\! \boldsymbol{P}_{\!c},\! \boldsymbol{\alpha}_{\!c},\! \boldsymbol{\beta}_{\!c} \!)$ & Rate of ${\sf{UE}}_{k}^{\sf{M}}$ served by ${\sf{AP}}_{\ell}$ at RB time $n_{\sf{M}}$ &
        $R_{\ell,k,[s]}^{{\sf{SF,D}}}(\boldsymbol{p}_{\!c},\boldsymbol{\alpha}_{\!c})$ & Rate of ${\sf{UE}}_{k}^{\sf{D}}$ at SF $s$ \\
        $\gamma_{0,k}^{{\sf{M}},[v_{\sf{M}},n_{\sf{M}}]}$ & SINR of signal from LSAT at ${\sf{UE}}_{k}^{\sf{M}}$ over ${\sf{RB}}_{[v_{\sf{M}},n_{\sf{M}}]}$ &
        $R_{0,k}^{{\sf{M}},[n_{\sf{M}}]}\!(\! \boldsymbol{P}_{\!c},\! \boldsymbol{\beta}_{\!c} \!)$ & Rate at RB time $n_{\sf{M}}$ of ${\sf{UE}}_{k}^{\sf{M}}$ from LSat \\
        $\gamma_{0,k}^{{\sf{S}},[v_{\sf{S}},n_{\sf{S}}]}$ & SNR of signal from LSAT at ${\sf{UE}}_{k}^{\sf{S}}$ over ${\sf{RB}}_{[v_{\sf{S}},n_{\sf{S}}]}$ &
        $R_{0,k}^{{\sf{S}},[n_{\sf{S}}]}\!(\! \boldsymbol{P}_{\!c},\! \boldsymbol{\beta}_{\!c} \!) $ & Rate at RB time $n_{\sf{S}}$ of ${\sf{UE}}_{k}^{\sf{M}}$ \\
        \hline
        \multicolumn{4}{|c|}{\textbf{Variables}} \\
        \hline
        $b^{\sf{x}}_{v_{\sf{x}},c}$, $\boldsymbol{b}_{c} \triangleq \{b^{\sf{x}}_{v_{\sf{x}},c} \}$ & \multicolumn{3}{l|}{Binary variable indicating using SC $v_{\sf{x}}$ or not and its set} \\
        $\alpha_{\ell,k}^{[v_{\sf{x}},n_{\sf{x}}]}$, $\boldsymbol{\alpha}_{c} = [\alpha_{\ell,k}^{[v_{\sf{x}},n_{\sf{x}}]}]$ & \multicolumn{3}{l|}{Binary variable indicating $\sf{AP}_{\ell}$ serves $\sf{UE}_{k}$ over $\sf{RB}_{[v_{\sf{x}},n_{\sf{x}}]}$ or not and its matrix in cycle $c$} \\
        $\beta_{k}^{[v_{\sf{x}},n_{\sf{x}}]}$, $\boldsymbol{\beta}_{c}= [\beta_{k}^{[v_{\sf{x}},n_{\sf{x}}]}]$ & \multicolumn{3}{l|}{Binary variable indicating $\sf{UE}_{k}$ $\sf{LSat}$ over ${\sf{RB}}_{[v_{\sf{x}} ,n_{\sf{x}}]}$ or not and its matrix in cycle $c$} \\
        $p_{\ell,k}^{[v_{\sf{x}},n_{\sf{x}}]}$, $\boldsymbol{p}_{c} = [p_{\ell,k}^{[v_{\sf{x}},n_{\sf{x}}]}]$ & \multicolumn{3}{l|}{Transmit power from ${\sf{AP}}_{\ell}$ to ${\sf{UE}}_{k}$ over ${\sf{RB}}_{[v_{\sf{x}},n_{\sf{x}}]}$ and its matrix in cycle $c$} \\
        $p_{0,k}^{[v_{\sf{x}},n_{\sf{x}}]}$, $\boldsymbol{p}_{0,c} = [p_{0,k}^{[v_{\sf{x}},n_{\sf{x}}]} \!]\!$ & \multicolumn{3}{l|}{Transmit power from LSAT to ${\sf{UE}}_{k}$ over ${\sf{RB}}_{[v_{\sf{x}},n_{\sf{x}}]}$ and its matrix in cycle $c$} \\
        $\omega_{k,c}^{\sf{CN,M}}$, $\omega_{\ell,k,c}^{\sf{x}}$ & \multicolumn{3}{l|}{Traffic steering from CN to TN; and that from TCU to $\sf{AP}_{\ell}$} \\
        $q_{\ell,k}^{{\sf{M}},[n_{\sf{M}}]}$, $q_{0,k}^{{\sf{M}},[n_{\sf{M}}]}$, $q_{0,k}^{{\sf{S}},[n_{\sf{S}}]}$ & \multicolumn{3}{l|}{Queue lengths (QLs) at ${\sf{AP}}_{\ell}/{\sf{LSat}}$ of flow $k$ of services ${\sf{M,S}}$ at each RB time} \\
        \hline
    \end{tabular}
    }
\end{table*}

\vspace{-4mm}
\subsection{5G-NR Standard-based Setting}
\vspace{-1mm}
\subsubsection{Frame and Sub-frame Setting}
According to 5G-NR standard, each frame duration is $10$~ms ($T_{\sf{TF}}=10$~ms); hence, duration of one SF is $1$~ms ($T_{\sf{SF}}=1$~ms). 
Assuming that each cycle $c$ contains $N_{\sf{TF}}$ frames, the number of SFs in one cycle is 
$N_{\sf{SF}} = 10 N_{\sf{TF}}$. 
Regarding the TDD transmission mode for TN, we assume that TN uses $N_{\sf{SF}}^{\sf{TN,DL}}$ beginning SFs for DL and $(10-N_{\sf{SF}}^{\sf{TN,DL}})$ ending SFs for UL. 

\subsubsection{Numerology Setting}
In this paper, numerologies used in BWPs are set as $\mu_{\sf{S}} = 0, \mu_{\sf{M}} = 1$ and $\mu_{\sf{D}} = 2$, reflecting the increasing tolerance to delay across the three service types \cite{Kihero_Access19,Korrai_OJCOM20}. Based on this, the RB grid is illustrated in Fig.~\ref{fig:RBGrid}.
\paragraph{Time Domain}
Hereafter, we define the RB duration for service $\sf{D}$ as a \textit{time slot} (TS) - the time unit.
Based on that, the time indices $n_{\sf{x}}$ of 
RBs of service $\sf{x}$ is defined as
\vspace{-2mm}
\begin{equation}\label{eq: time index}
    n_{\sf{D}} = t, \quad n_{\sf{M}} = \lceil t/2 \rceil, \quad
    n_{\sf{S}} = \lceil t/4 \rceil,     \vspace{-2mm} 
\end{equation}
respectively, where $\lceil \cdot \rceil$ is the ceil function.
Furthermore, the RB duration of service $\sf{x}$ can be given as $T_{\sf{x}} = 2^{-\mu_{\sf{x}}}$~ms; hence, there are $N_{\sf{x}}=2^{\mu_{\sf{x}}}$ RBs over one sub-channel of service $\sf{x}$ within a SF.
Then, the time indices of cycles ($c$), frames ($e$), sub-frames ($s$) can be counted based on $N_{\sf{D}}$ as
$c = \lceil t/(N_{\sf{D}}N_{\sf{SF}}) \rceil, e = \lceil t/(10 N_{\sf{D}}) \rceil, s = \lceil t/N_{\sf{D}} \rceil$.
In addition, the sets of RB indices in the time domain corresponding to service $\sf{x}$ within cycle $c$, frame $e$, and SF $s$ can be defined as $\mathcal{T}_{\sf{x}}^{{\sf{Cy}},c} \triangleq \{ (c-1)N_{\sf{x}}N_{\sf{SF}} +1,\dots, c N_{\sf{x}}N_{\sf{SF}} \}$, $\mathcal{T}_{\sf{x}}^{{\sf{TF}},e} \triangleq \{ (e-1)10N_{\sf{x}} +1,\dots, e 10 N_{\sf{x}}\}$, and $\mathcal{T}_{\sf{x}}^{{\sf{SF}},s} \triangleq \{ (s-1)N_{\sf{x}} +1,\dots, s N_{\sf{x}}\}$.

\paragraph{Frequency Domain}
Let $W_{{\sf{x}},c} = [f^{\sf{min}}_{{\sf{x}},c},f^{\sf{max}}_{{\sf{x}},c}]$ be the frequency ranges of BWPs sliced for service $\sf{x}$ in cycle $c$.
For simplicity, we set $f^{\sf{max}}_{{\sf{D}},c} \leq f^{\sf{min}}_{{\sf{M}},c}$ and $f^{\sf{max}}_{{\sf{M}},c} \leq f^{\sf{min}}_{{\sf{S}},c}$ for all $c$.
Then, the BW of service $\sf{x}$ in cycle $c$ is given as $W_{{\sf{x}},c} = f^{\sf{max}}_{{\sf{x}},c} - f^{\sf{min}}_{{\sf{x}},c}$. 
The BWPs are divided into sub-channels (SCs) due to various numerologies. For service $\sf{x}$, the RB channel spacing is given as $w_{\sf{x}} = 2^{\mu_{\sf{x}}} \times 180$~kHz. Let $W^{\sf{tot}} = f^{\sf{max}} - f^{\sf{min}}$ (Hz) denote the total BW. 
Then, the maximum SC number that can be assigned to the service
${\sf{x}}$ is $\bar{V}_{\sf{x}} = \lfloor W^{\sf{tot}} / w_{{\sf{x}}} \rfloor$.
We further denote $v_{\sf{x}}$ indices of the SCs in BWP ${\sf{x}}$, and $\mathcal{V}_{\sf{x}} \triangleq \{ 1,\dots, \bar{V}_{\sf{x}} \}$, i.e., $v_{\sf{x}} \in \mathcal{V}_{\sf{x}}$. 
Once, the BWP of service $\sf{x}$ is defined, the SCs of $\mathcal{V}_{\sf{x}}$ located in the corresponding frequency range will be activated for service-$\sf{x}$ transmission. 
Hereafter, the RBs in BWP $\sf{x}$ are identified based on both time and frequency indices, i.e., $[v_{\sf{x}}, n_{\sf{x}}]$.

Following 5G standard, guard bands (GBs) between BWPs are considered.
Here, the $\sf{D}-\sf{M}$ separating GB BW should be half of $w_{\sf{D}}$, i.e., $W_{\sf{G1}} = f^{\sf{min}}_{{\sf{M}},c} -  f^{\sf{max}}_{{\sf{D}},c} = w_{\sf{D}}/2$, and the $\sf{M}-\sf{S}$ separating GB BW is half of $w_{\sf{M}}$, $W_{\sf{G2}} = f^{\sf{min}}_{{\sf{S}},c} -  f^{\sf{max}}_{{\sf{M}},c} = w_{\sf{M}}/2$.
 

\begin{remark} \label{rmk01}
   Denote $v_{{\sf{D}},c}^{\sf{max}}$ and $v_{{\sf{M}},c}^{\sf{max}}$ the maximum indices of activated SCs of service $\sf{D}$ and $\sf{M}$ in cycle $c$. With respect to GB, we can obtain $f^{\sf{min}}_{{\sf{M}},c} = f^{\sf{min}} + (2v_{{\sf{D}},c}^{\sf{max}}+1) w_{\sf{M}} = f^{\sf{min}} + (4v_{{\sf{D}},c}^{\sf{max}}+2) w_{\sf{S}} $ and $f^{\sf{min}}_{{\sf{S}},c} = f^{\sf{min}} + (2v_{{\sf{M}},c}^{\sf{max}}+1) w_{\sf{S}} $.
\end{remark}

Let $\boldsymbol{b}_{\!c} \! \triangleq \! \{b^{\sf{x}}_{v_{\sf{x}},c} | \forall v_{\sf{x}} \! \in \! \mathcal{V}_{\sf{x}}, {\sf{x}} \! \in \! \mathcal{S} \}$ where $ b^{\sf{x}}_{v_{\sf{x}},c} = 1$ if SC $v_{\sf{x}}$ is activated in cycle $c$ and $ b^{\sf{x}}_{v_{\sf{x}},c} = 0$ otherwise. 
Remark~\ref{rmk01} yields,
\vspace{-6mm}
\begin{IEEEeqnarray}{ll}
    (C1): \; &b^{\sf{D}}_{v_{\sf{D}},c} \! + \! \scaleobj{0.8}{\sum\nolimits_{i=1}^{2v_{\sf{D}}+1} } \! b^{\sf{M}}_{i,c} \! + \! \scaleobj{0.8}{\sum\nolimits_{j=1}^{4v_{\sf{D}}+2 } } b^{\sf{s}}_{j,c} \leq 1, \; \forall (v_{\sf{D}}, c), \nonumber \vspace{-0.5mm} \\
    (C2): \; &b^{\sf{M}}_{v_{\sf{M}},c} +  \scaleobj{0.8}{\sum\nolimits_{j=1}^{2 v_{\sf{M}}+1} } b^{\sf{S}}_{j,c} \leq 1, \; \forall (v_{\sf{M}}, c). \nonumber
\end{IEEEeqnarray}
Additionally, the system BW constraint yields
\vspace{-2mm}
\begin{equation}
    (C3): \; \scaleobj{.8}{\sum\nolimits_{\sf{x} \in \mathcal{S}} \sum\nolimits_{\forall v_{\sf{x}} \in \mathcal{V}_{\sf{x}}}} b^{\sf{x}}_{v_{\sf{x}},c} w_{\sf{x}}  + W_{\sf{G1}} + W_{\sf{G2}} \leq W^{\sf{tot}}, \; \forall c. \nonumber
    \vspace{-3mm}
\end{equation}

\vspace{-2mm}

\subsection{Channel Model} \label{sec: channel}
\vspace{-1mm}
Let $\tilde{h}_{\ell,k}^{[v_{\sf{x}},n_{\sf{x}}]}$ and $\tilde{g}_{k}^{[v_{\sf{x}},n_{\sf{x}}]}$ be the channel coefficients of ${\sf{AP}}_{\ell}-{\sf{UE}}_{k}$ and ${\sf{LSat}}-{\sf{UE}}_{k}$ links over ${\sf{RB}}_{[v_{\sf{x}},n_{\sf{x}}]}$, respectively. 
Omitting $[v_{\sf{x}},n_{\sf{x}}]$, the channel coefficients are modeled as
\vspace{-1mm}
\bieq{ll}
    \tilde{h}_{\ell,k} = \sqrt{{\sf{PL}}_{\ell,k}} \left( \scaleobj{0.8}{\sqrt{\frac{\tilde{K}_{\ell,k}}{\tilde{K}_{\ell,k}+1}}} \tilde{h}_{\ell,k}^{{\sf{ls}}} + \scaleobj{0.8}{\sqrt{\frac{1}{\tilde{K}_{\ell,k}+1}}} \tilde{h}_{\ell,k}^{{\sf{nl}}}  \right),   \\
    \tilde{g}_{k} = \sqrt{{\sf{PL}}_{0,k}} \left( \scaleobj{0.8}{\sqrt{\frac{\tilde{K}_{0,k}}{\tilde{K}_{0,k}+1}}} \tilde{g}_{k}^{{\sf{ls}}} + \scaleobj{0.8}{\sqrt{\frac{1}{\tilde{K}_{0,k}+1}}} \tilde{g}_{k}^{{\sf{nl}}}  \right), \vspace{-2mm}
\eieq
based on the Rician model, where $\tilde{K}_{\ell,k}$, ${\sf{PL}}_{\ell,k}$, $\tilde{h}_{\ell,k}^{\sf{ls}}$, $\tilde{h}_{\ell,k}^{\sf{nl}}$, $\tilde{K}_{0,k}$, ${\sf{PL}}_{0,k}$, $\tilde{g}_{k}^{{\sf{ls}}}$, and $\tilde{g}_{k}^{{\sf{nl}}}$ are the K-factors, path-loss, the corresponding line-of-sight (LoS) and non-line-of-sight (NLoS) components of the TN and SatCom channels, respectively. 
In this work, we assume that TAPs can estimate perfectly CSIs for their served UEs in each frame thanks to the pilot sent in TN uplink SFs \cite{3gpp.38.211}. 
The channel gain in cycle $c$ and in frame $e$ are denoted by $\{\boldsymbol{h}_{\!c}, \boldsymbol{g}_{\!c} \}$ and $\{\boldsymbol{h}_{e}, \boldsymbol{g}_{e} \}$.
The relationship between the real-system and DT channels will be described in Section~\ref{sec: DT channel}.

\vspace{-2mm}
\subsection{Digital-Twin Model}\label{sec:DTmodel}
\vspace{-1mm}
The DT system is deployed in NMS to replicate the environment and network components. 
As illustrated in Fig.~\ref{fig:DT_model}, the DT model comprises a virtual system and computational functionalities. The virtual system consists of \textit{static components}, including environmental features, TAP positions, and terminal antenna characteristics, as well as \textit{dynamic components}, such as UE mobility, traffic demand, LSat movement, and CSIs. 
The DT model operation over each cycle is depicted in Fig.~\ref{fig:sysmodel}. 
Due to imperfect prediction, modeling limitations, and incomplete environmental information, mismatches inevitably arise. To enhance the accuracy, the DT system must be regularly calibrated by updating its components based on feedback from the actual system. 
Specifically, at the end of each cycle, the real mobility and traffic information are updated to the DT model. Using this information, computational functionalities are leveraged to emulate the real system's evolution in the next cycle and reflect it in the virtual one. The resulting predictions are used in resource management for the subsequent cycle.


\subsubsection{Geographical Environment} The DT system integrates 3D maps of the targeted areas. In cell $\ell$, the corresponding environment in DT system is represented by
\vspace{-2mm}
\begin{equation} \label{eq: cell DT}
    \widehat{\sf{env}}_{\ell} = \{{\sf{map}}_{\ell}\}, \quad 
    \widehat{\sf{AP}}_{\ell} = \{ {\sf{A}}_{\ell}^{\sf{ap}}, \boldsymbol{u}_{\ell}^{\sf{ap}} \}, \quad \forall \ell \in \mathcal{L}, 
    \vspace{-2mm}
\end{equation}
wherein ${\sf{map}}_{\ell}$ indicates the 3D map of the area in cell $\ell$, 
${\sf{A}}_{\ell}^{\sf{ap}}$ and $\boldsymbol{u}_{\ell}^{\sf{ap}}$ are the antenna properties and position of ${\sf{AP}}_{\ell}$.

\subsubsection{UE Information} Virtual $\sf{UE}_{k}$ in frame $e$ is modeled as 
\vspace{-2mm}
\begin{IEEEeqnarray}{ll} \label{eq: UE DT}
    \widehat{\sf{UE}}_{k,[e]}^{\sf{D}} = \{{\sf{A}}_{k}^{\sf{ue}}, \hat{\boldsymbol{u}}_{k,[e]}^{{\sf{ue}}} , \hat{\lambda}_{k,[s]}^{\sf{SF,D}}\vert_{\forall s \in \mathcal{T}_{\sf{SF}}^{{\sf{TF}},e}} \}, \; \forall k \in \mathcal{K}_{\sf{D}}, \subnum \\
    \widehat{\sf{UE}}_{k,[e]}^{\sf{x}} = \{{\sf{A}}_{k}^{\sf{ue}}, \hat{\boldsymbol{u}}_{k,[e]}^{{\sf{ue}}} , \hat{\lambda}_{k,[e]}^{\sf{TF,x}} \}, \; \forall k \in \mathcal{K}_{\sf{x}}, {\sf{x} \in \{\sf{M,S}\}}, \subnum \vspace{-1mm}
\end{IEEEeqnarray}
wherein $\mathcal{T}_{\sf{SF}}^{{\sf{TF}},e}$ is SF set in frame $e$,
${\sf{A}}_{k}^{\sf{ue}}$, $\hat{\boldsymbol{u}}_{k,[e]}^{{\sf{ue}}}$, $\hat{\lambda}_{k,[e]}^{\sf{TF,x}}$, and $\hat{\lambda}_{k,[s]}^{\sf{SF,D}}$ are the virtual antenna properties, position and traffic rate of ${\sf{UE}}_{k}$ in frame $e$ and SF $s$, respectively. \textit{The UE traffic rates will be described in Section~\ref{sec:traffic}}.
This UE information in the upcoming cycle is obtained based on prediction techniques and updated information from the real environment. 

\subsubsection{LSat Model} 
The LSat DT at frame $e$ is modeled as
\vspace{-1mm}
\begin{equation} \label{eq: Sat DT}
    \widehat{\sf{Sat}}_{[e]} = \{ {\sf{A}}^{\sf{sat}}, \hat{\boldsymbol{u}}^{\sf{sat}}_{[e]}, {\sf{TLE}}_{[e]} \}, \vspace{-1mm}
\end{equation} 
where ${\sf{A}}^{\sf{sat}}$, $ \hat{\boldsymbol{u}}^{\sf{sat}}_{[e]}$, and ${\sf{TLE}}_{[e]}$ are the antenna properties, predicted position, and two-line-element (TLE) data of LSat in frame $e$, respectively. 
Due to the orbit stability and the frequent update framework, we assume that the LSat position is predicted with a negligible error, that is, $\hat{\boldsymbol{u}}^{\sf{sat}}_{[e]} \approx \boldsymbol{u}^{\sf{sat}}_{[e]}$.

\subsubsection{Channel Prediction} \label{sec: DT channel}
Using the RayT tool as in Fig.~\ref{fig:DT_model}, the DT channel set is identified as our previous work \cite{Hung_TCOM25}
\vspace{-2mm}
\begin{equation} \label{eq: DT channel predict}
    {\sf{chan}}_{[e]} = {\sf{RT}}(\widehat{\sf{env}}_{\ell}, \widehat{\sf{AP}}_{\ell},\widehat{\sf{Sat}}_{[e]}, \widehat{\sf{UE}}_{k,[e]}), 
    \vspace{-1mm}
\end{equation}
where $\scaleobj{0.8}{\widehat{\sf{UE}}}_{k,[e]}$ and $\scaleobj{0.9}{\sf{chan}}_{[e]}$ are expressed as $\scaleobj{0.8}{\widehat{\sf{UE}}}_{k,[e]}=\{\scaleobj{0.8}{\widehat{\sf{UE}}_{k,[e]}^{\sf{x}}}\}_{\forall \sf{x}}$, $\scaleobj{0.9}{\sf{chan}}_{[e]} = \{ \tilde{h}_{\ell,k}^{{\sf{ls}},[v_{\sf{x}},n_{\sf{x}}]}, \bar{h}_{\ell,k}^{{\sf{nl}},[v_{\sf{x}},n_{\sf{x}}]},$ $ \scaleobj{0.8}{\sf{PL}}_{\ell,k}^{[v_{\sf{x}},n_{\sf{x}}]}, \tilde{K}_{\ell,k}^{[v_{\sf{x}},n_{\sf{x}}]}, \tilde{g}_{k}^{{\sf{ls}},[v_{\sf{x}},n_{\sf{x}}]}, \bar{g}_{k}^{{\sf{nl}},[v_{\sf{x}},n_{\sf{x}}]}, \scaleobj{0.8}{\sf{PL}}_{0,k}^{[v_{\sf{x}},n_{\sf{x}}]}, \tilde{K}_{0,k}^{[v_{\sf{x}},n_{\sf{x}}]} \}$, with $\forall (\ell ,k), \forall n_{\sf{x}} \in \mathcal{T}^{{\sf{TF}},e}_{\sf{x}}$.
In virtual environment, $\bar{h}_{\ell,k}^{{\sf{nl}},[v_{\sf{x}},n_{\sf{x}}]}$ and $\bar{g}_{k}^{{\sf{nl}},[v_{\sf{x}},n_{\sf{x}}]}$ are introduced to represent the NLoS components of DT channels. Those are different with $\tilde{h}_{\ell,k}^{{\sf{nl}},[v_{\sf{x}},n_{\sf{x}}]}$ and $\tilde{g}_{k}^{{\sf{nl}},[v_{\sf{x}},n_{\sf{x}}]}$ of real system channels, primarily due to the absence of information in the 3D map. The relationship between real and virtual NLoS components is modeled as
\vspace{-1mm}
\beq \label{eq: NLOS model}
    \tilde{h}_{\ell,k}^{{\sf{nl}}} \! = \! \sqrt{\xi} \bar{h}_{\ell,k}^{{\sf{nl}}} \! + \! \sqrt{(1 \! - \! \xi)} \delta_{\ell,k} \text{ and } \tilde{g}_{k}^{{\sf{nl}}} \! = \! \sqrt{\xi} \bar{g}_{k}^{{\sf{nl}}} \! + \! \sqrt{(1 \! - \! \xi)} \delta_{0,k}, \vspace{-1mm}
\eeq
where $\delta_{\ell,k}$ and $\delta_{0,k}$ indicate errors caused by the absence of map information, which are assumed as complex normal random variables. $\xi \in (0,1)$ is the correlation factor. 
In addition, it is worth noting that the position errors also contribute to mismatches between actual and DT channels, as they affect both LoS and NLoS components.

\begin{figure*}
    \centering
    \includegraphics[width=0.83\linewidth]{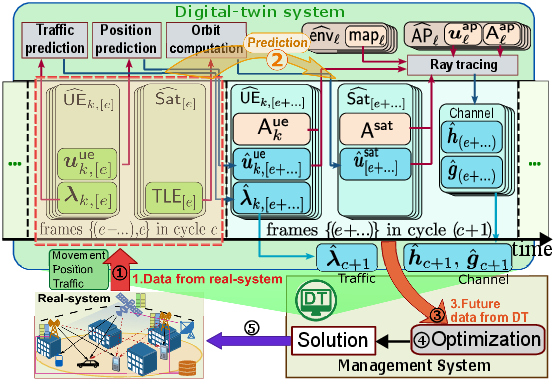}
    \captionsetup{font=footnotesize}
    \caption{Summary of the DT-system.}
    \label{fig:DT_model}
    \vspace{-4mm}
\end{figure*}

\vspace{-3mm}

\section{DT-based Optimization Problem Formulation}
\vspace{-2mm}
\subsection{User Association}
\vspace{-1mm}
Let $\boldsymbol{\alpha} = [\alpha_{\ell,k}^{[v_{\sf{x}},n_{\sf{x}}]}]$ be the binary association variable for $ {\sf{x}\in\{\sf{D,M}\}}$ where 
$\alpha_{\ell,k}^{[v_{\sf{x}},n_{\sf{x}}]} = 1$ if $\sf{AP}_{\ell}$ served $\sf{UE}_{k}$ over $\sf{RB}_{[v_{\sf{x}},n_{\sf{x}}]}$ for DL and $\alpha_{\ell,k}^{[v_{\sf{x}},n_{\sf{x}}]} = 0$ otherwise. Note that the association with ${\sf{UE}}^{\sf{x}}$ over SC $v_{\sf{x}}$ in BWP $\sf{x}$ is allowed only if this SC is activated, which is ensured as
\vspace{-1mm}
\begin{equation}
    (C4)\!: \alpha_{\ell,k}^{[v_{\sf{x}},n_{\sf{x}}]} \! \leq \! b^{\sf{x}}_{v_{\sf{x}},c}, \forall (c,\ell ,k,[v_{\sf{x}},n_{\sf{x}}]), {\sf{x}} \!\! \in \! \{\sf{D,M}\}. \nonumber \vspace{-1mm}
\end{equation}
In addition, one assumes that each RB of each AP can be assigned to at most one UE, which yields
\vspace{-1mm}
\begin{equation}
    (C5): \; \scaleobj{.8}{\sum\nolimits_{\forall k}} \alpha_{\ell,k}^{[v_{\sf{x}},n_{\sf{x}}]} \leq 1, \; \forall (c, \ell, [v_{\sf{x}},n_{\sf{x}}]), {\sf{x}} \in \{\sf{D,M}\}. \nonumber \vspace{-1mm}
\end{equation}
While each $\sf{UE}^{\sf{D}}$ can be served by multiple APs at each TS via different RBs, which is ensured as
\vspace{-1mm}
\begin{equation}
    (C6): \; \scaleobj{.8}{\sum\nolimits_{\forall \ell}} \alpha_{\ell,k}^{[v_{\sf{D}},n_{\sf{D}}]} \leq 1, \; \forall (c,k, [v_{\sf{D}}, n_{\sf{D}}]). \nonumber \vspace{-1mm}
\end{equation}

For LSat-UE association, we introduce a variable $\boldsymbol{\beta}= [\beta_{k}^{[v_{\sf{x}},n_{\sf{x}}]}]$ for $\sf{x} \in \{\sf{M,S}\}$ where
$\beta_{k}^{[v_{\sf{x}},n_{\sf{x}}]} = 1$ if $\sf{UE}_{k}$ is served by $\sf{LSat}$ over ${\sf{RB}}_{[v_{\sf{x}} ,n_{\sf{x}}]}$, $\beta_{k}^{[v_{\sf{x}},n_{\sf{x}}]} = 0$ otherwise.
Similar to the AP-UE association, we have
\vspace{-1mm}
\beqn
    &&(C7): \; \beta_{k}^{[v_{\sf{x}},n_{\sf{x}}]} \leq b^{\sf{x}}_{v_{\sf{x}},c}, \forall (c,k, [v_{\sf{x}}, n_{\sf{x}}], \nonumber \\
    &&(C8): \; \scaleobj{.8}{\sum\nolimits_{\forall k}} \beta_{k}^{[v_{\sf{x}},n_{\sf{x}}]} \leq 1, \; \forall (c, [v_{\sf{x}},n_{\sf{x}}]), {\sf{x}} \in \{\sf{M,S}\}, \nonumber \vspace{-1mm}
\eeqn
Additionally, we assume that UEs using $\sf{M}$ services can be served by both TAPs and LSat at the same time via different RBs in BWP $\sf{M}$, which yields the constraint
\vspace{-1mm}
\begin{equation}
    (C9): \; \scaleobj{.8}{\sum\nolimits_{\forall \ell}} \alpha_{\ell,k}^{[v_{\sf{M}},n_{\sf{M}}]} + \beta_{k}^{[v_{\sf{M}},n_{\sf{M}}]} \leq 1, \; \forall (c, k, [v_{\sf{M}}, n_{\sf{M}}]). \nonumber \vspace{-2mm}
\end{equation}

\subsection{Service-$\sf{D}$ Transmission}
\vspace{-1mm}
Assuming that ${\sf{UE}}_{k}^{\sf{D}}$ is served by ${\sf{AP}}_{\ell}$ over ${\sf{RB}}_{[v_{\sf{D}},n_{\sf{D}}]}$, the received signal $y_{k}^{[v_{\sf{D}},n_{\sf{D}}]}$ and its components are expressed as
\vspace{-2mm}
\begin{align}\label{eq: yk_d}
    y_{k}^{{\sf{D}},[v_{\sf{D}},n_{\sf{D}}]} = \tilde{y}_{k}^{{\sf{TN}},[v_{\sf{D}},n_{\sf{D}}]} +  \varsigma_{k}^{[v_{\sf{D}},n_{\sf{D}}]} ,
    \vspace{-2mm}
\end{align}
where $\tilde{y}_{k}^{{\sf{TN}},[v_{\sf{x}},n_{\sf{x}}]} \!\!=\! {\sum_{\forall (i ,j)} \!\!\! \scaleobj{0.9}{ \sqrt{ \alpha_{i,j}^{[v_{\sf{x}},n_{\sf{x}}]} p_{i,j}^{[v_{\sf{x}},n_{\sf{x}}]} } } \tilde{h}_{i,k}^{[v_{\sf{x}},n_{\sf{x}}]} x_{i,j}^{[v_{\sf{x}},n_{\sf{x}}]}}$; $p_{\ell,k}^{[v_{\sf{x}},n_{\sf{x}}]}$ and $x_{\ell,k}^{[v_{\sf{x}},n_{\sf{x}}]}$ are the transmit power and transmission symbol from ${\sf{AP}}_{\ell}$ to ${\sf{UE}}_{k}$ over ${\sf{RB}}_{[v_{\sf{x}},n_{\sf{x}}]}$. $\varsigma_{k}^{[v_{\sf{x}},n_{\sf{x}}]} \sim \mathcal{CN}(0,\sigma_{\sf{x},k}^{2})$ is the AGWN at ${\sf{UE}}_{k}$.
Hence, the SINR of ${\sf{UE}}_{k}$ is expressed as
\vspace{-1mm}
\begin{equation}\label{eq: SINR BWPds}
    \gamma_{\ell,k}^{{\sf{D}},[v_{\sf{D}},n_{\sf{D}}]}(\! \boldsymbol{p}_{\!c},\! \boldsymbol{\alpha}_{\!c}\!) = \scaleobj{.8}{\frac{\alpha_{\ell,k}^{[v_{\sf{D}},n_{\sf{D}}]} p_{\ell,k}^{[v_{\sf{D}},n_{\sf{D}}]} h_{\ell,k}^{[v_{\sf{D}},n_{\sf{D}}]}}{ \Psi_{\ell,k}^{[v_{\sf{D}},n_{\sf{D}}]}(\! \boldsymbol{p}_{\!c},\! \boldsymbol{\alpha}_{\!c}\!) + \sigma_{{\sf{D}},k}^2}}, \vspace{-1mm}
\end{equation}
where $\Psi_{\ell,k}^{[v_{\sf{x}},n_{\sf{x}}]}(\! \boldsymbol{p}_{\!c},\! \boldsymbol{\alpha}_{\!c}\!)$ is the inter-cell-interference (ICI) power,
\vspace{-1mm}
\begin{equation} \label{eq:ICI}
    \hspace{-3mm} \Psi_{\ell,k}^{[v_{\sf{x}},n_{\sf{x}}]}(\! \boldsymbol{p}_{\!c},\! \boldsymbol{\alpha}_{\!c}\!) = \scaleobj{.8}{\sum\nolimits_{\substack{ \forall i \neq \ell,  \forall j}} } \alpha_{i,j}^{[v_{\sf{x}},n_{\sf{x}}]} p_{i,j}^{[v_{\sf{x}},n_{\sf{x}}]} h_{i,k}^{[v_{\sf{x}},n_{\sf{x}}]}, \vspace{-1mm}
\end{equation}
with $(\boldsymbol{p}_{\!c},\boldsymbol{\alpha}_{\!c} ) \! \triangleq \! \{ (p_{\ell,k}^{[v_{\sf{x}},n_{\sf{x}}]},\alpha_{\ell,k}^{[v_{\sf{x}},n_{\sf{x}}]}) \vert_{\forall ( \ell, k, [v_{\sf{x}},n_{\sf{x}}]), {\sf{x}} \in \{\sf{D,M}\} } \}$ and $h_{\ell,k}^{[v_{\sf{x}},n_{\sf{x}}]} = |\tilde{h}_{\ell,k}^{[v_{\sf{x}},n_{\sf{x}}]}|^2$.
Due to the delay-sensitive requirement, the short packet framework is used to model the transmission of $\sf{D}$ service. Hence, the aggregated achievable rate of ${\sf{UE}}_{k}^{\sf{D}}$ served by ${\sf{AP}}_{\ell}$ at SF $s$ can be expressed as \cite{polyanskiy_FiniteCode}
\begin{IEEEeqnarray}{ll}
    R_{\ell,k,[s]}^{{\sf{SF,D}}}(\boldsymbol{p}_{\!c},\boldsymbol{\alpha}_{\!c}) = & w_{\sf{D}} \Big[ \scaleobj{0.8}{\sum\limits_{\forall (v_{\sf{D}},n_{\sf{D}}) } } \!\!\!\! \log_2{(1 + \gamma_{\ell,k}^{{\sf{D}},[v_{\sf{D}},n_{\sf{D}}]}(\! \boldsymbol{p}_{\!c},\! \boldsymbol{\alpha}_{\!c}\!)})   \nonumber \\
    & \;\; - \frac{1}{\ln(2)} \! \times \! {\frac{\alpha_{\ell,k}^{[v_{\sf{D}},n_{\sf{D}}]} \!\! \scaleobj{0.8}{\sqrt{V_{\ell,k}^{{\sf{D}},[v_{\sf{D}},n_{\sf{D}}]}} } Q^{-1}(P_{\epsilon}) }{{\sqrt{ { \sum\nolimits_{\forall (v_{\sf{D}},n_{\sf{D}})} } \alpha_{\ell,k}^{[v_{\sf{D}},n_{\sf{D}}]} T_{\sf{d}} w_{\sf{d}}}} } } \Big], \quad \quad
\end{IEEEeqnarray}
where $V_{\ell,k}^{[v_{\sf{D}},n_{\sf{D}}]} = 1 - (1 + \gamma_{\ell,k}^{{\sf{D}},[v_{\sf{D}},n_{\sf{D}}]}(\! \boldsymbol{p}_{\!c},\! \boldsymbol{\alpha}_{\!c}\!))^{-2}$, $Q^{-1}(\cdot)$ and $P_{\epsilon}$ are the channel dispersion, the inverse of the Q-function, and the error probability. One can see that the channel dispersion can be approximated as $V_{\ell,k}^{[v_{\sf{D}},n_{\sf{D}}]} \approx 1$ for a sufficiently high $\gamma_{\ell,k}^{{\sf{D}},[v_{\sf{D}},n_{\sf{D}}]}(\! \boldsymbol{p}_{\!c},\! \boldsymbol{\alpha}_{\!c}\!) \geq \gamma_{0}^{\sf{D}}$ with $\gamma_{0}^{\sf{D}} \geq 5$~dB \cite{schiessl_Fading_FiniteCode}. Regarding this approximation, we consider the following constraint
\vspace{-1mm}
\begin{equation}
    (C10):  \gamma_{\ell,k}^{{\sf{D}},[v_{\sf{D}},n_{\sf{D}}]}(\! \boldsymbol{p}_{\!c},\! \boldsymbol{\alpha}_{\!c}\!) \geq \alpha_{\ell,k}^{[v_{\sf{D}},n_{\sf{D}}]} \gamma_{0}^{\sf{D}}, \; \forall (\ell ,k,v_{\sf{D}},n_{\sf{D}}) . \nonumber \vspace{-1mm}
\end{equation}
Subsequently, $R_{\ell,k,[s]}^{{\sf{SF,D}}}(\! \boldsymbol{p}_{\!c},\! \boldsymbol{\alpha}_{\!c}\!)$ can be rewritten as 
\vspace{-1mm}
\begin{equation}\label{eq: rate DS}
    R_{\ell,k,[s]}^{{\sf{SF,D}}}\!\!(\!\boldsymbol{p}_{\!c},\boldsymbol{\alpha}_{\!c} \!) \!=\! 
    w_{\sf{D}} \!\!\!\!\! \scaleobj{0.8}{\sum_{ \forall (v_{\sf{D}},n_{\sf{D}}) } } \!\!\!\! \log_2{\!(\! 1 \!+\! \gamma_{\ell,k}^{{\sf{D}},[v_{\sf{D}},n_{\sf{D}}]}(\! \boldsymbol{p}_{\!c},\! \boldsymbol{\alpha}_{\!c}\!) } ) 
    \!-\! \chi_{\sf{D}} \!\! \scaleobj{0.8}{\sqrt{ \hspace{-4mm} \scaleobj{0.8}{\sum_{\hspace{5mm} \forall (v_{\sf{D}},n_{\sf{D}}) }} \!\!\!\!\!\!\! \alpha_{\ell,k}^{[v_{\sf{D}},n_{\sf{D}}]} } }, \vspace{-1mm}
\end{equation}
where $\chi_{\sf{D}} \!=\! \scaleobj{0.8}{\sqrt{V w_{\sf{D}} / T_{\sf{D}} } } Q^{-1}(P_{\epsilon}) / \ln(2)$ and $V \approx V_{\ell,k}^{{\sf{D}},[v_{\sf{D}}, n_{\sf{D}}]} \approx 1$.


\subsection{Service-$\sf{M}$ Transmission}
In this BWP, UEs can be served by APs and the LSat. 
The received signal $y_{k}^{{\sf{M}},[v_{\sf{M}},n_{\sf{M}}]}$ at ${\sf{UE}}_{k}^{\sf{M}}$ over ${\sf{RB}}_{[v_{\sf{M}},n_{\sf{M}}]}$ and its components is expressed as
\vspace{-2mm}
\begin{IEEEeqnarray}{ll}
    y_{k}^{{\sf{M}},[v_{\sf{M}},n_{\sf{M}}]} = \tilde{y}_{k}^{{\sf{TN}}, [v_{\sf{M}},n_{\sf{M}}]} + \tilde{y}_{k}^{{\sf{Sat}},{\sf{M}},[v_{\sf{M}},n_{\sf{M}}]} + \varsigma_{k}^{[v_{\sf{M}},n_{\sf{M}}]}, \\
    \tilde{y}_{k}^{{\sf{Sat}},{\sf{M}},[v_{\sf{M}},n_{\sf{M}}]} \triangleq \scaleobj{0.8}{\sum\nolimits_{\forall j}} \scaleobj{0.8}{ \sqrt{\beta_{j}^{[v_{\sf{M}},n_{\sf{M}}]} p_{0,j}^{[v_{\sf{M}},n_{\sf{M}}]} } } \tilde{g}_{k}^{[v_{\sf{M}},n_{\sf{M}}]} x_{j}^{[v_{\sf{M}},n_{\sf{M}}]},
    \vspace{-1mm}
\end{IEEEeqnarray}
where $p_{0,k}^{[v_{\sf{x}},n_{\sf{x}}]}$ is the transmit power from the LSat to ${\sf{UE}}_{k}$ over ${\sf{RB}}_{[v_{\sf{x}},n_{\sf{x}}]}$, and $g_{k}^{[v_{\sf{x}},n_{\sf{x}}]} = |\tilde{g}_{k}^{[v_{\sf{x}},n_{\sf{x}}]}|^2$. 
For brevity, let's define $(\boldsymbol{p}_{\!0,c},\boldsymbol{\beta}_{\!c}) \triangleq \{ ( p_{0,k}^{[v_{\sf{x}},n_{\sf{x}}]},\beta_{k}^{[v_{\sf{x}},n_{\sf{x}}]} ) \vert_{ \forall (k,v_{\sf{x}},n_{\sf{x}}), {\sf{x}} \in\{\sf{M,S}\} } \}$ and $\boldsymbol{P}_{\!c} \triangleq \{ \boldsymbol{p}_{\!c},\boldsymbol{p}_{\!0,c} \}$. Hereafter, argument $\boldsymbol{p}_{\!c}$ in defined functions is replaced appropriately by $\boldsymbol{P}_{\!c}$.

\subsubsection{TN-served UEs}
Assuming that ${\sf{UE}}_{k}^{\sf{M}}$ is served by ${\sf{AP}}_{\ell}$, the corresponding SINR is expressed as
\begin{equation}\label{eq: SINR BWPms}
\hspace{-3mm}    \gamma_{\ell,k}^{{\sf{M}},[v_{\sf{M}},n_{\sf{M}}]}\!(\! \boldsymbol{P}_{\!c},\! \boldsymbol{\alpha}_{\!c},\! \boldsymbol{\beta}_{\!c} \!) 
    = \scaleobj{0.8}{\frac{\alpha_{\ell,k}^{[v_{\sf{M}},n_{\sf{M}}]} p_{\ell,k}^{[v_{\sf{M}},n_{\sf{M}}]} h_{\ell,k}^{[v_{\sf{M}},n_{\sf{M}}]}}{ \Psi_{\ell,k}^{[v_{\sf{M}},n_{\sf{M}}]}\!(\! \boldsymbol{P}_{\!c},\! \boldsymbol{\alpha}_{\!c} \!) + \Theta_{k}^{{\sf{TN}},[v_{\sf{M}},n_{\sf{M}}]}\!(\! \boldsymbol{P}_{\!c},\! \boldsymbol{\beta}_{\!c} \!) + \sigma_{{\sf{M}},k}^2} }, \!\!\!\!\!
\end{equation}
where $\Psi_{\ell,k}^{[v_{\sf{M}},n_{\sf{M}}]} $ is defined in \eqref{eq:ICI} and  $\Theta_{k}^{{\sf{TN}},[v_{\sf{M}},n_{\sf{M}}]}\!(\! \boldsymbol{P}_{\!c},\! \boldsymbol{\beta}_{\!c} \!)$ is the inter-system-interference (ISyI) caused by the LSat to ${\sf{UE}}_{k}^{\sf{M}}$ served by TN over ${\sf{RB}}_{[v_{\sf{M}},n_{\sf{M}}]}$ which is given as
\begin{equation}
    \Theta_{k}^{{\sf{TN}},[v_{\sf{M}},n_{\sf{M}}]}\!(\! \boldsymbol{P}_{\!c},\! \boldsymbol{\beta}_{\!c} \!) = \scaleobj{0.8}{\sum\nolimits_{\forall j} } \beta_{j}^{[v_{\sf{M}},n_{\sf{M}}]} p_{0,j}^{[v_{\sf{M}},n_{\sf{M}}]}  g_{k}^{[v_{\sf{M}},n_{\sf{M}}]},
\end{equation}
The corresponding achievable rate over ${\sf{RB}}_{[v_{\sf{M}},n_{\sf{M}}]}$ and the aggregated rate at RB time $n_{\sf{M}}$ are expressed as
\vspace{-1mm}
\begin{IEEEeqnarray}{ll}\label{eq: rate ms ap}
    \!\!\!\!\! R_{\ell,k}^{{\sf{M}},[v_{\sf{M}},n_{\sf{M}}]}\!(\! \boldsymbol{P}_{\!c},\! \boldsymbol{\alpha}_{\!c},\! \boldsymbol{\beta}_{\!c} \!) 
    \! = \! w_{\sf{M}} \log_2 \!\! \left( \! 1 \! + \! \gamma_{\ell,k}^{{\sf{M}},[v_{\sf{M}},n_{\sf{M}}]}\!(\! \boldsymbol{P}_{\!c},\! \boldsymbol{\alpha}_{\!c},\! \boldsymbol{\beta}_{\!c} \!) \! \right)\!, \subnum \\
    \!\!\!\!\! R_{\ell,k}^{{\sf{M}},[n_{\sf{M}}]}\!(\! \boldsymbol{P}_{\!c},\! \boldsymbol{\alpha}_{\!c},\! \boldsymbol{\beta}_{\!c} \!) = \scaleobj{0.8}{\sum\nolimits_{\forall v_{\sf{M}} \in \mathcal{V}_{\sf{M}}} } R_{\ell,k}^{{\sf{M}},[v_{\sf{M}},n_{\sf{M}}]}\!(\! \boldsymbol{P}_{\!c},\! \boldsymbol{\alpha}_{\!c},\! \boldsymbol{\beta}_{\!c} \!). \subnum \vspace{-1mm}
\end{IEEEeqnarray}
\subsubsection{UE served by SatNet} Assuming that ${\sf{UE}}_{k}^{\sf{M}}$ is served by the LSat over ${\sf{RB}}_{[v_{\sf{M}},n_{\sf{M}}]}$, the corresponding SINR is given as
\vspace{-1mm}
\begin{equation}
    \gamma_{0,k}^{{\sf{M}},[v_{\sf{M}},n_{\sf{M}}]}\!(\! \boldsymbol{P}_{\!c},\! \boldsymbol{\alpha}_{\!c},\! \boldsymbol{\beta}_{\!c} \!) = \scaleobj{0.8}{ \frac{\beta_{k}^{[v_{\sf{M}},n_{\sf{M}}]} p_{0,k}^{[v_{\sf{M}},n_{\sf{M}}]} g_{k}^{[v_{\sf{M}},n_{\sf{M}}]} }{ \Theta_{k}^{{\sf{Sat}},[v_{\sf{M}},n_{\sf{M}}]}\!(\! \boldsymbol{P}_{\!c},\! \boldsymbol{\alpha}_{\!c} \!) + \sigma_{{\sf{M}},k}^2} }, \vspace{-1mm}
\end{equation}
where $\Theta_{k}^{{\sf{Sat}},[v_{\sf{M}},n_{\sf{M}}]}\!(\! \boldsymbol{P}_{\!c},\! \boldsymbol{\alpha}_{\!c} \!)$ is the ISyI caused by TAPs to ${\sf{UE}}_{k}^{\sf{M}}$ served by the LSat which is defined as
\vspace{-1mm}
\begin{equation}
    \Theta_{k}^{{\sf{Sat}},[v_{\sf{M}},n_{\sf{M}}]}\!(\! \boldsymbol{P}_{\!c},\! \boldsymbol{\alpha}_{\!c} \!) = \scaleobj{0.8}{ \sum\nolimits_{\forall (i,j)}  } \alpha_{i,j}^{[v_{\sf{M}},n_{\sf{M}}]} p_{i,j}^{[v_{\sf{M}},n_{\sf{M}}]} h_{i,k}^{[v_{\sf{M}},n_{\sf{M}}]}.
    \vspace{-1mm}
\end{equation}
Therefore, the achievable rate over ${\sf{RB}}_{[v_{\sf{M}},n_{\sf{M}}]}$ and the aggregated rate at RB time $n_{\sf{M}}$ of ${\sf{UE}}_{k}^{\sf{M}}$ served by the LSat can be respectively expressed as
\vspace{-1mm}
\begin{IEEEeqnarray}{ll}\label{eq: rate ms sat}
   \hspace{-5mm} R_{0,k}^{{\sf{M}},[v_{\sf{M}},n_{\sf{M}}]}\!(\! \boldsymbol{P}_{\!c},\! \boldsymbol{\alpha}_{\!c},\! \boldsymbol{\beta}_{\!c} \!) \! = \! w_{\sf{M}}\!  \log_2 \!\! \left( \! 1 \! + \! \gamma_{0,k}^{{\sf{M}},[v_{\sf{M}},n_{\sf{M}}]}\!(\! \boldsymbol{P}_{\!c},\! \boldsymbol{\alpha}_{\!c},\! \boldsymbol{\beta}_{\!c} \!) \! \right)\!, \subnum \\
   \hspace{-5mm} R_{0,k}^{{\sf{M}},[n_{\sf{M}}]}\!(\! \boldsymbol{P}_{\!c},\! \boldsymbol{\alpha}_{\!c},\! \boldsymbol{\beta}_{\!c} \!) = \scaleobj{0.8}{ \sum\nolimits_{\forall v_{\sf{M}} \in \mathcal{V}_{\sf{M}} } } R_{0,k}^{{\sf{M}},[v_{\sf{M}},n_{\sf{M}}]}\!(\! \boldsymbol{P}_{\!c},\! \boldsymbol{\alpha}_{\!c},\! \boldsymbol{\beta}_{\!c} \!). \subnum \vspace{-1mm}
\end{IEEEeqnarray} 



\subsection{Service-$\sf{S}$ Transmission}
\vspace{-1mm}
In this BWP, if ${\sf{UE}}_{k}^{\sf{S}}$ is served by LSat over ${\sf{RB}}_{[v_{\sf{S}},n_{\sf{S}}]}$, the corresponding SNR is given by
\vspace{-1mm}
\begin{equation} \label{eq: SINR ss}
    \gamma_{0,k}^{{\sf{S}},[v_{\sf{S}},n_{\sf{S}}]}\!(\! \boldsymbol{P}_{\!c},\! \boldsymbol{\beta}_{\!c} \!) = (\beta_{k}^{[v_{\sf{S}},n_{\sf{S}}]} p_{0,k}^{[v_{\sf{S}},n_{\sf{S}}]} g_{k}^{[v_{\sf{S}},n_{\sf{S}}] })/ \sigma_{{\sf{S}},k}^2, \vspace{-1mm}
\end{equation}
Therefore, the achievable rate over ${\sf{RB}}_{[v_{\sf{S}},n_{\sf{S}}]}$ and
the aggregated rate at RB time $n_{\sf{S}}$ of ${\sf{UE}}_{k}^{\sf{S}}$ can be expressed as
\vspace{-2mm}
\begin{IEEEeqnarray}{ll}\label{eq: rate ss sat}
    \hspace{-2mm} R_{0,k}^{{\sf{S}},[v_{\sf{S}},n_{\sf{S}}]}\!(\! \boldsymbol{P}_{\!c},\! \boldsymbol{\beta}_{\!c} \!) &= w_{\sf{S}} \log_2\left( 1 + \gamma_{0,k}^{{\sf{S}},[v_{\sf{S}},n_{\sf{S}}]}\!(\! \boldsymbol{P}_{\!c},\! \boldsymbol{\beta}_{\!c} \!) \right), \subnum \\
    R_{0,k}^{{\sf{S}},[n_{\sf{S}}]}\!(\! \boldsymbol{P}_{\!c},\! \boldsymbol{\beta}_{\!c} \!) &=  \scaleobj{0.8}{\sum\nolimits_{\forall v_{\sf{S}} \in \mathcal{V}_{\sf{S}}} } R_{0,k}^{{\sf{S}},[v_{\sf{S}},n_{\sf{S}}]}\!(\! \boldsymbol{P}_{\!c},\! \boldsymbol{\beta}_{\!c} \!). \subnum \vspace{-1mm}
\end{IEEEeqnarray} 
Regarding the TAP/LSat power budget, one must satisfy
\vspace{-1mm}
\begin{equation}
\begin{array}{ll}
    (C11):\; \sum_{\sf{x} \in \{D,M\}} \sum_{\forall k} \sum_{\forall v_{\sf{x}} \in \mathcal{V}_{\sf{x}}} p_{\ell,k}^{[v_{\sf{x}},n_{\sf{x}}]} \leq p_{{\sf{AP}},\ell}^{\sf{max}}, \forall (\ell, t), \nonumber \\
    (C12): \; \sum_{\sf{x} \in \{M,S\}} \sum_{\forall k} \sum_{\forall v_{\sf{x}} \in \mathcal{V}_{\sf{x}}} p_{0,k}^{[v_{\sf{x}},n_{\sf{x}}]} \leq p_{\sf{LSat}}^{\sf{max}}, \forall t. \nonumber
    \vspace{-1mm}
\end{array}
\end{equation}


\subsection{Traffic and Queuing Model} \label{sec:traffic}
\vspace{-1mm}
Regarding latency tolerance, different time scales are applied to serve UEs associated with different services. 
Specifically, the arrival data of $\sf{M}$ and $\sf{S}$ services are buffered and scheduled for transmission in the following frame, whereas the data stream of $\sf{D}$ services is processed at the SF time scale to meet stringent latency requirements.
Assuming $K$ traffic flows arriving the CN due to $K$ UEs. Denote $\boldsymbol{\lambda}_{[e]}^{{\sf{TF,M}}}=\{\lambda_{k,[e]}^{{\sf{TF,M}}} \vert_{\forall k\in \mathcal{K}_{\sf{M}}} \}$, $\boldsymbol{\lambda}_{[e]}^{{\sf{TF,S}}}=\{\lambda_{k,[e]}^{{\sf{TF,S}}} \vert_{\forall k\in \mathcal{K}_{\sf{S}}} \}$, and $\boldsymbol{\lambda}_{[s]}^{{\sf{SF,D}}}=\{ \lambda_{k,[s]}^{{\sf{SF,D}}} \vert_{\forall k\in \mathcal{K}_{\sf{D}}} \}$, where $\lambda_{k,[e]}^{{\sf{TF,M}}}$, $\lambda_{i,[e]}^{{\sf{TF,S}}}$, and $\lambda_{j,[s]}^{{\sf{SF,D}}}$ are the data arrival rates of ${\sf{UE}}_{k}^{\sf{M}}$, ${\sf{UE}}_{i}^{\sf{S}}$ at frame $e$, and of ${\sf{UE}}_{j}^{\sf{D}}$ at SF $s$, respectively. 
Additionally, the traffic rates in cycle $c$ and frame $e$ are denoted by $\boldsymbol{\lambda}_{\!c} $ and $\boldsymbol{\lambda}_{e} $.

\begin{figure}[!t]
    \centering
    \includegraphics[width=1\columnwidth]{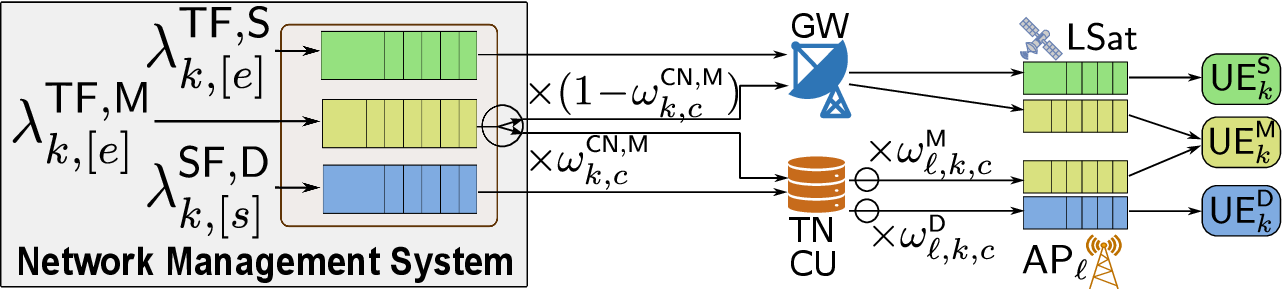}
    \vspace{-6mm}
    \captionsetup{font=footnotesize}
    \caption{Traffic steering model.}
    \label{fig:QueueModel}
    \vspace{-3mm}
\end{figure}

\subsubsection{Traffic Steering} As depicted in Fig.~\ref{fig:QueueModel}, at the CN, the $\sf{D}$ and $\sf{S}$ traffic flows are routed to the TNs and SatNet, respectively, while $\sf{M}$ traffic is split and dynamically steered to both domains. Subsequently, at the TCU, the $\sf{D}$ and $\sf{M}$ flows are further directed to TAPs due to a TN steering scheduler.
 
For $\sf{M}$-traffic steering, we introduce variables $\boldsymbol{\omega}_{\!c}^{\sf{CN,M}} = [ \omega_{k,c}^{\sf{CN,M}}], \forall k \in \mathcal{K}_{\sf{M}}$, where $\omega_{k,c}^{\sf{CN,M}} \in [0,1]$ and $(1 - \omega_{k,c}^{\sf{CN,M}})$ are the portions of ${\sf{UE}}_{k}^{\sf{M}}$'s traffic which are steered to TN and SatNet in cycle $c$, respectively.
In the TN, let $\boldsymbol{\omega}_{\!c}^{\sf{x}} = [\omega_{\ell,k,c}^{\sf{x}}], \forall (\ell ,k) \in (\mathcal{L} \times \mathcal{K}_{\sf{x}}), {\sf{x}} \in \{\sf{D,M}\}$ be the TN steering variables, wherein $\omega_{\ell,k,c}^{\sf{x}} \in [0,1]$ indicates the flow-split portion of service $\sf{x}\in \{\sf{D,M}\}$ from TCU to $\sf{AP}_{\ell}$ in cycle $c$. Besides, the integrity of traffic flows is ensured by
\vspace{-1mm}
\begin{equation} \label{eq: flow-split}
    (C13): \quad \scaleobj{0.8}{\sum\nolimits_{\forall \ell \in \mathcal{L}} } \omega_{\ell,k,c}^{\sf{x}} = 1, \forall (c,k), {\sf{x}} \in \{\sf{D,M}\}. \nonumber \vspace{-1mm}
\end{equation}
Subsequently, the data arrival of $\sf{D}$ service in SF $s$, and that of $\sf{M}$ and $\sf{S}$ services in frame $e$ at ${\sf{AP}}_{\ell}$/LSat corresponding to UEs are respectively expressed as
\vspace{-1mm}
\begin{IEEEeqnarray}{ll} \label{eq: traffic arrival}
    \!\!\!\! \lambda_{\ell,k,[s]}^{{\sf{SF,D}}} \!=\! \omega_{\ell,k,c}^{\sf{D}} \lambda_{k,[s]}^{{\sf{SF,D}}}, \; \forall (\ell ,k,s) \in \mathcal{L} \times \mathcal{K}_{\sf{D}} \times \mathcal{T}_{\sf{SF}}^{{\sf{Cy}},c}, \forall c \subnum \label{eq: traffic arrival a}\\
    \!\!\!\! \lambda_{\ell,k,[e]}^{{\sf{TF,M}}} \!=\! \omega_{k,c}^{\sf{CN,M}} \omega_{\ell,k,c}^{\sf{M}} \lambda_{k,[e]}^{{\sf{TF,M}}}, \; \forall (\ell ,k,e) \!\! \in \!\! \mathcal{L} \! \times \! \mathcal{K}_{\sf{M}} \! \times \! \mathcal{T}_{\sf{TF}}^{{\sf{Cy}},c} \!,\! \forall c, \quad\;\; \subnum  \label{eq: traffic arrival b}\\
    \!\!\!\! \lambda_{0,k,[e]}^{{\sf{TF,M}}} \!=\! (1-\omega_{k}^{\sf{CN,M}})  \lambda_{k,[e]}^{{\sf{TF,M}}}, \; \forall (k,e) \in \mathcal{K}_{\sf{M}} \times \mathcal{T}_{\sf{TF}}^{{\sf{Cy}},c}, \forall c,  \subnum \label{eq: traffic arrival c}\\
    \!\!\!\! \lambda_{0,k,[e]}^{{\sf{TF,S}}} \!=\! \lambda_{k,[e]}^{{\sf{TF,S}}}, \; \forall k \in \mathcal{K}_{\sf{S}},\forall e, \subnum  \label{eq: traffic arrival d} \vspace{-1mm}
\end{IEEEeqnarray}
where $\mathcal{T}_{\sf{TF}}^{{\sf{Cy}},c}$ and $\mathcal{T}_{\sf{SF}}^{{\sf{Cy}},c}$ are the frame and SF sets in cycle $c$.
For $\sf{D}$ service, the SF time-scale data-processing requirement is cast by the following constraint,
\vspace{-1mm}
\begin{equation}
    (C14): \quad T_{\sf{D}} R_{\ell,k,[s]}^{{\sf{SF,D}}} \! \geq \! \lambda_{\ell,k,[s]}^{{\sf{SF,D}}}, \forall \ell, \forall k \in \mathcal{K}_{\sf{D}}, \forall s. \nonumber \vspace{-2mm}
\end{equation}

\subsubsection{Service queues} Assume that TAPs and LSat are equipped with service-specific buffers.  
Each TAP/LSat is assumed to maintain separate queues for the UE data flows \cite{kavehmadavani_TWC23}. 
Subsequently, let $q_{\ell,k}^{{\sf{M}},[n_{\sf{M}}]}$, $q_{0,k}^{{\sf{M}},[n_{\sf{M}}]}$, and $q_{0,k}^{{\sf{S}},[n_{\sf{S}}]}$ be the queue lengths (QLs) at ${\sf{AP}}_{\ell}/{\sf{LSat}}$ of flow $k$ of services ${\sf{M,S}}$ at each RB time, the QL evolution over time is expressed as
\begin{IEEEeqnarray}{ll} \label{eq:queue TS}
    \hspace{-8mm} q_{\ell,k}^{{\sf{M}},[n_{\sf{M}} + 1]} \!\!=\!\! \left[ q_{\ell,k}^{{\sf{M}},[n_{\sf{M}}]} \!+\!  \lambda_{\ell,k}^{{\sf{M}},[n_{\sf{M}}]} \!-\! T_{\sf{M}} R_{\ell,k}^{{\sf{M}},[n_{\sf{M}}]} \right]^{+} \!\!,\! \forall \ell, \forall k \! \in \! \mathcal{K}_{\sf{M}} ,  \subnum \\
    \hspace{-8mm} q_{0,k}^{{\sf{x}},[n_{\sf{x}} + 1]} \!\!=\!\! \left[q_{0,k}^{{\sf{x}},[n_{\sf{x}}]} \!\!+\!  \lambda_{0,k}^{{\sf{x}},[n_{\sf{x}}]} \!\!-\! T_{\sf{x}} R_{0,k}^{{\sf{x}},[n_{\sf{x}}]} \right]^{+} \!\!\!, \! \forall k \! \in \! \mathcal{K}_{\sf{x}}, \! \forall {\sf{x}} \! \in \!\! \{\sf{ M,\! S }\}, \subnum
\end{IEEEeqnarray}
where $[\cdot]^{+} = \max \{0, \cdot\}$; $\lambda_{\ell,k}^{{\sf{M}},[n_{\sf{M}}]} $, $\lambda_{0,k}^{{\sf{M}},[n_{\sf{M}}]}$, and $\lambda_{0,k}^{{\sf{S}},[n_{\sf{S}}]}$ are intermediate data arrival notations with setting
\begin{IEEEeqnarray}{ll} \label{eq: set data arrival}
\!\!\!\!\!\!
    \begin{cases}
        \!\! \lambda_{\ell,k}^{{\sf{M}},[n_{\sf{M}}]} \!\!=\!\! \lambda_{\ell,k,[e]}^{{\sf{TF,M}}}, \; \lambda_{0,k}^{{\sf{x}},[n_{\sf{x}}]} \!\!=\!\! \lambda_{0,k,[e]}^{{\sf{TF,x}}} \text{ if } {\sf{mod}}(n_{\sf{x}},10 N_{\sf{x}}) \!=\! 1,  \quad \\
        \!\! \lambda_{\ell,k}^{{\sf{M}},[n_{\sf{M}}]}=0, \;\lambda_{0,k}^{{\sf{x}},[n_{\sf{x}}]}= 0 \text{ otherwise},
    \end{cases}
\end{IEEEeqnarray}
where $\{ n_{\sf{x}}\vert_{{\sf{mod}}(n_{\sf{x}},10 N_{\sf{x}})=1}\}$ indicates first RB time of each frame.
To maintain stability, the total QL of each service $\sf{x} \in \{M,S\}$ at ${\sf{AP}}_{\ell}/{\sf{LSat}}$ must satisfy \cite{kavehmadavani_TWC23}
\begin{equation} \label{eq:queue-length}
\begin{array}{ll}
    (C15): \; q_{\ell}^{{\sf{M}},[n_{\sf{M}}]} = \sum_{\forall k \in \mathcal{K}_{\sf{m}}} q_{\ell,k}^{{\sf{M}},[n_{\sf{M}}]} \leq q_{\ell}^{\sf{M,max}}, \; \forall (\ell, n_{\sf{M}}), \nonumber \\
    (C16): \; q_{0}^{{\sf{x}},[n_{\sf{x}}]} = \sum_{\forall k \in \mathcal{K}_{\sf{x}}} q_{0,k}^{{\sf{x}},[n_{\sf{x}}]} \leq q_{0}^{\sf{x,max}}, \; \forall n_{\sf{x}} , {\sf{x}} \in \{\sf{M,S}\}, \nonumber 
\end{array}
\end{equation}
where $q_{\ell}^{\sf{M,max}}$ and $q_{0}^{\sf{x,max}}$ are the maximum QL of ${\sf{AP}}_{\ell}$ and LSat for corresponding services, respectively. For convenience, let's denote $\boldsymbol{q}_{\!c}^{\sf{TN}} = \{ q_{\ell,k}^{{\sf{M}},[n_{\sf{M}}]} \vert_{\forall \ell,\forall k, \forall n_{\sf{M}} \in \mathcal{T}_{\sf{M}}^{{\sf{Cy}},c}} \}$, $\boldsymbol{q}_{\!c}^{\sf{Sat}} = \{ q_{0,k}^{{\sf{x}},[n_{\sf{x}}]} \vert_{\forall k, \forall n_{\sf{x}} \in \mathcal{T}_{\sf{x}}^{{\sf{Cy}},c}, \forall {\sf{x}} \in \{\sf{M,S}\}} \}$, and $\boldsymbol{q}_{\!c} \triangleq \{ \boldsymbol{q}_{\!c}^{\sf{TN}},\boldsymbol{q}_{\!c}^{\sf{Sat}} \}$.

\begin{remark}
    Under continuous traffic, constraints $(C15)$ and $(C16)$ ensure user fairness in terms of service fulfillment while also guaranteeing a finite service completion time.
\end{remark}

\begin{remark}
    The arrival data packet of $\sf{D}$ services is completely served within the next SF, as ensured by constraint $(C14)$. Hence, the buffer of $\sf{D}$ services is cleared after each SF, and the QL of $\sf{D}$ services is not considered. 
\end{remark}

\begin{remark}
    With knowledge of the queue length evolution and arrival rate, the ``effective'' sum-rate and throughput--corresponding to the amount of served data from the remaining buffer queue and/or newly arrived packets--can be determined.
\end{remark}

\vspace{-2mm}
\subsection{Problem Formulation}
\vspace{-1mm}
Typically, the uneven traffic demand of services and the inefficient spectrum utilization may lead to congestion in systems. Hence, this work aims to minimize the system congestion, i.e., the mean system QL, by optimizing the BWA, traffic split decision, AP-UE and LSat-UE associations, RB assignment, and power control. 
As a result, the utility function is the total average QL per RB time of $\sf{M}$ and $\sf{S}$ services at TAPs and LSats in each cycle $c$ defined as
\vspace{-1mm}
\begin{equation} \label{eq: objective}
\vspace{-1mm}
    f_{\sf{obj}}(\boldsymbol{q}_{\!c}) = \scaleobj{0.8}{\frac{1}{N_{\sf{M}} N_{\sf{SF}}} } \!\!\!\!\!\!\!\!\!\!\!\!\! \scaleobj{0.8}{\sum_{ \;\;\;\;\;\;\;\;\;\;\;\;\; \forall n_{\sf{M}} \in \mathcal{T}_{\sf{M}}^{{\sf{Cy}},c}, \forall \ell} } \!\!\!\!\!\!\!\!\!\!\!\! q_{\ell}^{{\sf{M}},[n_{\sf{M}}]} + \scaleobj{0.8}{ \sum_{\sf{x \in \{\sf{M,S}\}}} } \scaleobj{0.8}{\frac{1}{N_{\sf{x}} N_{\sf{SF}}} } \!\!\!\!\!\!\!\! \scaleobj{0.8}{\sum_{\;\;\;\;\;\; \forall n_{\sf{x}} \in \mathcal{T}_{\sf{x}}^{{\sf{Cy}},c}} } \!\!\!\!\!\!\! q_{0}^{{\sf{x}},[n_{\sf{x}}]}. 
\vspace{-1mm}
\end{equation}
Subsequently, two problems, ``DT-JointRA'' and ``RT-Refine'', will be studied. The former optimizes systems for each cycle $c$ and the latter optimizes TN short-term decisions for each SF, which are solved at the CNC and TCU, respectively.
\subsubsection{DT-JointRA Problem}
For cycle $c$, we consider
\vspace{-2mm}
\begin{equation}\label{eq: Prob0}
\begin{array}{lll} 
    \!\! (\mathcal{P}_{0})_{\!c}: \!\!\!\!\! &  \min\limits_{\boldsymbol{b}_{\!c},\boldsymbol{\omega}_{\!c},\boldsymbol{P}_{\!c}, \boldsymbol{\alpha}_{\!c}, \boldsymbol{\beta}_{\!c} ,\boldsymbol{q}_{\!c} } f_{\sf{obj}}(\boldsymbol{q}_{\!c})  \;
    \st \; \text{constraints } (C1)\!-\!(C16), \nonumber \\
    & (C0): b_{v_{\sf{x}},c}^{\sf{x}}, \alpha_{\ell,k}^{[v_{\sf{x}},n_{\sf{x}}]}, \beta_{k}^{[v_{\sf{x}},n_{\sf{x}}]}  \in \! \{0,1\}, \forall (\ell ,k,v_{\sf{x}},n_{\sf{x}},c). \nonumber
    \vspace{-1mm}
\end{array}
\end{equation}

\subsubsection{RT-Refine Problem}
Regarding implementation aspects, problem $(\mathcal{P}_{0})_{\!c}$ requires future information $\{ \boldsymbol{h}_{\!c},\boldsymbol{g}_{\!c},\boldsymbol{\lambda}_{\!c} \}$ in cycle $c$ which is challenging in practice. 
Assuming that the predicted information $\{ \hat{\boldsymbol{h}}_{\!c},\hat{\boldsymbol{g}}_{\!c},\hat{\boldsymbol{\lambda}}_{\!c} \}$ is used instead of $\{ \boldsymbol{h}_{\!c},\boldsymbol{g}_{\!c},\boldsymbol{\lambda}_{\!c} \}$, the corresponding solution should be adjusted appropriately with actual information. 
Hence, this section formulates the refinement problem to adjust the obtained solution at each SF by re-optimizing TN short-term decisions.
However, due to the long propagation distance in SatNet, the channel estimation and RA in SatNet at each SF are challenging. 
Hence, the given RA in SatNet can remain unchanged.

Let's call UEs served by TAPs and LSat as TN UEs and SatNet UEs. Assuming that TAPs only estimate TN UE CSIs, the predicted CSIs of SatNet UEs should be used. 
For a given SatCom RA, the channel uncertainty of $\sf{M}$ UEs served by LSat leads to errors in the ISyI power $\Theta_{k}^{{\sf{TN}},[v_{\sf{M}},n_{\sf{M}}]}$ and $\Theta_{k}^{{\sf{Sat}},[v_{\sf{M}},n_{\sf{M}}]}$. Hence, we impose the interference margin coefficient $\kappa \geq 1$ and use $\hat{\Theta}_{k}^{{\sf{TN}},[v_{\sf{M}},n_{\sf{M}}]} = \kappa \Theta_{k}^{{\sf{TN}},[v_{\sf{M}},n_{\sf{M}}]}$ and $\hat{\Theta}_{k}^{{\sf{Sat}},[v_{\sf{M}},n_{\sf{M}}]} = \kappa \Theta_{k}^{{\sf{Sat}},[v_{\sf{M}},n_{\sf{M}}]}$ for refinement problem.

Let $\boldsymbol{h}_{e}^{*}$ be the channel gain matrix in frame $e$ where CSIs of TN UEs are actual while those of SatNet UEs are extracted from $\hat{\boldsymbol{h}}_{e}$.
Based on channel gains $\{\boldsymbol{h}_{e}^{*}, \hat{\boldsymbol{g}}_{e}\}$ and actual traffic $\boldsymbol{\lambda}_{e}$ in frame $e$, for a fixed $\{ \boldsymbol{b}_{\!c}, \boldsymbol{\omega}_{\!c}, \boldsymbol{p}_{\!0,c},  \boldsymbol{\beta}_{\!c} \}$, the refinement problem for TN short-term decisions at SF $s$ is formulated as
\begin{equation}
\begin{array}{lll} \label{eq: Prob0Refine}
    \!\!\!(\mathcal{P}_{0}^{\sf{refi}})_{s} \!\! : \min\limits_{\boldsymbol{p}[s], \boldsymbol{\alpha}[s] ,\boldsymbol{q}[s] } \;f_{\sf{obj}}(\boldsymbol{q}[s]) \nonumber \\
    \hspace{-2mm} \st \; (C4)_{s}\!\!-\!\!(C6)_{s}\!,\! (C9)_{s},\!(C10)_{s}, \! (C11)_{s}, \!(C14)_{s}, \!(C15)'_{s} \!, \!(C16)'_{s} \!, \nonumber
\end{array}
\end{equation}
where index $s$ used in $\boldsymbol{x}[s]$ and $(Cx)_{s}$ indicates time index adaptation to SF $s$ for variable $\boldsymbol{x}$ and constraint $(Cx)$, 
$(C15)'$ and $(C16)'$ are the constraint revised from $(C15)$ and $(C16)$ in which terms $\hat{\Theta}_{k}^{{\sf{TN}},[v_{\sf{M}},n_{\sf{M}}]}$ and $\hat{\Theta}_{k}^{{\sf{Sat}},[v_{\sf{M}},n_{\sf{M}}]}$ are used instead of $\Theta_{k}^{{\sf{TN}},[v_{\sf{M}},n_{\sf{M}}]}$ and $\Theta_{k}^{{\sf{Sat}},[v_{\sf{M}},n_{\sf{M}}]}$ in rate functions, respectively.

\textit{Challenges in solving $(\mathcal{P}_{0})_{\!c}$ and $(\mathcal{P}_{0}^{\sf{refi}})_{s}$}: these problems contain the non-convex constraints, SINR and rate functions. Especially, problem $(\mathcal{P}_{0})_{\!c}$ further consists of both continuous and binary variables. Hence, $(\mathcal{P}_{0}^{\sf{refi}})_{s}$ and $(\mathcal{P}_{0})_{\!c}$ are non-convex and MINLP problems. 
Additionally, the unknown channel and traffic information in $(\mathcal{P}_{0})_{\!c}$ makes it more complicated.

\subsubsection{Overall workflow}
Let $\{\hat{\boldsymbol{x}}_{\!c} \}$ and $\{\boldsymbol{x}_{\!c} \}$ with $ \boldsymbol{x}  \in \{\boldsymbol{h}, \boldsymbol{g},\boldsymbol{\lambda} ,\boldsymbol{u}^{\sf{ue}} ,\boldsymbol{u}^{\sf{sat}} ,{\sf{TLE}} \}$ be the DT-virtual and real information vector in cycle $c$, the overall workflow is summarized as
\begin{enumerate}[label=\protect\large\textcircled{\small \arabic*}]
    \item \textbf{DT update:} The UE position $\hat{\boldsymbol{u}}_{c-1}^{{\sf{ue}}}$, orbit information ${\sf{TLE}}_{c-1}$, and traffic $\boldsymbol{\lambda}_{c-1}$ in cycle $(c-1)$ are collected and used to update DT-system
    \item \textbf{Prediction:} The DT model first predicts $\hat{\boldsymbol{\lambda}}_{\!c}$, $\{ \hat{\boldsymbol{u}}_{\!c}^{\sf{ue}} \!,\! \hat{\boldsymbol{u}}_{\!c}^{\sf{sat}} \}$, and then $\{\hat{\boldsymbol{h}}_{\!c},\hat{\boldsymbol{g}}_{\!c} \}$ in cycle $c$.
    \item \textbf{Extract information:} Traffic and channel information $\{\hat{\boldsymbol{h}}_{\!c}, \hat{\boldsymbol{g}}_{\!c},\hat{\boldsymbol{\lambda}}_{\!c}\}$ is extracted from the predicted information.
    \item \textbf{DT-based Joint-RA:} Based on $\{\hat{\boldsymbol{h}}_{\!c}, \hat{\boldsymbol{g}}_{\!c},\hat{\boldsymbol{\lambda}}_{\!c}\}$, problem $(\mathcal{P}_{0})_{\!c}$ is solved at the CNC to obtain long-term decisions and preliminary estimates of short-term ones $\{ \boldsymbol{b}_{\!c}, \boldsymbol{\omega}_{\!c}, \boldsymbol{P}_{\!c}, \boldsymbol{\alpha}_{\!c}, \boldsymbol{\beta}_{\!c} \}$ for cycle $c$.
    \item \textbf{Execute solution:} BWA $\{\boldsymbol{b}_{c}\}$, traffic steering  $\{\boldsymbol{\omega}_{c}\}$, and RA  $\{\boldsymbol{p}_{0,c}, \boldsymbol{\beta}_{c}\}$ for SatCom are applied.
    \item \textbf{RT-Refinement:} For each frame $e$ in cycle $c$, using channel gain $\boldsymbol{h}_{e}^{*}$, arrival data rate $\boldsymbol{\lambda}_{e}$ in the real system, and given decisions $\{ \boldsymbol{b}_{\!c}, \boldsymbol{\omega}_{\!c}, \boldsymbol{p}_{0,c}, \boldsymbol{\beta}_{\!c} \}$, the initial power control, UE association, and RB assignment $\{\boldsymbol{p}_{\!c},\boldsymbol{\alpha}_{\!c}\}$ in the TN, which are outcomes of step 4, are refined by solving $(\mathcal{P}_{0}^{\sf{refi}})_{s}$ at TCU at each SF $s$ in frame $e$.

\end{enumerate}

\begin{figure}
\captionsetup{font=footnotesize}
    \centering
    \includegraphics[width=1\linewidth]{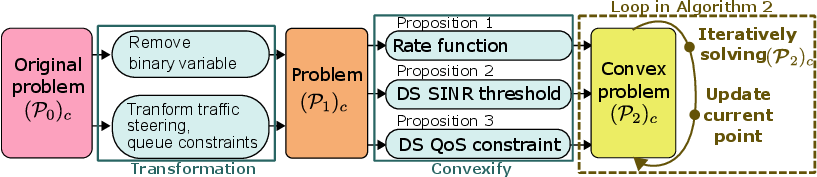}
    \vspace{-1mm}
    \caption{Developing workflow.}
    \label{fig:flowchart_DTJRA}
    \vspace{-3mm}
\end{figure}

\vspace{-3mm}
\section{Proposed Solutions} \label{sec: solution}
\vspace{-1mm}
This section first proposes the solution for problem $(\mathcal{P}_{0})_{\!c}$. Subsequently, the solution for the problem $(\mathcal{P}_{0}^{\sf{refi}})_{s}$ is proposed to refine TN short-term decisions with actual information. The solution workflows are summarized in Figs.~\ref{fig:flowchart_DTJRA} and ~\ref{fig:flowchart_RTRefine}.

\vspace{-3mm}
\subsection{DT-based Prediction}
At the beginning of every cycle, DT model needs the updated UEs position and LSat orbit information for channel prediction. Assuming that the cycle duration is sufficiently small, the changes in the movement trend of UEs and velocity are negligible. Hence, with a given updated route and velocity in cycle $(c-1)$, the UE position in cycle $c$ can be estimated by the approach given in \cite{Hung_TCOM25}. Subsequently, channel components in cycle $c$ are identified by \eqref{eq: DT channel predict}.
Based on output ${\sf{chan}}_{\!c} \! \triangleq \! \{\scaleobj{0.9}{\sf{chan}}_{[e]} | \forall e \! \in \! \mathcal{T}_{\sf{TF}}^{{\sf{Cy}},c} \}$, channel gain $\{ \hat{\boldsymbol{h}}_{\!c}, \hat{\boldsymbol{g}}_{\!c} \}$ used for problem $(\mathcal{P}_{0})_{\!c}$ is computed with DT channel coefficient $\xi=1$.

Regarding UE traffic, one can assume that UEs continue to use specific services over a given time duration. As a result, their traffic pattern and arrival traffic remain relatively stable. Therefore, the arrival data in cycle $c$ can be estimated by the average in cycle $(c-1)$ as
\begin{IEEEeqnarray}{ll}\label{eq: traffic pred}
    \!\!\!\!\!\!\!\!\! \hat{\lambda}_{k,[s]}^{{\sf{SF,D}}} = \scaleobj{0.8}{\frac{1}{N_{\sf{SF}}} } \scaleobj{0.8}{ \sum\nolimits_{s' \in \mathcal{T}_{\sf{SF}}^{{\sf{Cy}},c-1} } } \lambda_{k,[s']}^{{\sf{SF,D}}}, \; \forall s \in \mathcal{T}_{\sf{SF}}^{{\sf{Cy}},c}, \subnum \\
    \!\!\!\!\!\!\!\!\! \hat{\lambda}_{k,[e]}^{{\sf{TF,x}}} = \scaleobj{0.8}{\frac{1}{N_{\sf{TF}}} } \scaleobj{0.8}{\sum\nolimits_{e' \in \mathcal{T}_{\sf{TF}}^{{\sf{Cy}},c-1} } } \lambda_{k,[e']}^{{\sf{TF,x}}}, \; \forall e \in \mathcal{T}_{\sf{TF}}^{{\sf{Cy}},c}, {\sf{x}} \in \scaleobj{0.9}{\sf{\{M,S \}} }. \subnum
\end{IEEEeqnarray}
The DT-based prediction mechanism is summarized in Alg.~\ref{alg: DTPred}.

\begin{algorithm}[!t]
\scriptsize
\begin{algorithmic}[1]
    \captionsetup{font=small}
    \protect\caption{\
    \textsc{DT-based Prediction}}
    \label{alg: DTPred}
    \STATE \textbf{Input:} Actual information $\{ \boldsymbol{\lambda}_{c-1}, \boldsymbol{u}_{c-1}^{\sf{ue}}, {\sf{TLE}}_{c-1}\}$ in cycle $(c-1)$.
    \STATE Predict $\{ \hat{\boldsymbol{u}}_{\!c}^{\sf{ue}}, \hat{\boldsymbol{u}}_{\!c}^{\sf{sat}}\}$ and construct $\{ \widehat{\scaleobj{0.8}{\sf{UE}}}[e], \widehat{\scaleobj{0.8}{\sf{Sat}}}[e] |\forall e \in \mathcal{T}_{\sf{TF}}^{{\sf{Cy}},c} \}$, 
    \STATE Predict $\{ \scaleobj{0.8}{\sf{chan}}_{\!c}, \hat{\boldsymbol{\lambda}}_{\!c} \}$ by \eqref{eq: DT channel predict}, \eqref{eq: traffic pred}.
    \STATE Construct $\{\hat{\boldsymbol{h}}_{\!c}, \hat{\boldsymbol{g}}_{\!c} \}$ by using ${\sf{chan}_{\!c}}$ with $\xi=1$.
    \STATE \textbf{Output:} Predicted information $\{\hat{\boldsymbol{h}}_{\!c}, \hat{\boldsymbol{g}}_{\!c}, \hat{\boldsymbol{\lambda}}_{\!c} \}$.
\end{algorithmic}
\normalsize 
\end{algorithm}

\vspace{-3mm}

\subsection{DT-JointRA Solution}
\subsubsection{Compressed-Sensing-Based Relaxation}
Considering the system model and the optimization problem $(\mathcal{P}_{0})_{\!c}$, one can see the relationship between the binary variable $\{ \boldsymbol{\alpha}_{\!c}, \boldsymbol{\beta}_{\!c} \}$ and the continuous variable $\boldsymbol{P}_{\!c}$ as follows. Consider ${\sf{RB}}_{[v_{\sf{x}},n_{\sf{x}}]}$:

\begin{itemize}
    \item If ${\sf{AP}}_{\ell}$ serves ${\sf{UE}}_{k} \Rightarrow \alpha_{\ell,k}^{[v_{\sf{x}},n_{\sf{x}}]}=1$, $p_{\ell,k}^{[v_{\sf{x}},n_{\sf{x}}]} > 0$; \\
    if  ${\sf{AP}}_{\ell}$ does not serve ${\sf{UE}}_{k}$ $\alpha_{\ell,k}^{[v_{\sf{x}},n_{\sf{x}}]}=0, p_{\ell,k}^{[v_{\sf{x}},n_{\sf{x}}]}=0$.
    \item If the LSat serves ${\sf{UE}}_{k} \Rightarrow \beta_{k}^{[v_{\sf{x}},n_{\sf{x}}]}=1$, $p_{0,k}^{[v_{\sf{x}},n_{\sf{x}}]} > 0$; \\
    if the LSat does not serve ${\sf{UE}}_{k}$ $\beta_{k}^{[v_{\sf{x}},n_{\sf{x}}]}=0, p_{0,k}^{[v_{\sf{x}},n_{\sf{x}}]}=0$.
\end{itemize}
Based on this relationship, the binary variables $\boldsymbol{\alpha}_{\!c}$ and $\boldsymbol{\beta}_{\!c}$ can be respectively represented by the continuous one $\boldsymbol{P}_{\!c}$ as
\vspace{-1mm}
\begin{align}\label{eq: reduce binary}
    \alpha_{\ell,k}^{[v_{\sf{x}},n_{\sf{x}}]} = \Vert p_{\ell,k}^{[v_{\sf{x}},n_{\sf{x}}]} \Vert_{0}, \quad \beta_{k}^{[v_{\sf{x}},n_{\sf{x}}]} = \Vert p_{0,k}^{[v_{\sf{x}},n_{\sf{x}}]} \Vert_{0}.
    \vspace{-1mm}
\end{align}
Besides, consider $(C4),(C7)$, one can see that SC $v_{\sf{x}}$ is activated in cycle $c$, i.e., $b^{\sf{x}}_{v_{\sf{x}},c}=1$, only if it is used by at least one TAP/LSat-UE link at any RB time $n_{\sf{x}} \in \mathcal{T}_{\sf{x}}^{{\sf{Cy}},c}$, i.e., $\alpha_{\ell,k}^{[v_{\sf{x}},n_{\sf{x}}]} \!\!=\!\!1$ or $\beta_{k}^{[v_{\sf{x}},n_{\sf{x}}]} \!\!=\!\! 1$. 
Combining \eqref{eq: reduce binary}, $\boldsymbol{b}_{\!c}$ is rewritten as
\vspace{-1mm}
\bieq{ll} \label{eq: BW by norm0}
    \hspace{-6mm} b^{\sf{D}}_{v_{\sf{D}},c} = \Vert \!\!\!\! \scaleobj{0.8}{\sum_{\forall (\ell ,k), \forall n_{\sf{D}} \in \mathcal{T}_{\sf{D}}^{{\sf{Cy}},c} } } \!\!\!\! \alpha_{\ell,k}^{[v_{\sf{D}},n_{\sf{D}}]} \Vert_{0} = \Vert p^{{\sf{D}}}_{v_{\sf{D}},c} \Vert_{0}, \; \forall (v_{\sf{D}}, c), \subnum \\
    \hspace{-6mm} b^{\sf{M}}_{v_{\sf{M}},c} \!=\! \Vert \hspace{-12mm} \scaleobj{0.8}{\sum_{\hspace{12mm} \forall (\ell ,k), \forall n_{\sf{M}} \in \mathcal{T}_{\sf{M}}^{{\sf{Cy}},c}} } \hspace{-11mm} \alpha_{\ell,k}^{[v_{\sf{M}},n_{\sf{M}}]} 
    \!+\! \hspace{-10mm} \scaleobj{0.8}{ \sum_{\hspace{12mm} \forall k, \forall n_{\sf{M}} \in \mathcal{T}_{\sf{M}}^{{\sf{Cy}},c}} } \hspace{-10mm} \beta_{k}^{[v_{\sf{M}},n_{\sf{M}}]} \Vert_{0} 
    = \Vert P^{\sf{M}}_{v_{\sf{M}},c} \Vert_{0}, \; \forall (v_{\sf{M}}, c), \subnum \\
    \hspace{-6mm} b^{\sf{S}}_{v_{\sf{S}},c} = \Vert \scaleobj{0.8}{\sum\nolimits_{\forall k, \forall n_{\sf{S}} \in \mathcal{T}_{\sf{S}}^{{\sf{Cy}},c}} } \beta_{k}^{[v_{\sf{S}},n_{\sf{S}}]} \Vert_{0} = \Vert p^{\sf{S}}_{0,v_{\sf{S}},c} \Vert_{0}, \; \forall (v_{\sf{S}},c). \subnum \vspace{-1mm}
\eieq
where $P^{\sf{M}}_{v_{\sf{M}},c}=p^{\sf{M}}_{v_{\sf{M}},c} \!\!+\!\! p^{\sf{M}}_{0,v_{\sf{M}},c}$, 
$p^{\sf{x}}_{v_{\sf{x}},c} =  {\sum\nolimits_{\substack{\forall (\ell ,k),  \forall n_{\sf{x}} \in \mathcal{T}_{\sf{x}}^{{\sf{Cy}},c}}} } p_{\ell,k}^{[v_{\sf{x}},n_{\sf{x}}]} $ and $P^{\sf{x}}_{v_{\sf{x}},c}= {\sum_{\substack{\forall k, \forall n_{\sf{x}} \in \mathcal{T}_{\sf{x}}^{{\sf{Cy}},c} }}} p_{0,k}^{[v_{\sf{x}},n_{\sf{x}}]} $.
Thanks to relationship \eqref{eq: BW by norm0}, constraints $(C4)$,$(C7)$ can be omitted.

Exploiting \eqref{eq: reduce binary}, the binary components can be replaced by the corresponding $\ell_{0}$-norm components. Moreover, the binary variables in the production components and the binary arguments of SINR/rate functions can be omitted.
Specifically, $R_{\ell,k,[s]}^{{\sf{SF,D}}}(\boldsymbol{P}_{\!c},\boldsymbol{\alpha}_{\!c})$ is reformulated as
\vspace{-1mm}
\beq \label{eq: rate DS 2}
    R_{\ell,k,[s]}^{{\sf{SF,D}}}(\boldsymbol{P}_{\!c}) = w_{\sf{D}} \!\!\!\!\! \scaleobj{0.8}{\sum_{ \forall (v_{\sf{D}},n_{\sf{D}}) } } \!\!\!\!\! \log_2{(1 + \gamma_{\ell,k}^{{\sf{D}},[v_{\sf{D}},n_{\sf{D}}]}(\boldsymbol{P}_{\!c})}) 
    - \chi_{\sf{D}} \!\! \scaleobj{0.8}{\sqrt{ \!\!\!\!\!\!\! \scaleobj{0.9}{\sum_{\;\;\;\;\;\;\; \forall (v_{\sf{D}}, n_{\sf{D}}) } }  \!\!\!\!\!\! \Vert p_{\ell,k}^{[v_{\sf{D}},n_{\sf{D}}]} \Vert_{0} }  }  , \vspace{-1mm}
\eeq

Subsequently, the $\ell_{0}$-norm component $\Vert x \Vert_{0}, \forall x \!\! \geq \!\! 0$, can be approximated by concave function $\mathcal{F}_{\! \sf{apx}}\!(x)\triangleq 1 - e^{-x/ \epsilon}, \; 0<\epsilon \ll 1 $, as $\Vert x \Vert_{0} \! \approx \! \mathcal{F}_{\! \sf{apx}}\!(x)$. Let $\mathcal{F}_{\! \sf{apx}}^{(i)}\!(x)$ be an upper bound of $\mathcal{F}_{\! \sf{apx}}\!(x)$ at feasible point $(x^{(i)})$, the $\ell_{0}$-norm components in \eqref{eq: reduce binary} and \eqref{eq: BW by norm0} can be approximated at iteration $i$ as
\begin{equation}\label{eq: apx norm0}
    \Vert \tau \Vert_{0} \!\! \approx \!\! \mathcal{F}_{\! \sf{apx}}\!(\tau) \!\! \leq \!\! \mathcal{F}_{\! \sf{apx}}^{(i)}\!(\tau), \; \tau \!\in\!\! \{ p_{\ell,k}^{[v_{\sf{x}},n_{\sf{x}}]} \!,\! p_{0,k}^{[v_{\sf{x}},n_{\sf{x}}]} \!,\! p^{\sf{x}}_{v_{\sf{x}},c} ,\! p^{\sf{x}}_{0,v_{\sf{x}},c} ,\!   P^{\sf{x}}_{v_{\sf{x}},c} \},
\end{equation}
where upper bound $\mathcal{F}_{\! \sf{apx}}^{(i)}\!(x)$ can be obtained based on \cite{Hung_TCOM25} as
\vspace{-1mm}
\begin{equation}\label{eq: fapx}
    \mathcal{F}_{\! \sf{apx}}\!(x) \leq \mathcal{F}_{\! \sf{apx}}^{(i)}\!(x) \triangleq 1/ \epsilon \exp(-x^{(i)} / \epsilon)(x - x^{(i)} - \epsilon) + 1.
    \vspace{-1mm}
\end{equation} 
Utilize \eqref{eq: reduce binary}, \eqref{eq: BW by norm0}, and \eqref{eq: apx norm0}, 
by replacing binary component $b_{v_{\sf{x}},c}^{\sf{x}}$, $\alpha_{\ell,k}^{[v_{\sf{x}},n_{\sf{x}}]}$, and $\beta_{k}^{[v_{\sf{x}},n_{\sf{x}}]}$ by the corresponding approximated terms $\mathcal{F}_{\! \sf{apx}}^{(i)}\!(\cdot)$ at iteration $i$, 
constraints $(C1)-(C3),(C5),(C6),(C8)-(C9)$ are directly transformed into convex ones so-called $\boldsymbol{(\tilde{C}1)}-\boldsymbol{(\tilde{C}3)},\boldsymbol{(\tilde{C}5)},\boldsymbol{(\tilde{C}6)},\boldsymbol{(\tilde{C}8)}-\boldsymbol{(\tilde{C}9)}$; and $(C10)$ is approximated as
\vspace{-2mm}
\begin{IEEEeqnarray}{ll}
    (\bar{C}10): \; \gamma_{\ell,k}^{{\sf{D}},[v_{\sf{D}},n_{\sf{D}}]}(\boldsymbol{P}_{\!c}) \geq \mathcal{F}_{\! \sf{apx}}^{(i)}\!(p_{\ell,k}^{[v_{\sf{D}},n_{\sf{D}}]}) \gamma_{0}^{\sf{D}}, \; \forall (\ell, k,v_{\sf{D}},n_{\sf{D}}). \nonumber
    \vspace{-1mm}
\end{IEEEeqnarray}
\subsubsection{Transform Traffic and Queue Constraints}
$(C14)-(C16)$ are non-convex due to non-convex rate functions, coupling between traffic steering variables, and QL dependency.
First, we introduce variable $\bar{\boldsymbol{\omega}} = \{ \bar{\omega}_{\ell,k,c}^{\sf{x}} | \forall (\ell ,k,c,{\sf{x}}) \} $ and use it instead of $\omega_{k,c}^{\sf{CN,M}}$ and $\omega_{\ell,k,c}^{\sf{x}}$ with the relationship
\vspace{-2mm}
\begin{equation} \label{eq: transform flow-split}
    \bar{\omega}_{\ell,k,c}^{\sf{D}}= \omega_{\ell,k,c}^{\sf{D}} \quad \text{ and } \quad
    \bar{\omega}_{\ell,k,c}^{\sf{M}} = \omega_{k,c}^{\sf{CN,M}} \omega_{\ell,k,c}^{\sf{M}},
    \vspace{-1mm}
\end{equation}
Hence, constraint $(C13)$ is rewritten for the $\sf{D}$ service as
\vspace{-2mm}
\begin{equation} \label{eq: flow-split 1}
    (\tilde{C}13): \quad \scaleobj{0.8}{\sum\nolimits_{\forall \ell \in \mathcal{L} } } \bar{\omega}_{\ell,k,c}^{\sf{D}} = 1, \forall k \in  \mathcal{K}_{\sf{D}}, \forall c. \nonumber 
    \vspace{-1mm}
\end{equation}
The traffic arrival rates \eqref{eq: traffic arrival a}-\eqref{eq: traffic arrival c} are rewritten as
\vspace{-2mm}
\bieq{ll} \label{eq: traffic arrival 1}
    \lambda_{\ell,k,[s]}^{{\sf{SF,D}}} = \bar{\omega}_{\ell,k,c}^{\sf{D}} \lambda_{k,[s]}^{{\sf{SF,D}}}, \;\; \forall \ell, \forall k \in \mathcal{K}_{\sf{D}}, \forall s \in \mathcal{T}_{\sf{SF}}^{{\sf{Cy}},c}, \forall c, \subnum \label{eq: traffic arrival 1a} \\
    \lambda_{\ell,k,[e]}^{{\sf{TF,M}}} = \bar{\omega}_{\ell,k,c}^{\sf{M}} \lambda_{k,[e]}^{{\sf{TF,M}}}, \quad \forall \ell,\forall k \in \mathcal{K}_{\sf{M}}, \;\; e \in \mathcal{T}_{\sf{TF}}^{{\sf{Cy}},c}, \forall c, \subnum \label{eq: traffic arrival 1b}\\
    \lambda_{0,k,[e]}^{{\sf{TF,M}}} = (1 \!-\! \scaleobj{0.8}{\sum_{\forall \ell}} \bar{\omega}_{\ell,k,c}^{\sf{M}} ) \lambda_{k,[e]}^{{\sf{TF,M}}}, \;\; \forall k \! \in \! \mathcal{K}_{\sf{M}}, \forall e \! \in \! \mathcal{T}_{\sf{TF}}^{{\sf{Cy}},c} \!, \forall c. \quad\quad\; \subnum \label{eq: traffic arrival 1c}
    \vspace{-2mm}
\eieq

\begin{remark}
    The integrity of the $\sf{M}$ traffic is ensured thanks to representation \eqref{eq: transform flow-split}, \eqref{eq: traffic arrival 1}. Particularly, based on \eqref{eq: transform flow-split}, \eqref{eq: traffic arrival c} is rewritten by \eqref{eq: traffic arrival 1c} due to $\sum_{\forall \ell}{\bar{\omega}}_{\ell,k,c}^{\sf{M}} = \omega_{k,c}^{\sf{CN,M}}$. Since $\sum_{\forall \ell} \lambda_{\ell,k,[e]}^{{\sf{TF,M}}} + \lambda_{0,k,[e]}^{{\sf{TF,M}}} = \lambda_{k,[e]}^{{\sf{TF,M}}}$, the integrity is ensured.
\end{remark}

Subsequently, 
we introduce slack variables $\bar{\boldsymbol{q}}_{\!c} \!\!=\!\! \{ \! \bar{q}_{\ell,k}^{{\sf{x}},[n_{\sf{x}}]} \!,\! \bar{q}_{0,k}^{{\sf{x}},[n_{\sf{x}}]}$ $ \vert \forall (\ell ,k),$ $ \forall n_{\sf{x}} \! \in \! \mathcal{T}_{\sf{x}}^{{\sf{Cy}},c}, \forall {\sf{x}} \! \in \! \scaleobj{0.9}{\{\sf{M,S}\}} \}$ as the QL upper bound and $\boldsymbol{r}_{\!c} \!=\! \{ r_{\ell,k}^{{\sf{x}},[v_{\sf{x}},n_{\sf{x}}]} \!, r_{0,k}^{{\sf{x}},[v_{\sf{x}},n_{\sf{x}}]}, r_{\ell,k,[s]}^{{\sf{SF,D}}} | \forall (\ell ,k, s, v_{\sf{x}},n_{\sf{x}}), \forall {\sf{x}} \! \in \! \mathcal{S} \}$ as the rate's lower bound; and transform \eqref{eq:queue TS} and $(C14)-(C16)$ as
\vspace{-1mm}
\bieq{ll}
    (\tilde{C}14): T_{\sf{D}} r_{\ell,k,[s]}^{{\sf{SF,D}}} \! \geq \! \lambda_{\ell,k,[s]}^{{\sf{SF,D}}}, \forall \ell, \forall k \in \mathcal{K}_{\sf{d}}, \forall s, \nonumber \\
    (\tilde{C}15a)\!:\!  \bar{q}_{\ell,k}^{{\sf{M}},[n_{\sf{M}}]} \!\!+\! \lambda_{\ell,k}^{{\sf{m}},[n_{\sf{M}}]} \!\!-\! T_{\sf{M}} r_{\ell,k}^{{\sf{M}},[n_{\sf{M}}]} \leq \bar{q}_{\ell,k}^{{\sf{M}},[n_{\sf{M}}+1]}, \forall (\ell, k, n_{\sf{M}}), \nonumber \\
    (\tilde{C}15b): 0 \leq \bar{q}_{\ell,k}^{{\sf{M}},[n_{\sf{M}}]},  \scaleobj{0.8}{\sum\nolimits_{\forall k \in \mathcal{K}_{\sf{M}}}} \bar{q}_{\ell,k}^{{\sf{M}},[n_{\sf{M}}]} \leq q_{\ell}^{\sf{M,max}}, \forall (\ell,  n_{\sf{M}}), \nonumber \\
    (\tilde{C}16a)\!:\! \bar{q}_{0,k}^{{\sf{x}},[n_{\sf{x}}]} \!\!+\! \lambda_{0,k}^{{\sf{x}},[n_{\sf{x}}]} \!\!-\! T_{\sf{x}} r_{0,k}^{{\sf{x}},[n_{\sf{x}}]} \!\! \leq\! \bar{q}_{0,k}^{{\sf{x}},[n_{\sf{x}} + 1]} \!, \forall (k, n_{\sf{x}}) ,\! \forall {\sf{x}} \!\! \in \!\! \{\sf{M,S}\}, \nonumber \\
    (\tilde{C}16b):  0 \leq \bar{q}_{0,k}^{{\sf{x}},[n_{\sf{x}}]},  \scaleobj{0.8}{\sum\nolimits_{\forall k \in \mathcal{K}_{\sf{x}}}} \bar{q}_{0,k}^{{\sf{x}},[n_{\sf{x}}]} \!\! \leq \! q_{0}^{\sf{x,max}}, \; \forall n_{\sf{x}}, \forall {\sf{x}} \!\! \in \!\! \{\sf{M,S}\}, \nonumber 
\eieq
where $r_{\ell,k}^{{\sf{M}},[n_{\sf{M}}]} = \sum_{\forall v_{\sf{M}}} r_{\ell,k}^{{\sf{M}},[v_{\sf{M}},n_{\sf{M}}]}$, $r_{0,k}^{{\sf{x}},[n_{\sf{x}}]} = \sum_{\forall v_{\sf{x}}} r_{0,k}^{{\sf{x}},[v_{\sf{x}},n_{\sf{x}}]}$ while $r_{\ell,k}^{{\sf{M}},[v_{\sf{M}},n_{\sf{M}}]}$, $r_{\ell,k,[s]}^{{\sf{SF,D}}}$, and $r_{0,k}^{{\sf{x}},[v_{\sf{x}},n_{\sf{x}}]}$ satisfy constraints
\bieq{ll}
    (C17.1)\!: R_{\ell,k}^{{\sf{M}},[v_{\sf{M}},n_{\sf{M}}]}(\boldsymbol{P}_{\!c}) \! \geq \! r_{\ell,k}^{{\sf{M}},[v_{\sf{M}},n_{\sf{M}}]}, \; \forall (\ell, v_{\sf{M}},n_{\sf{M}}), \forall k \! \in \! \mathcal{K}_{\sf{M}}, \nonumber \\
    (C17.2)\!: R_{\ell,k,[s]}^{{\sf{SF,D}}}(\boldsymbol{P}_{\!c}) \! \geq \! r_{\ell,k,[s]}^{{\sf{SF,D}}}, \; \forall \ell, \forall k \in \mathcal{K}_{\sf{D}}, \forall s, \nonumber \\
    (C18)\!: r_{0,k}^{{\sf{x}},[v_{\sf{x}},n_{\sf{x}}]}(\boldsymbol{P}_{\!c}) \! \geq \! r_{0,k}^{{\sf{x}},[v_{\sf{x}},n_{\sf{x}}]}, \; \forall k \!\! \in \! \mathcal{K}_{\sf{x}} \!, \forall (v_{\sf{x}},n_{\sf{x}}), \forall {\sf{x}} \!\! \in \!\! \{\sf{M,S}\}, \nonumber
\eieq

Therefore, problem $(\mathcal{P}_{0})_{\!c}$ can be rewritten as
\vspace{-2mm}
\bieq{lr} \label{eq: Prob1}
    (\mathcal{P}_{1})_{\!c}: & \min_{\bar{\boldsymbol{\omega}}_{\!c},\boldsymbol{P}_{\!c}, \bar{\boldsymbol{q}}_{\!c}, \boldsymbol{r}_{\!c} } \; f_{\sf{obj}}(\bar{\boldsymbol{q}}_{\!c}) \; 
    \st \;  (\! \tilde{C}1 \!) \!\!-\!\! (\! \tilde{C}3), (\! \tilde{C}5 \!), (\! \tilde{C}6 \!), (\! \tilde{C}8 \!), (\! \tilde{C}9 \!), 
    \nonumber \vspace{-2mm} \\ 
    & (\! \bar{C}10 \!), (\! C11 \!),(\! C12 \!),(\! \tilde{C}14 \!) \!\!-\!\! (\! \tilde{C}16 \!), 
    (\! C17 \!),(\! C18 \!). \nonumber
\eieq
Obviously, problem $(\mathcal{P}_{1})$ is still non-convex due to the non-convexity of SINR and rate constraints $(\bar{C}10),(C17),(C18)$.

\subsubsection{Convexify SINR and Rate Constraints}
These constraints can be convexified by the following propositions.

\begin{proposition} \label{pro: rate cons}
    Constraints $(C17.1),(C18)$ are convexified as
    \bieq{ll}
        (\tilde{C}17a)\!\!: w_{\sf{M}} F_{\ell,k}^{{\sf{R,M}},[v_{\sf{M}},n_{\sf{M}}]}(\boldsymbol{P}_{\!c}, \boldsymbol{\eta}_{c}) \geq r_{\ell,k}^{{\sf{M}},[v_{\sf{M}},n_{\sf{M}}]}, \nonumber \\
        (\tilde{C}17b)\!\!: \Psi_{\ell,k}^{[v_{\sf{M}},n_{\sf{M}}]} \!(\! \boldsymbol{P}_{\!c} \!) \!+\!  \Theta_{k}^{{\sf{TN}},[v_{\sf{M}},n_{\sf{M}}]} \!(\!\boldsymbol{P}_{\!c} \!) \!+\! \sigma_{{\sf{M}},k}^{2} 
        \! \leq \! F_{\! \sf{exp}}^{(i)}\!(\eta_{\ell,k}^{[v_{\sf{M}},n_{\sf{M}}]} \!) , \nonumber \\
        (\tilde{C}18a)\!\!: w_{\sf{M}} F_{0,k}^{{\sf{R,M}},[v_{\sf{M}},n_{\sf{M}}]}(\boldsymbol{P}_{\!c}, \boldsymbol{\eta}_{c}) \geq r_{0,k}^{{\sf{M}},[v_{\sf{M}},n_{\sf{M}}]}, \nonumber \\
        (\tilde{C}18b)\!\!: \Theta_{k}^{{\sf{Sat,M}},[v_{\sf{M}},n_{\sf{M}}]} (\!\boldsymbol{P}_{\!c}) + \sigma_{{\sf{M}},k}^{2} 
        \! \leq \! \mathcal{F}_{\sf{exp}}^{(i)}(\eta_{0,k}^{[v_{\sf{M}},n_{\sf{M}}]}) , \nonumber \\
        (\tilde{C}18c)\!\!: R_{0,k}^{{\sf{S}},[v_{\sf{S}},n_{\sf{S}}]}(\boldsymbol{P}_{\!c}) \geq R_{0,k}^{{\sf{S}},[v_{\sf{S}},n_{\sf{S}}]}, \nonumber 
    \eieq
    with $\mathcal{F}_{\sf{exp}}^{(i)}(u) \triangleq \exp(u^{(i)})(u - u^{(i)} + 1)$, 
    $F_{\ell,k}^{{\sf{R,M}},[v_{\sf{M}}\!,n_{\sf{M}}]}(\boldsymbol{P}_{\!c}, \boldsymbol{\eta}_{c}) \triangleq \log_2{\!( p_{\ell,k}^{[v_{\sf{M}}\!,n_{\sf{M}}]} h_{\ell,k}^{[v_{\sf{M}}\!,n_{\sf{M}}]} \!+\! \Psi_{\ell,k}^{{\sf{M}},[v_{\sf{M}}\!,n_{\sf{M}}]} (\! \boldsymbol{P}_{\!c} \!)} \!+\!  \Theta_{k}^{{\sf{TN}},[v_{\sf{M}}\!,n_{\sf{M}}]} (\!\boldsymbol{P}_{\!c}) + \sigma_{{\sf{M}},k}^{2}) - \eta_{\ell,k}^{[v_{\sf{M}}\!,n_{\sf{M}}]}/\ln(2)$, and $F_{0,k}^{{\sf{R,M}},[v_{\sf{M}}\!,n_{\sf{M}}]}(\boldsymbol{P}_{\!c}, \boldsymbol{\eta}_{c}) \triangleq \log_2{\!( p_{0,k}^{[v_{\sf{M}}\!,n_{\sf{M}}]} g_{k}^{[v_{\sf{M}}\!,n_{\sf{M}}]} \!+\!  \Theta_{k}^{{\sf{Sat}},[v_{\sf{M}}\!,n_{\sf{M}}]} \!(\!\boldsymbol{P}_{\!c}) + \sigma_{{\sf{M}},k}^{2})} - \eta_{0,k}^{[v_{\sf{M}}\!,n_{\sf{M}}]} / \ln(2)$; and $\boldsymbol{\eta}_{c} = \{\eta_{0,k}^{[v_{\sf{M}}\!,n_{\sf{M}}]},\eta_{\ell,k}^{[v_{\sf{x}},n_{\sf{x}}]} | \forall (\ell, k, v_{\sf{x}},n_{\sf{x}}), {\sf{x}} \in \{\sf{D,M}\}\}$ is a slack variable.
\end{proposition}

\begin{IEEEproof}
Please see appendix~\ref{app: rate cons}.
\end{IEEEproof}

The SINR  constraint $(\bar{C}11)$ and rate constraint $(C17.2)$ for $\sf{D}$ services can be convexified by the following proposition.
\begin{proposition} \label{pro: SINR cons}
Constraint $(\bar{C}10)$ can be convexified as
\bieq{ll}
    (\tilde{C}10a)\!\!: F_{\ell,k}^{{\sf{R,D}},[v_{\sf{D}},n_{\sf{D}}]}(\boldsymbol{P}_{\!c}, \boldsymbol{\eta}_{c}) \geq \mathcal{F}_{\! \sf{apx}}^{(i)}\!(p_{\ell,k}^{[v_{\sf{D}},n_{\sf{D}}]}) \log_2(1+\gamma_{0}^{\sf{D}}) , \nonumber \\
    (\tilde{C}10b)\!\!: \Psi_{\ell,k}^{[v_{\sf{D}},n_{\sf{D}}]} (\boldsymbol{P}_{\!c}) \!+\! \sigma_{{\sf{D}},k}^{2} 
    \! \leq \! \mathcal{F}_{\sf{exp}}^{(i)}(\eta_{\ell,k}^{[v_{\sf{D}},n_{\sf{D}}]}), \nonumber
\eieq
with $F_{\ell,k}^{{\sf{R,D}},[\! v_{\sf{D}},n_{\sf{D}} \!]}\!(\!\boldsymbol{P}_{\! c} ,\! \boldsymbol{\eta}_{c} \!) \!\!=\!  \log_2{\! \big( p_{\ell,k}^{\! [\! v_{\sf{D}},n_{\sf{D}} \!]}  h_{\ell,k}^{\! [\! v_{\sf{D}},n_{\sf{D}} \!]} \!+\! \Psi_{\ell,k}^{\! [\!v_{\sf{D}},n_{\sf{D}} \!]} \!(\! \boldsymbol{P}_{\! c} \!)} + \sigma_{\!{\sf{D}},k}^{2} \big)
    - \eta_{\ell,k}^{\![\! v_{\sf{D}},n_{\sf{D}} \!]} / \ln(2) $.
\end{proposition}
\begin{IEEEproof}
Please see appendix~\ref{app: SINR cons}.
\end{IEEEproof}

\begin{proposition}\label{pro: DS rate cons}
Constraint $(C17.2)$ can be convexified as
\bieq{ll}
    (\tilde{C}17c)\!\!: w_{\sf{D}} \!\!\! \scaleobj{0.8}{\sum_{ \forall (v_{\sf{D}},n_{\sf{D}}) } } \!\!\! F_{\ell,k}^{{\sf{R,D}},[v_{\sf{D}},n_{\sf{D}}]}(\boldsymbol{P}_{\!c}, \boldsymbol{\eta}_{c})
    \!-\!\! \chi_{\sf{D}} \mathcal{F}_{\sf{sqrt}}^{(i)}( \zeta_{\ell,k,[s]}) \!\! \geq \!\! r_{\ell,k,[s]}^{{\sf{SF,D}}} , \nonumber \\
    (\tilde{C}17d)\!\!: \zeta_{\ell,k,[s]} \! \geq \! \scaleobj{0.8}{\sum\nolimits_{\forall (v_{\sf{x}},n_{\sf{x}}) } }  \mathcal{F}_{\! \sf{apx}}^{(i)}\!(p_{\ell,k}^{[v_{\sf{D}},n_{\sf{D}}]}), \nonumber 
\eieq
where $\mathcal{F}_{\sf{sqrt}}^{(i)}(x) \triangleq {x}/{(2\sqrt{x^{(i)}})} + {\sqrt{x^{(i)}}}/{2}$ and $\boldsymbol{\zeta}_{\!c} \triangleq \{\zeta_{\ell,k,[s]} \vert \forall \ell, \forall k \in \mathcal{K}_{\sf{D}}, \forall s \in \mathcal{T}_{\sf{SF}}^{{\sf{Cy}},c} \}$ is a slack variable.
\end{proposition}
\begin{IEEEproof}
See appendix~\ref{app: DS rate cons}.
\end{IEEEproof}

Thanks to proposition~\ref{pro: rate cons}, \ref{pro: SINR cons}, and \ref{pro: DS rate cons} problem $(\mathcal{P}_{1})_{\!c}$ is transformed into the iterative convex problem $(\mathcal{P}_{2})_{\!c}$ as
\bieq{lll} \label{eq: Prob2}
    (\mathcal{P}_{2})_{\!c}: \; \min_{\boldsymbol{\omega}_{\!c},\boldsymbol{P}_{\!c}, \bar{\boldsymbol{q}}_{\!c}, \boldsymbol{r}_{\!c}, \boldsymbol{\eta}_{c}, \boldsymbol{\zeta}_{\!c} } \!\!\! f_{\sf{obj}}(\bar{\boldsymbol{q}}_{\!c}) \;\;
    \st \; (\tilde{C}1)-(\tilde{C}3),(\tilde{C}5),(\tilde{C}6), \nonumber \\ 
    \hspace{12mm} (\tilde{C}8)-(\tilde{C}10), (C11),(C12),(\tilde{C}13)-(\tilde{C}18), \nonumber
\eieq
Using outcome of solving $(\mathcal{P}_{2})_{\!c}$, $\boldsymbol{\alpha}_{\!c}$ and $\boldsymbol{\beta}_{\!c}$ are recovered as
\bieq{ll} \label{eq: recover binary}
    \!\!\! \alpha_{\ell,k}^{[v_{\sf{x}},n_{\sf{x}}]} \!=\! 1 \text{ if } p_{\ell,k}^{[v_{\sf{x}},n_{\sf{x}}]} \! \geq \! \epsilon_{\sf{rec}} \text{; } \alpha_{\ell,k}^{[v_{\sf{x}},n_{\sf{x}}]} \!=\! 0 \text{ otherwise, } \quad \subnum \\
    \!\!\! \beta_{k}^{[v_{\sf{x}},n_{\sf{x}}]} = 1 \text{ if } p_{0,k}^{[v_{\sf{x}},n_{\sf{x}}]} \geq \epsilon_{\sf{rec}} \text{; } \beta_{k}^{[v_{\sf{x}},n_{\sf{x}}]}=0 \text{ otherwise. } \subnum
\eieq
With an appropriately chosen $\epsilon_{\sf{rec}}$, only the links with sufficiently large transmit power remain active.
Besides, the removal of unselected links--with negligible transmit powers--has an insignificant impact on the aggregated rate. Consequently, the feasibility of the constraints is preserved.
The BWA variable $\boldsymbol{b}_{\!c}$ is recovered by \eqref{eq: BW by norm0}. The proposed solution \textit{``DT-JointRA''} to solve problem $(\mathcal{P}_{0})_{\!c}$ using predicted information $\{\hat{\boldsymbol{h}}_{\!c}, \hat{\boldsymbol{g}}_{\!c}, \hat{\boldsymbol{\lambda}}_{\!c}\}$ is summarized in Alg.~\ref{alg: DTJRA}.

\begin{proposition} \label{pro: convergence DTJRA}
    Alg.~\ref{alg: DTJRA} is guaranteed to converge to a local optimal solution of problem $(\mathcal{P}_{0})_{\!c}$.
\end{proposition}

\begin{IEEEproof}
    Proposition~\ref{pro: convergence DTJRA} can be proved similarly as \textit{proposition 4} in \cite{Hung_TCOM25}. Particularly, 
    due to the properties of the SCA-based problem \cite{SCA} and limited transmit power, the generated objective value sequence by solving $(\mathcal{P}_{2})_{\!c}$ is non-increasing and bounded, resulting in Alg.~\ref{alg: DTJRA}'s convergence. In additional, by recovering binary variables, the feasible set of $(\mathcal{P}_{2})_{\!c}$ is a subset of that of $(\mathcal{P}_{0})_{\!c}$. Hence, Alg.~\ref{alg: DTJRA} converges to a local optimal solution of $(\mathcal{P}_{0})_{\!c}$.
\end{IEEEproof}

\vspace{-3mm}

\subsection{RT-Refine Solution}
Problems $(\mathcal{P}_{0}^{\sf{refi}})_{s}$ and $(\mathcal{P}_{0})_{\!c}$ share a similar structure with certain omitted constraints and variables in $(\mathcal{P}_{0}^{\sf{refi}})_{s}$. Hence, by using approximation steps for $(\mathcal{P}_{0})_{\!c}$, $(\mathcal{P}_{0}^{\sf{refi}})_{s}$ is transformed into iterative convex problem $(\mathcal{P}_{1}^{\sf{refi}})_{s}$ as
\begin{equation}
\begin{array}{lll} \label{eq: Prob1Refine}
    \!\!\!(\mathcal{P}_{1}^{\sf{refi}})_{s} \!\! : \min\limits_{\boldsymbol{p}[s], \boldsymbol{r}[s], \boldsymbol{\eta}[s], \bar{\boldsymbol{q}}[s] } \;f_{\sf{obj}}(\bar{\boldsymbol{q}}[s]) \;
    \hspace{1mm} \st \; (\tilde{C}5)_{s}, \!(\tilde{C}6)_{s}, \!(\tilde{C}9)_{s}, \nonumber\\
    \!(\tilde{C}10)_{s},  \! (C11)_{s}, \!(\tilde{C}14)_{s}, \!(\tilde{C}15)_{s} \!, \!(\tilde{C}16)_{s}, \!(\tilde{C}17)'_{s}, \!(\tilde{C}18a,b)'_{s} \!, \nonumber
\end{array}
\end{equation}
where $(\tilde{C}17)'$, $(\tilde{C}18a)'$, and $(\tilde{C}18b)'$ are the convex constraint revised from $(\tilde{C}17)$, $(\tilde{C}18a)$, and $(\tilde{C}18b)$ using terms $\hat{\Theta}_{k}^{{\sf{TN}},[v_{\sf{M}},n_{\sf{M}}]}$ and $\hat{\Theta}_{k}^{{\sf{Sat}},[v_{\sf{M}},n_{\sf{M}}]}$ instead of $\Theta_{k}^{{\sf{TN}},[v_{\sf{M}},n_{\sf{M}}]}$ and $\Theta_{k}^{{\sf{Sat}},[v_{\sf{M}},n_{\sf{M}}]}$, respectively. It is worth noting that information  $\{\boldsymbol{h}_{e}^{*}, \hat{\boldsymbol{g}}_{e}, \boldsymbol{\lambda}_{e}\}$ is used to solve $(\mathcal{P}_{1}^{\sf{refi}})_{s}$.
The overall solution of RT-Refine Algorithm, is described in Alg.~\ref{alg: RTRefine}, with the workflow illustrated in Fig.~\ref{fig:flowchart_RTRefine}.

\begin{proposition} \label{pro: convergence RTRefine}
    Alg.~\ref{alg: RTRefine} converges to a local optimal solution of problem sequence $\{ (\mathcal{P}_{0}^{\sf{refi}})_{s} \}_{\forall s}$.
\end{proposition}

\begin{IEEEproof}
    For each cycle $c$, Alg.~\ref{alg: RTRefine} solves $(\mathcal{P}_{2})_{\!c}$ and $(\mathcal{P}_{1}^{\sf{refi}})_{s}$. Solving $(\mathcal{P}_{2})_{\!c}$ and $(\mathcal{P}_{1}^{\sf{refi}})_{s}$ are ensured to converge since they are the SCA-based problems \cite{SCA}. Besides, for each loop of SF $s$ (steps $8-16$), the $(\mathcal{P}_{1}^{\sf{refi}})_{s}$'s feasible set is a subset of that of $(\mathcal{P}_{0}^{\sf{refi}})_{s}$. Hence, Alg.~\ref{alg: RTRefine} converges to a local optimal solution of $\{ (\mathcal{P}_{0}^{\sf{refi}})_{s} \}_{\forall s}$.
\end{IEEEproof}

\begin{algorithm}[!t]
\scriptsize
\begin{algorithmic}[1]
    \captionsetup{font=footnotesize}
    \protect\caption{\
    \textsc{DT-based Joint-RA Algorithm (DT-JointRA)}}
    \label{alg: DTJRA}
    \STATE \textbf{Input:} Predicted $\{ \hat{\boldsymbol{h}}_{\!c}, \hat{\boldsymbol{g}}_{\!c},\hat{\boldsymbol{\lambda}}_{\!c} \}$ from DT.
    \STATE Set $i=1$ and generate an initial point $(\boldsymbol{P}_{\!c}^{(0)}, \boldsymbol{\eta}_{c}^{(0)}, \boldsymbol{\zeta}_{\!c}^{(0)})$.\\
    \REPEAT
    \STATE Solve problem $(\mathcal{P}_{2})_{\!c}$ to obtain $(\boldsymbol{P}_{\!c}^{\star}, \boldsymbol{\eta}_{c}^{\star}, \boldsymbol{\zeta}_{\!c}^{\star} )$.
    \STATE Update $ (\boldsymbol{P}_{\!c}^{(i)}, \boldsymbol{\eta}_{c}^{(i)}, \boldsymbol{\zeta}_{\!c}^{(i)})=( \boldsymbol{P}_{\!c}^{\star}, \boldsymbol{\eta}_{c}^{\star}, \boldsymbol{\zeta}_{\!c}^{\star})$ and $i:=i+1$.
    \UNTIL Convergence
    \STATE Recovery $\boldsymbol{\omega}_{\!c}^{\star}$ and binary solutions $\boldsymbol{b}_{\!c}^{\star}$, $\boldsymbol{\alpha}_{\!c}^{\star}$, and $\boldsymbol{\beta}_{\!c}^{\star}$  by \eqref{eq: transform flow-split}, \eqref{eq: BW by norm0}, and \eqref{eq: recover binary}.
    \STATE \textbf{Output:} Solution $\{\boldsymbol{b}_{\!c}^{\star},\boldsymbol{\omega}_{\!c}^{\star},\boldsymbol{P}_{\!c}^{\star}, \boldsymbol{\alpha}_{\!c}^{\star}, \boldsymbol{\beta}_{\!c}^{\star}, \boldsymbol{\eta}_{c}^{\star}\}$ for cycle $c$.
\end{algorithmic}
\normalsize 
\end{algorithm}

\begin{figure}[!t]
    \vspace{-3mm}
    \centering
    \includegraphics[width=1\linewidth]{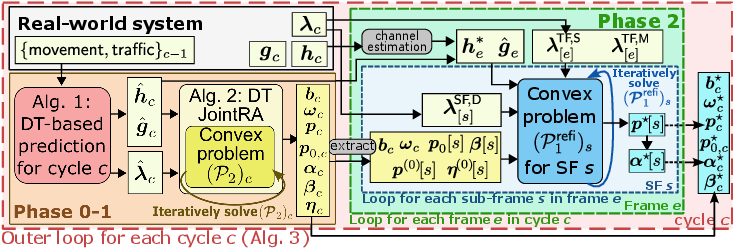}
    \vspace{-6mm}
    \captionsetup{font=footnotesize}
    \caption{Flowchart of the proposed DT-based algorithm.}
    \label{fig:flowchart_RTRefine}
    \vspace{-3mm}
\end{figure}

\subsection{RT-Refine Algorithm Implementation}
\vspace{-1mm}
\begin{figure}[!t]
  \captionsetup{font=footnotesize}
    \includegraphics[width=85mm]{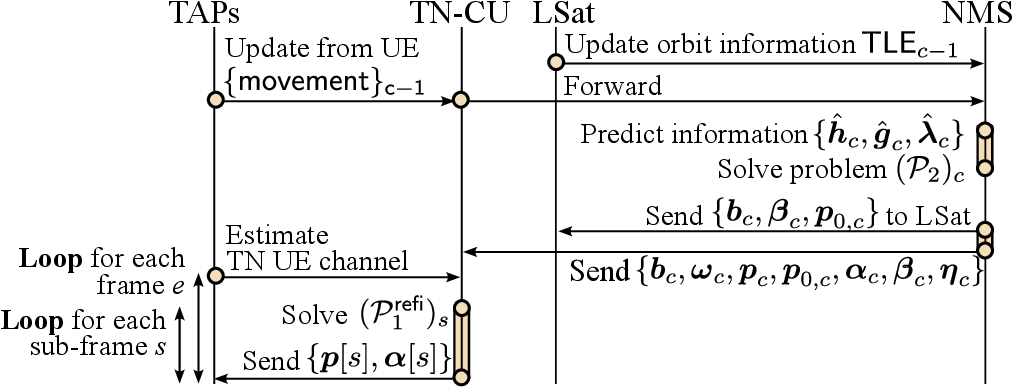}
    \vspace{-2mm}
    \caption{Data exchange in DT-based algorithm.}
    \label{fig:RTRefine_signaling}
    \vspace{-4mm}
\end{figure}

Regarding implementation, Fig.~\ref{fig:RTRefine_signaling} shows execution steps and data exchange among TAPs, TCU, LSat, and NMS in the RT-Refine algorithm. The procedure is summarized as
\begin{itemize}
    \item At each cycle, TAPs and LSat send the updated UE movement and orbit information to NMS. Based on which, the necessary information is predicted and Alg.~\ref{alg: DTJRA} is executed to pre-solve $(\mathcal{P}_{2})_{\!c}$ for upcoming cycle $c$.
    \item Once $(\mathcal{P}_{2})_{\!c}$ is solved, all outcomes are sent to TCU for the next phase, while only SatNet's BWA, power control, and UE/RB assignment solutions are sent to LSat.
    \item For the given solution sent from NMS, TCU refines TN short-term decisions at each SF $s$. Subsequently, solutions are sent to TAPs.
\end{itemize}
One can see that, at each cycle, there is only one round-trip communication between LSat and ground segments, i.e., LSat updates orbit information and then NMS sends transmission solutions. Additionally, depending on the stability and accuracy of orbit information (TLE data), LSat can update less frequently, which further reduces costly signaling.
The complexity aspect is further discussed in the following section.

\begin{remark}
    In practical implementations, resource management tasks must be completed within a finite time horizon, which may vary depending on the specific system. For simplicity, we assume that the computational capacities of the NMS and TCU are sufficient to meet these time constraints.
\end{remark}

\begin{algorithm}[!t]
\scriptsize
\begin{algorithmic}[1]
    \captionsetup{font=footnotesize}
    \protect\caption{\
    \textsc{Real-time-Refine Algorithm (RT-Refine)}}
    \label{alg: RTRefine}
    \STATE \textbf{Initialize:} Actual information $\{ \boldsymbol{\lambda}_{0}, \boldsymbol{u}_{0}^{\sf{ue}}, {\sf{TLE}}_{0}\}$ in cycle $0$. 
    \FOR{Each cycle $c$, $c=1\rightarrow N_{\sf{Cy}}$}
    \STATE \textbf{Phase 0:} Predict $\{ \hat{\boldsymbol{h}}_{\!c}, \hat{\boldsymbol{g}}_{\!c},\hat{\boldsymbol{\lambda}}_{\!c} \}$ by Alg.~\ref{alg: DTPred}.
    \STATE \textbf{Phase 1:} Execute Alg.~\ref{alg: DTJRA} for solution $\{\boldsymbol{b}_{\!c}^{\star},\boldsymbol{\omega}_{\!c}^{\star},\boldsymbol{P}_{\!c}^{\star}, \boldsymbol{\alpha}_{\!c}^{\star}, \boldsymbol{\beta}_{\!c}^{\star}, \boldsymbol{\eta}_{c}^{\star} \}$ for cycle $c$..
    \STATE \textbf{Phase 2: Re-optimize} using $\boldsymbol{h}_{e}^{*}, \! \hat{\boldsymbol{g}}_{e},\!\boldsymbol{\lambda}_{e}$ and initial point from output in \textbf{step 4}.
    \FOR{Each frame $e$ in cycle $c$}
        \STATE Collect $\{ \boldsymbol{\lambda}_{e}^{\sf{TF,x}} \}, {\sf{x}} \in \scaleobj{0.9}{\sf{\{M,S\}}}$. Estimate channel of TN UEs and build $\boldsymbol{h}_{e}^{*}$.
        \FOR{Each SF $s$ in frame $e$}
        \STATE Collect $ \boldsymbol{\lambda}_{s}^{\sf{SF,D}} $, set i=0 and extract initial point from solution in \textbf{step 4}.
            \REPEAT
                \STATE Solve problem $(\mathcal{P}_{1}^{\sf{refi}})_{s}$ to obtain $(\boldsymbol{p}^{\star}[s], \boldsymbol{\eta}^{\star}[s])$.
                \STATE Update $ (\boldsymbol{p}^{(i)}[s], \boldsymbol{\eta}^{(i)}[s]) = (\boldsymbol{p}^{\star}[s], \boldsymbol{\eta}^{\star}[s])$ and $i:=i+1$.
            \UNTIL Convergence
            \STATE Recover binary variable $\boldsymbol{\alpha}^{\star}[s]$ by \eqref{eq: recover binary}.
            \STATE \textbf{Output SF $s$:} The adjusted solution  $(\boldsymbol{p}^{\star}[s], \boldsymbol{\alpha}^{\star}[s])$.
        \ENDFOR
    \ENDFOR
    \STATE \textbf{Output cycle $c$:} Solution $\{\boldsymbol{b}_{\!c}^{\star},\boldsymbol{\omega}_{\!c}^{\star},\boldsymbol{p}_{\!c}^{\star}, \boldsymbol{p}_{\!0,c}^{\star}, \boldsymbol{\alpha}_{\!c}^{\star}, \boldsymbol{\beta}_{\!c}^{\star} \}$ with adjusted $(\boldsymbol{p}_{\!c}^{\star}, \boldsymbol{\alpha}_{\!c}^{\star})$.
    \ENDFOR
    \STATE \textbf{Output:} Solution for $N_{\sf{cy}}$  cycles.
\end{algorithmic}
\normalsize 
\end{algorithm}

\begin{table}[]
    \setlength{\tabcolsep}{2pt}
    \captionsetup{font=footnotesize}
    \centering
    \caption{Comparison of examined schemes.}
    \label{tab:ExaminedSchemes}
    \vspace{-2mm}
    \scalebox{0.72}{
    \begin{tabular}{|l|c|l|l|c|c|}
    \hline
        \textbf{Scheme} & \makecell{\textbf{DT}\\ \textbf{system}} & \textbf{CSI information} & \textbf{Traffic information} & \makecell{\textbf{Interference}\\\textbf{aware}} & \makecell{\textbf{Convergence}\\\textbf{acceleration}} \\
        \hline
        \textbf{Full-info Alg.} & \multicolumn{3}{c|}{Full actual network state knowledge} & \cmark & \xmark \\
        \hline
        \textbf{RT-Refine Alg.} & \cmark & & & \cmark & \cmark \\
        \multicolumn{1}{|r|}{Phase 1} & \cmark & Cycle-Predicted CSI & Cycle-Predicted traffic & \cmark & \\
        \multicolumn{1}{|r|}{Phase 2} & \cmark & Subframe-Actual CSI & Subfr./frame-Actual traffic & \cmark & \cmark \\
        \hline 
        \textbf{SF-RA Alg.} & \xmark & Subframe-Actual CSI & Cycle-Actual traffic & \cmark & \xmark \\
        \hline
        \makecell{\textbf{Ref./Heuristic/}\\ \textbf{Greedy Algs.}} & \xmark & Subframe-Actual CSI & Cycle-Actual traffic & \xmark & \xmark \\
        \hline
    \end{tabular}}
\end{table}

\vspace{-5mm}

\section{Benchmarks and Complexity Analysis} \label{sec: benchmark-complexity}
\vspace{-1mm}
This section provides several bechmarks for comparison purposes. The differences across schemes are summarized in Table~\ref{tab:ExaminedSchemes}. Due to space limit, we describe the overall big-O complexity instead of a detailed analysis where those of solving convex problems are derrived by \cite{complexity}. 

\vspace{-5mm}
\subsection{Greedy Algorithm}
\vspace{-1mm}
The \textbf{Greedy Algorithm} is summarized in Alg.~\ref{alg: greedy}. The number of used SCs for $\sf{D}$ services is fixed to $N_{\sf{SC}}^{\sf{D}}$ while BW for the $\sf{M}$ and $\sf{S}$ services is uniformly allocated. The AP/LSat-UE association and RB assignment are selected based on channel gain matrices in cycle $c$. Afterward, the water-filling algorithm is utilized for power allocation by ignoring interference with $\epsilon_{\sf{pow}}$ tolerance for water level searching. 
The traffic steering decision is set proportional to the UE rate. One notes that this algorithm does not guarantee SINR/rate and QL constraints $(C10),(C14)-(C16)$.

\begin{table}[]
    \captionsetup{font=footnotesize}
    \centering
    \caption{QL and remaining traffic $\sf{D}$ of Heuristic algorithm versus $N_{\sf{SC}}^{\sf{D}}$.}
    \vspace{-2mm}
    \scalebox{0.8}{
    \begin{tabular}{|l|c|c|c|c|c|}
    \hline
         Num. SCs for $\sf{D}$ services, $N_{\sf{SC}}^{\sf{D}}$ & 1 & 2& 3& 4&5 \\ \hline
         Mean QL (MB) & 20.12 & 23.12 & 27.08 & 31.55 & 37.29 \\
         Remaining $\sf{D}$ traffic (\%) & 22.2 & 19.5 & 18.4 & 17.8 & 17.4 \\ \hline
    \end{tabular}}
    \label{tab:HeuAlg_Nsb}
\end{table}

\vspace{-4mm}
\subsection{Heuristic Algorithm}
\vspace{-1mm}
A \textbf{Heuristic Algorithm} is further proposed by updating the greedy algorithm above, where the BWA is adopted from \cite{kavehmadavani_TWC23}. 
In particular, the approach is developed from Alg.~\ref{alg: greedy} as: \textbf{step 1)} traffic $\boldsymbol{\lambda}_{\!c}$ is added to input; \textbf{step 4)} BW for the $\sf{M}$ and $\sf{S}$ services is set proportional to the remaining QL in cycle $(c-1)$; \textbf{step 19)} after power allocation, the remaining QL is computed based on the rate of UE and traffic.
The complexity of the greedy (Alg.~\ref{alg: greedy}) and heuristic algorithms is
\vspace{-2mm}
\begin{equation}
    {\sf{X}}_{\sf{Grd-Heu}} = \mathcal{O}(N_{\sf{Cy}} (X_{\sf{Assoc}} -N\log_{2}(\epsilon_{\sf{pow}}))).
    \vspace{-3mm}
\end{equation}
\vspace{-8mm}
\subsection{Other Benchmark}
\vspace{-1mm}
Given the similar aspects with our work, the proposed algorithm in \cite{kavehmadavani_TWC23} is reused with modifications to adapt to our considered problem. Particularly, for the long-term decisions, BWA for BWPs is set propostional to the next cycle service traffics while the traffic steering for each node is set propotional to the previous cycle sum-rate of that node. For the short-term decision, power control and RB assignment are jointly optimized wherein the co-channel interference is ignored \cite{kavehmadavani_TWC23}. The adapted algorihm using full actual information is called as \textbf{Reference Algorithm} with a complexity of
\vspace{-1mm} 
\begin{equation} \label{eq: complexity RefAlg}
    \hspace{-2mm} X_{\sf{RefAlg}} = \mathcal{O}(N_{\sf{Cy}} N_{\sf{iter}} a_{\sf{DTJRA}}^2  (a_{\sf{RefAlg}} + b_{\sf{RefAlg}}) b_{\sf{RefAlg}}^{1/2} ), 
    \vspace{-2mm} 
\end{equation}
where $a_{\sf{RefAlg}} = N_{\sf{SF}}( L( K_{\sf{D}}( 8\bar{V}_{\sf{D}} + 2) 
+ K_{\sf{M}}( 6 \bar{V}_{\sf{M}} + 2)) + K_{\sf{M}}(6 \bar{V}_{\sf{M}} + 2) 
+ K_{\sf{S}}(3 \bar{V}_{\sf{S}} + 1))$, 
$b_{\sf{RefAlg}} = N_{\sf{SF}} \big( L( K_{\sf{D}}( 8 \bar{V}_{\sf{D}} + 3 ) + K_{\sf{M}}(\bar{V}_{\sf{M}} + 1) 
+ 2(2\bar{V}_{\sf{D}} + \bar{V}_{\sf{M}} + 3 ) )
+ 2( 2K_{\sf{D}}\bar{V}_{\sf{D}} + K_{\sf{M}}(3 \bar{V}_{\sf{M}} + 2)) 
+ 2\bar{V}_{\sf{M}} + K_{\sf{S}} (\bar{V}_{\sf{S}} + 2) + \bar{V}_{\sf{S}} + 5 \big)$, and
$N_{\sf{iter}}$ is the number of iterations for convergence.


Furthermore, one introduces an interference-aware benchmark for the fairness comparisons. In particular, the BWA and traffic steering are adopted from \textbf{Reference Algorithm} over cycles. The RA for each sub-frame is obtained by solving an iterative algorithm which is adapted from $(\mathcal{P}_{2})_{c}$ for one sub-frame with the fixed BWA and traffic steering. This benchmark is called \textbf{Sub-frame-RA (SF-RA) Algorithm}. Since the SF-RA algorithm and RT-Refine phase 2 execute a similar problem for one sub-frame, their complexities in terms of big-O are equivalent.

In addition, a \textbf{Full Information Algorithm (FIA)} is further examined where DT-JointRA algorithm is executed with full actual information.

\begin{algorithm}[t]
\scriptsize
\begin{algorithmic}[1]
     \captionsetup{font=small}
    \protect\caption{\textsc{Greedy Algorithm}}
    \label{alg: greedy}
    \STATE \textbf{Input:} Channel $ \boldsymbol{h}_{\!c} \!=\! [\!h^{\sf{x}}(\ell ,k,v_{\sf{x}},n_{\sf{x}})\!]$, $ \boldsymbol{g}_{\!c}\!=\![\!g^{\sf{x}}(k,v_{\sf{x}},n_{\sf{x}})\!]$, ${\forall (\ell ,k,v_{\sf{x}},n_{\sf{x}},{\sf{x}}, c)}$.
    \STATE \textbf{Initialize:} Zero matrices $\boldsymbol{\alpha}_{\!c}$, $\boldsymbol{\beta}_{\!c}$, and $\boldsymbol{b}_{\!c}$ with $\forall c$.
    \FOR{Each cycle $c$}
    \STATE Use $N_{\sf{sb}}^{\sf{d}}$ $\sf{D}$ SCs, uniformly allocate BW for services $\sf{x} \in \{M,S\}$. Build $\boldsymbol{b}_{\!c}$.
    \FOR{Each frame $e$}
    \STATE Extract $\{ \! \boldsymbol{h}_{e} \!, \boldsymbol{g}_{e} \!, \boldsymbol{\alpha}_{e} \!, \boldsymbol{\beta}_{e} \! \}$ for frame $e$. Set elements of non-using SCs to zero.
    \FOR{Each UE $k$ service $\sf{x}$}
        \WHILE{$\boldsymbol{h}_{e}^{\sf{x}}(:,k,:,:) \neq \mathbf{0}$}
            \STATE Find index $(\hat{\ell},\hat{v},\hat{n})$ where $\boldsymbol{h}_{e}^{\sf{x}}(\hat{\ell},k,\hat{v},\hat{n}) = {\sf{max}}( \boldsymbol{h}_{e}^{\sf{x}}(:,k,:,:))$.
            \STATE Set $\boldsymbol{\alpha}_{e}^{\sf{x}}(\hat{\ell},k,\hat{v},\hat{n}) = 1$, $\boldsymbol{h}_{e}^{\sf{x}}(:,k,\hat{v},\hat{n}) = 0$ and $\boldsymbol{h}_{e}^{\sf{x}}(\hat{\ell},:,\hat{v},\hat{n}) = 0$.
        \ENDWHILE
        \WHILE{$\boldsymbol{g}_{e}^{\sf{x}}(k,:,:) \neq \mathbf{0}$}
            \STATE Find index $(\hat{v},\hat{n})$ satisfying $\boldsymbol{g}_{e}^{\sf{x}}(k,\hat{v},\hat{n}) = {\sf{max}}( \boldsymbol{g}_{e}^{\sf{x}}(k,:,:))$.
            \STATE Set $\boldsymbol{\beta}_{e}^{\sf{x}}(k,\hat{v},\hat{n}) = 1$ and $\boldsymbol{g}_{e}^{\sf{x}}(:,\hat{v},\hat{n}) = 0$.
        \ENDWHILE
    \ENDFOR
    \ENDFOR
    \STATE Assign $\boldsymbol{\alpha_{e}}, \boldsymbol{\beta_{e}}$ to $\boldsymbol{\alpha_{\!c}}, \boldsymbol{\beta_{\!c}}$
    \STATE \textbf{Power control $\{ \boldsymbol{P}_{\!c} \}$:} using water-filling scheme by ignoring interference.
    \STATE \textbf{Traffic steering $\boldsymbol{\omega}_{\!c}$:} set proportional to UE rate.
    \ENDFOR
    \STATE \textbf{Output:} Solution $\{\boldsymbol{b}_{\!c},\boldsymbol{\omega}_{\!c},\boldsymbol{P}_{\!c}, \boldsymbol{\alpha}_{\!c}, \boldsymbol{\beta}_{\!c} \}, \forall c$ for $N_{\sf{Cy}}$  cycles.
\end{algorithmic} 
\normalsize
\end{algorithm}

\vspace{-2mm}
\subsection{Proposed Algorithm Complexity}
\vspace{-1mm}
According to Alg.~\ref{alg: DTJRA} and , the complexities of DT-JointRA and RT-Refine algorithms are expressed respectively as \cite{complexity}
\vspace{-2mm}
\begin{IEEEeqnarray}{ll} 
    \hspace{-2mm} X_{\scaleobj{0.8}{\sf{DTJRA}}} \!=\! \mathcal{O}(N_{\sf{Cy}} N_{\sf{iter}} a_{\scaleobj{0.8}{\sf{DTJRA}}}^2  (a_{\scaleobj{0.8}{\sf{DTJRA}}} + b_{\scaleobj{0.8}{\sf{DTJRA}}}) b_{\scaleobj{0.8}{\sf{DTJRA}}}^{1/2} ), \! \label{eq: complexity DTJRA} \\
    \hspace{-2mm} X_{\scaleobj{0.8}{\sf{RTRefi}}} \!=\! X_{\scaleobj{0.8}{\sf{DTJRA}}} \!+\!\! \mathcal{O} \! (N_{\scaleobj{0.8}{\sf{Cy}}} N_{\scaleobj{0.8}{\sf{SF}}} \bar{N}_{\scaleobj{0.8}{\sf{iter}}} a_{\scaleobj{0.8}{\sf{RTRefi}}}^2  (a_{\scaleobj{0.8}{\sf{RTRefi}}} \!\!+\!\! b_{\scaleobj{0.8}{\sf{RTRefi}}}) b_{\scaleobj{0.8}{\sf{RTRefi}}}^{1/2} ), \quad \label{eq: complexity RTRefine}
\end{IEEEeqnarray}
where $N_{\sf{iter}}$ and $\bar{N}_{\sf{iter}}$ are the iteration  numbers for convergence in solving $(\mathcal{P}_{2})_{c}$ and $(\mathcal{P}_{1}^{\sf{refi}})_{s}$. $a_{\sf{DTJRA}} = N_{\sf{SF}}( L( K_{\sf{D}}( 8\bar{V}_{\sf{D}} + 2) 
+ K_{\sf{M}}( 6 \bar{V}_{\sf{M}} + 2)) + K_{\sf{M}}(6 \bar{V}_{\sf{M}} + 2) 
+ K_{\sf{S}}(3 \bar{V}_{\sf{S}} + 1)) + L(\bar{V}_{\sf{D}} + \bar{V}_{\sf{M}})$, $b_{\sf{DTJRA}} = N_{\sf{SF}} \big( L( K_{\sf{D}}( 8 \bar{V}_{\sf{D}} + 3 ) + K_{\sf{M}}(\bar{V}_{\sf{M}} + 1) 
+ 2(2\bar{V}_{\sf{D}} + \bar{V}_{\sf{M}} + 3 ) )
+ 2( 2K_{\sf{D}}\bar{V}_{\sf{D}} + K_{\sf{M}}(3 \bar{V}_{\sf{M}} + 2)) 
+ 2\bar{V}_{\sf{M}} + K_{\sf{S}} (\bar{V}_{\sf{S}} + 2) + \bar{V}_{\sf{S}} + 5 \big)
+ \bar{V}_{\sf{D}} + \bar{V}_{\sf{M}} + 1$,  
$a_{\sf{RTRefi}} = L( K_{\sf{D}}( 8\bar{V}_{\sf{D}} + 2) 
+ K_{\sf{M}}( 6 \bar{V}_{\sf{M}} + 2)) + K_{\sf{M}}(4 \bar{V}_{\sf{M}} + 2)$, $b_{\sf{RTRefi}} =  L( K_{\sf{D}}( 8 \bar{V}_{\sf{D}} + 3 ) + K_{\sf{M}}(\bar{V}_{\sf{M}} + 1) 
+ 2(2\bar{V}_{\sf{D}} + \bar{V}_{\sf{M}} + 3 ) )
+ 2( 2K_{\sf{D}}\bar{V}_{\sf{D}} + K_{\sf{M}}(3 \bar{V}_{\sf{M}} + 2) + 1) $.

\section{Numerical Results} \label{sec: result}
\subsection{Simulation Setup}
The simulation is conducted with 3D map in an area of $3$km$^2$ at $(51.524^{\circ}\text{N},0.085^{\circ}\text{W})$ in London city. For simulation data, the vehicular UE route and LSat orbit are taken from Google Navigator and Starlink LEO TLE, the arrival traffic is taken from an actual dataset \cite{Bonati_2024_dataset}. Environment channels are generated using the DT model with coefficient $\xi=0.5$.

The key simulation parameters are: system BW $W^{\sf{tot}}=15$~MHz, operation frequency $f_{\!c}=3.4$~GHz, LSat altitude $500$~km, antenna parameters of AP, LSat, SUE as in \cite{3gpp.38.863, 3gpp.38.821}, vehicle UE antenna and RayT parameters as in \cite{Hung_TCOM25, Hung_VTC24}, 
number of  cycles $N_{\sf{Cy}}=20$, number of frames/cycle $N_{\sf{TF}}=5$,
numbers of TAPs, UEs $(N,K_{\sf{D}}, K_{\sf{M}}, K_{\sf{S}}) = (6, 4, 5, 3)$, power budget at AP and LSat $(p_{\sf{AP}}^{\sf{max}}, p_{\sf{Lsat}}^{\sf{max}}) = (34,36)$~dBm, maximum QL $q_{\ell}^{\sf{m,max}}=q_{0}^{\sf{x,max}}=2$~MB, interference margin for the RT-Refine algorithm $\kappa=1.1$. The simulations are conducted on a computer with an \textit{Intel Xeon E5 CPU @ 2.4~GHz} and \textit{128~GB RAM}, wherein the iterative algorithms are implemented by using the CVX modeling framework and MOSEK solver in the MATLAB programming environment.
In practical systems, the proposed schemes can be implemented through hardware-based platforms, such as FPGA accelerators \cite{thesis_FPGA_IPmethod}, as well as parallelized processing techniques \cite{MOSEK_guideline}, thereby reducing computational overhead and improving execution efficiency. 
The numerical results for each parameter scenario are obtained via Monte Carlo simulations over $50$ independent realizations.
For intuition, the simulation scenario with UE position and channel gain heatmap, and RayT result are shown in Fig.~\ref{fig:heatmap_withUE}. 
For heuristic algorithms, the mean QL and remaining $\sf{D}$ traffic versus $N_{\sf{SC}}^{\sf{D}}$ are shown in Table.~\ref{tab:HeuAlg_Nsb}. 
To minimize the congestion, we select $N_{\sf{SC}}^{\sf{D}}=1$.

\begin{figure}[]
    \vspace{-4mm}
    \captionsetup{font=footnotesize}
    \centering
    \begin{subfigure}{0.9\columnwidth}
        \hspace{-2mm}
        \includegraphics[width=85mm]{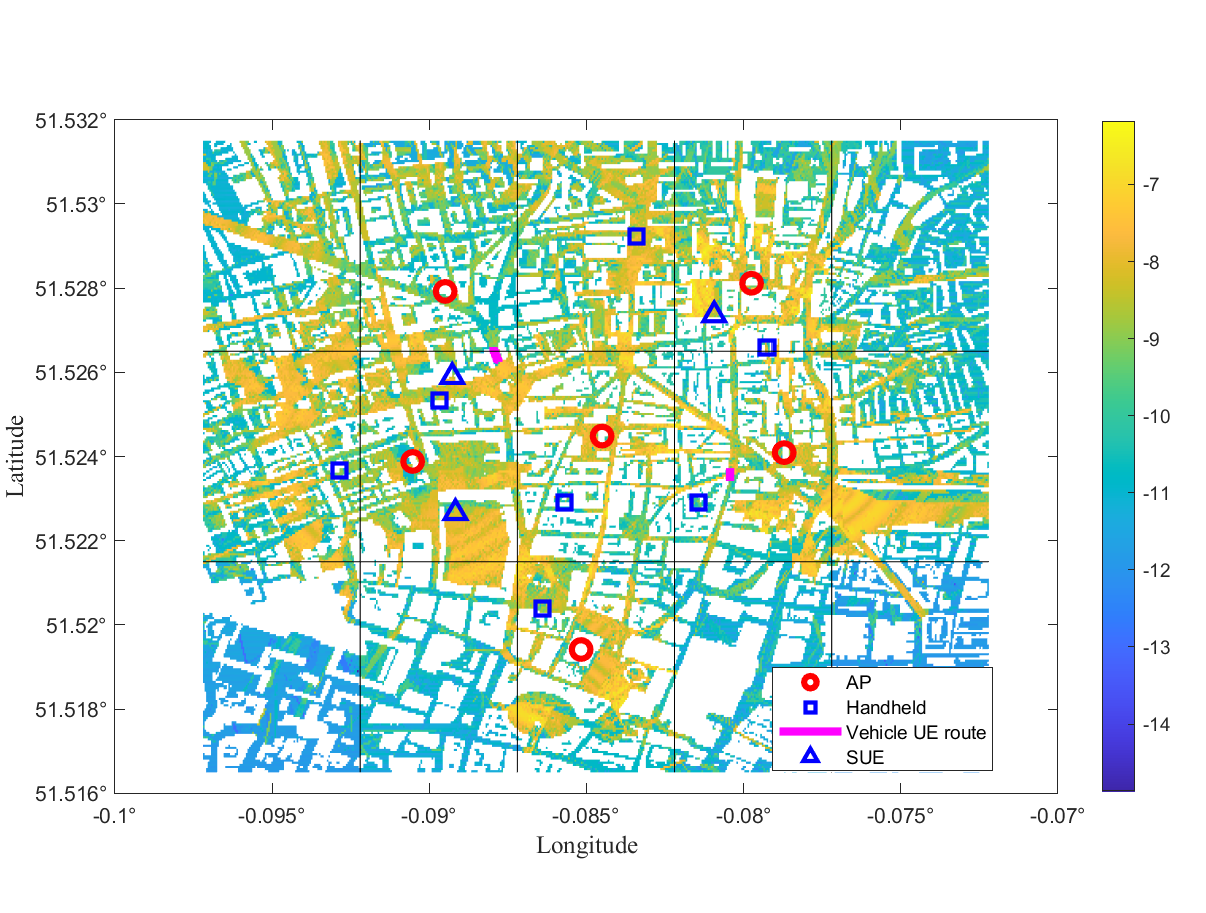}
        \vspace{-10mm}
        \caption{Simulation scenario and TAP channel gain heatmap (log10 scale).}
    \end{subfigure}
    \begin{subfigure}{0.8\columnwidth}
        \vspace{2mm}
        \centering
        \includegraphics[width=65mm]{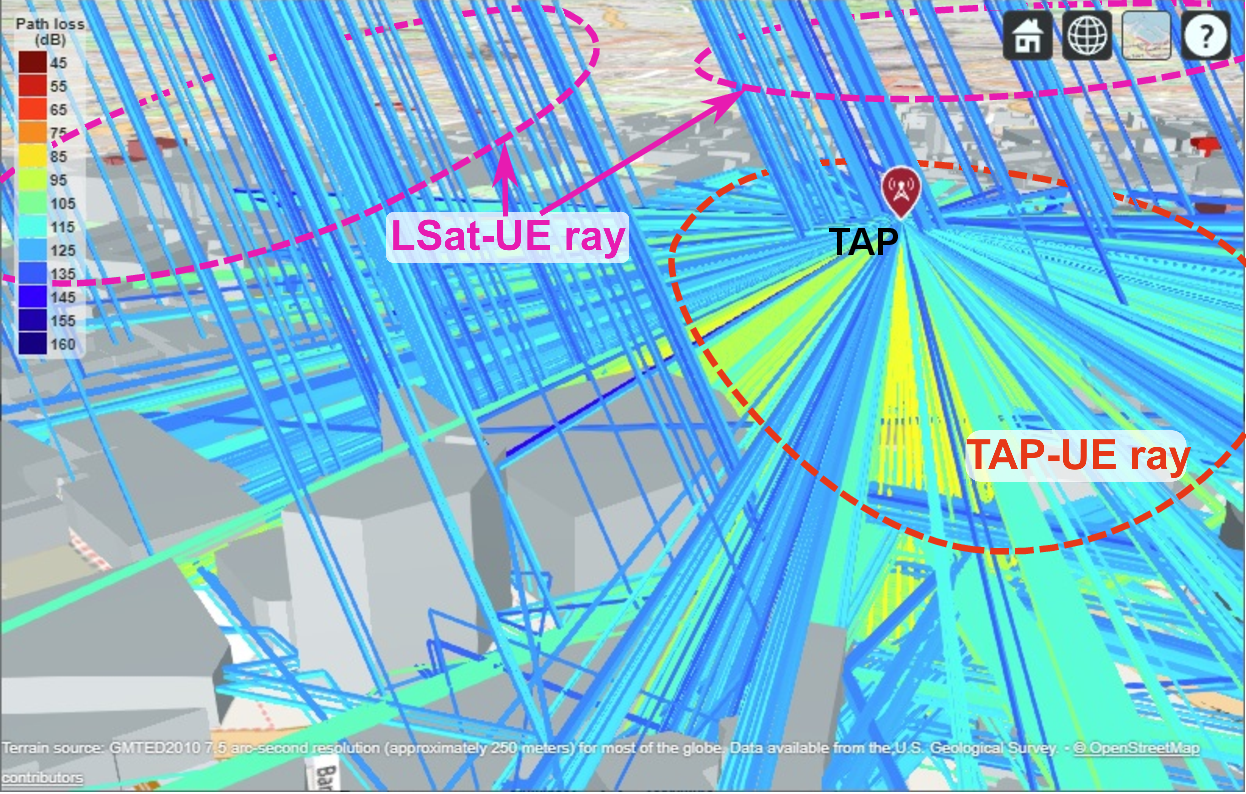}
        \caption{Ray-tracing result in the examined area.}
    \end{subfigure}
    \caption{Simulation scenario, TN channel heatmap (log10 scale), and RayT.}
    \label{fig:heatmap_withUE}
    \vspace{-3mm}
\end{figure}

\subsection{Numerical Results}

\begin{figure}
    \captionsetup{font=footnotesize}
    \centering
    \includegraphics[width=85mm]{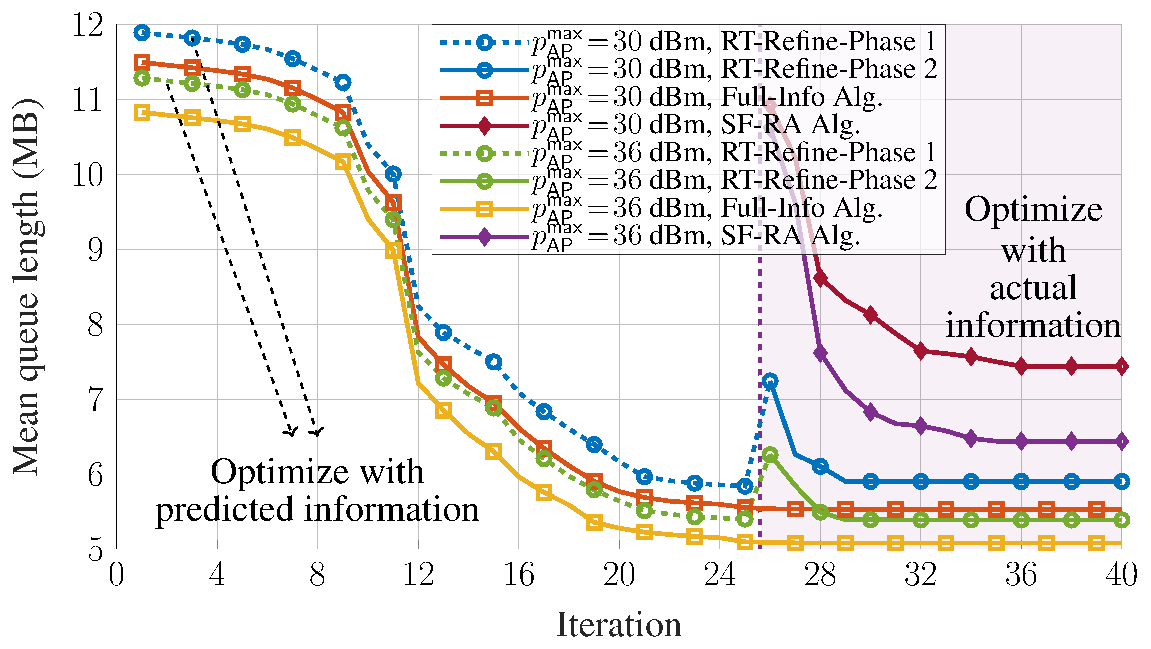}
    \vspace{-2mm}
    \caption{Convergence rate of the examined algorithms.}
    \label{fig:Q_convergence}
    \vspace{-3mm}
\end{figure}
Fig.~\ref{fig:Q_convergence} shows the QL convergence of the proposed algorithm in two phases and of the benchmarks in different cases.
It is worth noting that the RT-Refine algorithm consists of phase 1-the DT-JointRA algorithm and phase 2-the refinement stage, which operate on different input types, i.e., predicted and actual information, respectively.
Since the DT-JointRA algorithm and FIA share the same structure with only differences in inputs, their convergence rates are similar. Particularly, the DT-JointRA algorithm and FIA require only about $25$ iterations for convergence in both cases $p_{\sf{AP}}^{\sf{max}}=30$ dBm and $p_{\sf{AP}}^{\sf{max}}=36$ dBm. Although in the RT-Refine algorithm, phase 2 is executed right after phase 1 (i.e., DT-JointRA), it is bumpy at the beginning of phase 2, which is due to the changes in inputs, i.e., replacing $\{ \hat{\boldsymbol{h}}_{e}, \hat{\boldsymbol{\lambda}}_{e} \}$ by $\{\boldsymbol{h}^{*}_{e}, \boldsymbol{\lambda}_{e} \}$. However, phase 2 requires only about $3$ iterations for convergence. Especially, leveraging extracted output from RT-Refine phase 1 as an initial point for phase 2 facilitates a faster convergence rate, compared to about $10$~iterations of the SF-RA algorithm. This emphasizes the practical implementation aspect of the proposed RT-Refine algorithm.
Additionally, since the solution is improved over iterations, if the practical running time is excessive, the algorithm can be terminated before full convergence to obtain a sufficiently good solution.

Table~\ref{tab:RunningTime} shows the average running time per execution--defined as solving the corresponding problem until convergence--of the proposed framework and the SF-RA algorithm. Since RT-Refine phase 1 involves solving problem $(\mathcal{P}_{2})_{c}$ over an entire cycle, its running time is significantly longer than that of RT-Refine phase 2 and the SF-RA algorithm, which require solving the problem over a single sub-frame. However, it is important to emphasize that RT-Refine phase 1 is a preparation stage, where the problem is solved using predicted information rather than in realtime. Hence, leveraging prediction allows this phase to be executed in advance to meet the horizon time.
Regarding sub-frame execution, by leveraging the initial point provided by RT-Refine phase 1, phase 2 achieves a shorter running time, i.e., approximately $0.2237$~seconds compared to about $0.6452$~seconds for the SF-RA algorithm. In particular, the initial point extracted from the output of RT-Refine phase 1 facilitates faster convergence in phase 2, as depicted in Fig.~\ref{fig:Q_convergence}, thereby reducing the overall execution time.

\begin{table}[!h]
    \vspace{-2mm}
    \captionsetup{font=footnotesize}
    \centering
    \caption{ Average running time per execution.}
	\label{tab:RunningTime}
    \scalebox{1}{
    \begin{tabular}{|l|l|l|l|l|l|}
    \hline
		RT-Refine phase 1 & RT-Refine phase 2 & SF-RA algorithm \\
		\hline 
		$12.1867$ sec & $0.2237$ sec & $0.6452$ sec \\
		\hline
    \end{tabular}}
    \vspace{-3mm}
\end{table}

\begin{figure}[!h]
    \captionsetup{font=footnotesize}
    \centering
    \includegraphics[width=85mm]{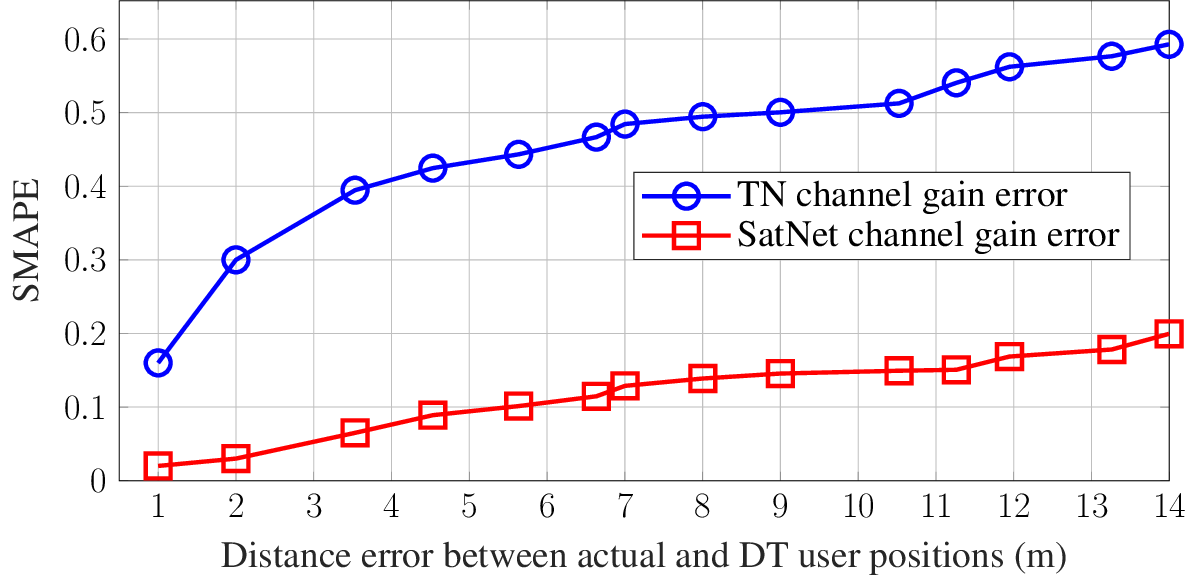}
    \vspace{-2mm}
    \caption{Symmetric mean absolute percentage error (SMAPE) between actual and DT channels versus position mismatch.}
    \label{fig:SMAPE_DisError}
    \vspace{-3mm}
\end{figure}
Fig.~\ref{fig:SMAPE_DisError} depicts the mismatch between the actual and DT channel gains in terms of the SMAPE metric under different levels of position error. In general, larger position errors lead to higher channel gain mismatches. However, at a certain error level, the TN channel exhibits a higher mismatch than that of the SatNet. In the SatNets, due to the long transmission distances and high elevation angles at the receiver, the channel gain variations within the coverage beam are relatively negligible. In contrast, in TNs, the shorter transmission distances and lower elevation angles result in more complicated signal propagation effects, especially in urban environments. Therefore, the position errors have a stronger impact on the TN channel gain. Nevertheless, by leveraging refinement based on feedback from the actual system, the proposed framework can effectively mitigate these channel mismatches, thereby improving network performance and resource allocation decisions, which will be discussed further later.
\begin{figure}[!t]
    \captionsetup{font=footnotesize}
    \centering
    \includegraphics[width=85mm]{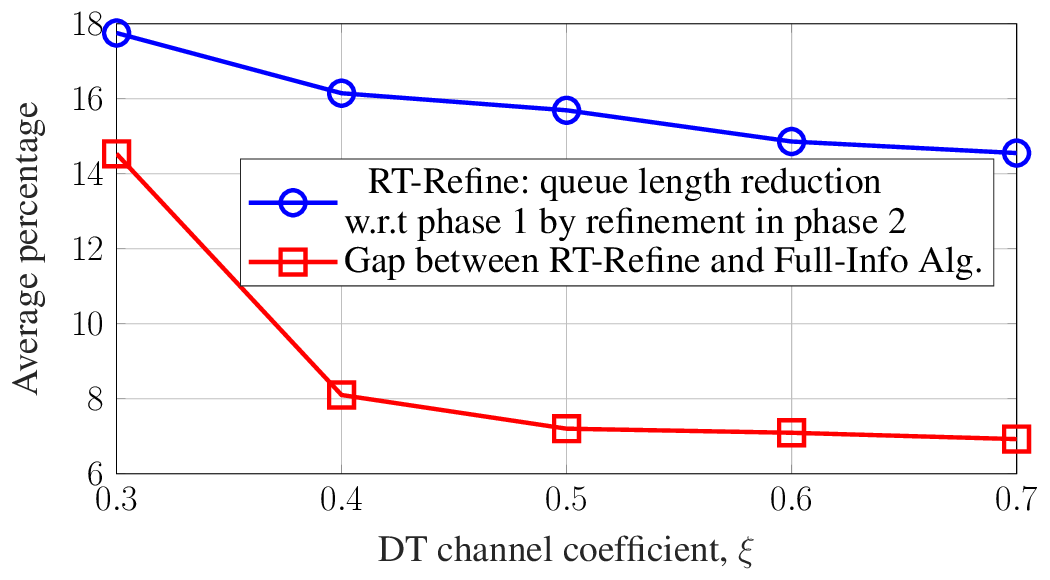}
    \vspace{-2mm}
    \caption{QL versus DT channel coefficient, $\xi$.}
    \label{fig:Q_CoeDT} \vspace{-3mm}
\end{figure}

Regarding the impact of DT channel coefficient $\xi$ on the performance of the proposed RT-Refine algorithm, Fig.~\ref{fig:Q_CoeDT} illustrates the QL improvement of re-optimization in phase 2 of the RT-Refine algorithm and the QL gap between the RT-Refine algorithm and FIA. Based on the channel model in Section~\ref{sec: channel}, one can see that a larger $\xi$ results in a smaller difference between the emulated real environment and its DT reduces. Hence, both lines decrease as $\xi$ increases. Particularly, at $\xi=0.3$, $\xi=0.5$, and $\xi=0.7$, re-optimizing in phase 2 can reduce the QL by about $17.7\%$, $15.7\%$, and $14.5\%$ compared to DT-JointRA outcomes, and the percentage gaps between the RT-Refine algorithm and FIA are about $14.5\%$, $7.3\%$, and $6.9\%$ compared to FIA outcomes,  respectively.

\begin{figure}[!t]
    \captionsetup{font=footnotesize}
    \centering
    \includegraphics[width=85mm]{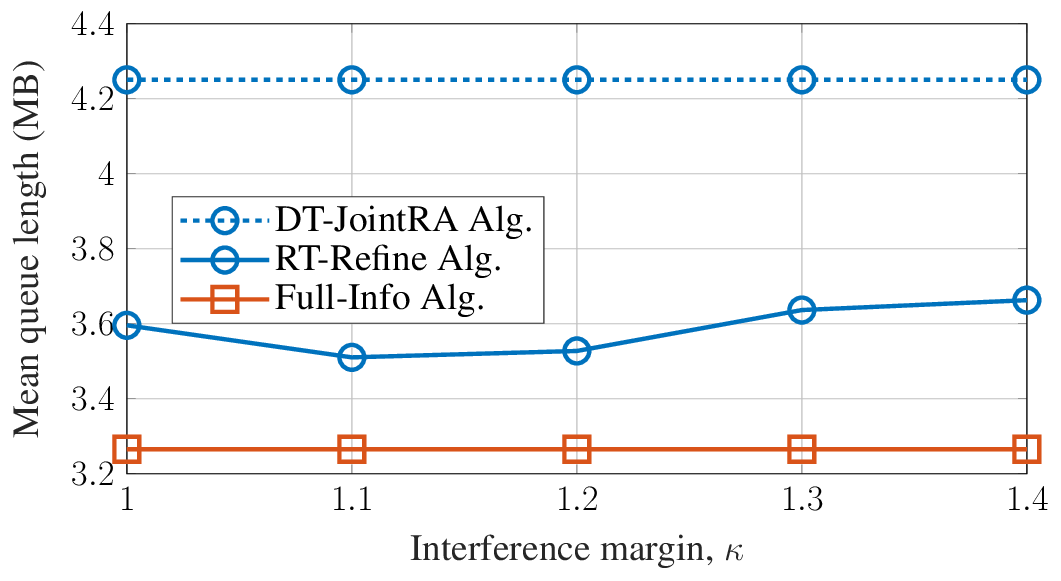}
    \vspace{-1mm}
    \caption{QL versus margin interference.}
    \label{fig:Q_marginH} \vspace{-3mm}
\end{figure}
Fig.~\ref{fig:Q_marginH} depicts the mean QL versus interference margin $\kappa$ in the RT-Refine algorithm. One can see that the mean QL outcome first decreases and then increases as $\kappa$ increases. 
Particularly, the mean QL at $\kappa=(1,1.1,1.2,1.3)$ is about $(3.59,3.51,3.53,3.63)$ MB.
This phenomenon can be explained as follows. First, one notes that in phase 2 of the RT-Refine algorithm, only actual channels from TAPs to their own served UEs are used while other channels are kept as predicted ones. For $\kappa=1$, the prediction error in these interference channels is not accounted for. In contrast, for $\kappa>1$, the prediction uncertainty is considered, promoting TAPs to reduce their transmit power to protect UEs served by LSat. Hence, at higher $\kappa$, e.g., at $\kappa=1.1$ and $\kappa=1.2$, mean QL decreases. However, for sufficiently high $\kappa$, the prediction error is overestimated, leading the TAPs to reduce transmit power more than necessary, which results in rate degradation and an increase in QL. By appropriately choosing this margin, i.e., $\kappa=1.1$, re-optimization in phase 2 can reduce about $0.75$ MB in mean QL, and the gap between the RT-Refine and FIA lines is only about $0.25$ MB.
This result demonstrates the effectiveness of considering this margin against interference channel prediction errors while avoiding costly channel estimation. Based on this result, margin $\kappa=1.1$ is selected for the remaining simulations.

\begin{figure}[!t]
    \captionsetup{font=footnotesize}
    \centering
    \includegraphics[width=85mm]{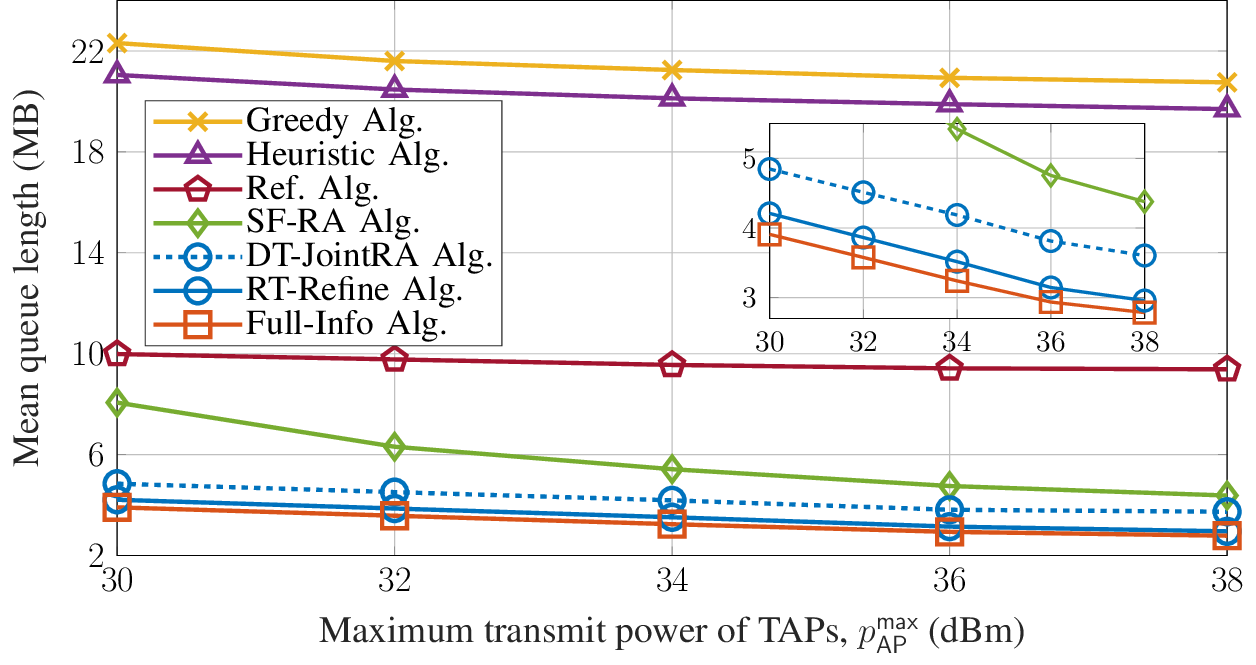}
    \vspace{-1mm}
    \caption{QL versus TAP power budget.}
    \label{fig:Q_PA} \vspace{-3mm}
\end{figure}
Fig.~\ref{fig:Q_PA} shows the mean QL outcomes of proposed algorithms and benchmarks versus the AP power budget. Regarding the greedy and heuristic benchmarks, the mean QL gap between two lines is about $1.12$~MB which is due to the adaptive BWA based on the remaining QL of the heusistic algorithm. 
Moreover, thanks to the joint user association, RB assignment, and power control optimization, the reference algorithm achieves an additional mean QL reduction of approximately $10-11$~MB.
Furthermore, by incorporating the interference-aware design, the SF-RA algorithm further reduces the mean QL about $2-5$~MB. In particular, the performance gap between the reference and SF-RA algorithms increases with the TAP power budget, as the interference-aware design enables more effective power control to mitigate harmful interference, rather than merely increasing transmit power proportionally to enhance the desired signal.
However, compared to the proposed algorithms which incorporate the DT-based network decision design, the QL performance gap between the proposed and benchmark algorithms is still significant, i.e., approximately $17$~MB, $5-5.7$~MB, and $1.5-4$~MB compared with greedy/heustic, reference, and SF-RA algorithms, respectively. 
Additionally, the proposed algorithms satisfy $\sf{D}$ service requirements in all considered cases while the unserved $\sf{D}$ traffics of the heuristic and reference algorithms are about $22.2\%$ and $17.5\%$, respectively.
Regarding proposed algorithms, the gap between the DT-JointRA and RT-Refine lines, which indicates improvement by the refinement in RT-Refine-phase 2, is about $0.7$~MB. Besides, compared to the proposed benchmark FIA which uses full actual information for optimization, the gap between the RT-Refine and FIA lines is only about $0.27$~MB. Hence, in addition to superiority compared to two benchmarks, the proposed RT-Refine scheme shows the effectiveness of refining TAP power control to adjust solutions with the actual information at all considered cases of $p_{\sf{AP}}^{\sf{max}}$. 
Additionally, a similar outcome trend of the proposed algorithms can be recognized. Particularly, with TAP power budget $p_{\sf{AP}}^{\sf{max}}=30 \rightarrow 36$~dBm, the mean QL decreases quickly, and with $p_{\sf{AP}}^{\sf{max}}=36 \rightarrow 38$~dBm, their decrease trend slightly reduces. 
While a higher power budget enables higher data rates, certain traffic flows can be fully served once the power budget reaches a sufficiently high level.
This phenomenon shows that beyond $p_{\sf{AP}}^{\sf{max}}=36$~dBm, the power budget starts to be redundant for QL minimization.

\begin{figure}[!t]
    \captionsetup{font=footnotesize}
    \centering
    \includegraphics[width=85mm]{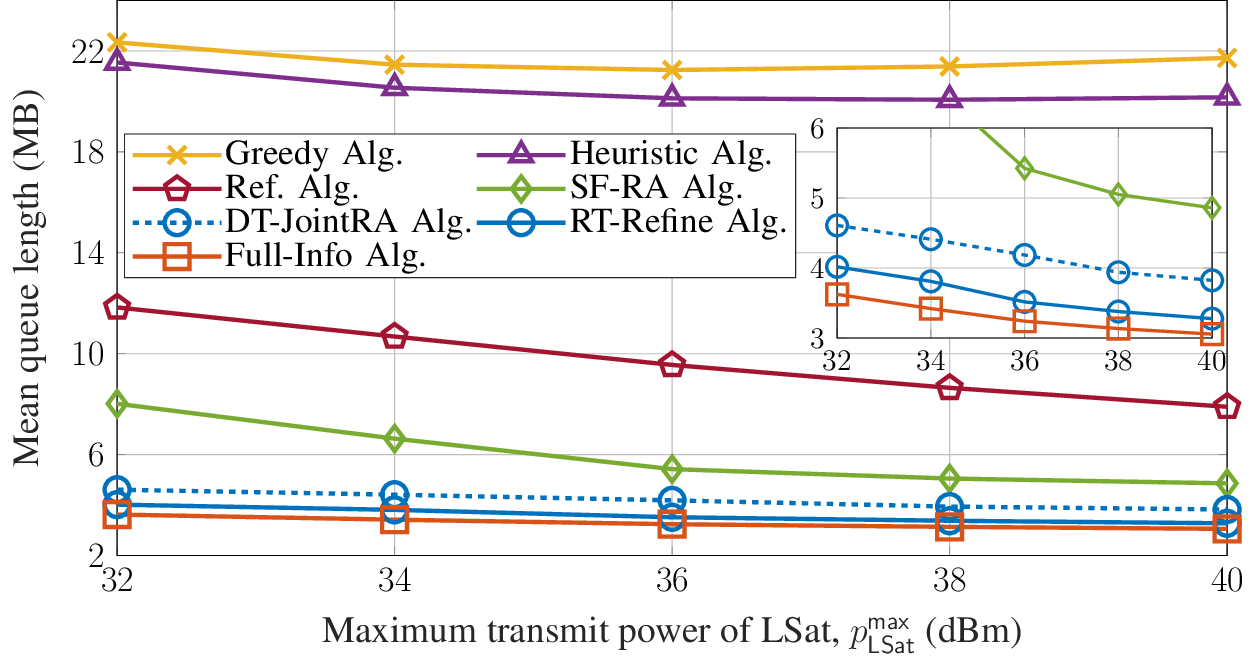}
    \vspace{-2mm}
    \caption{QL versus LSat power budget.}
    \label{fig:Q_PS} \vspace{-4mm}
\end{figure}
Fig.~\ref{fig:Q_PS} depicts the mean QL performance of the examined schemes with different LSat power budgets. Regarding the greedy and heuristic benchmarks, one can see that the mean QL outcomes first decrease and then increase. This is due to the impact of interference caused by LSat in the whole coverage area is not addressed in the power control phase. Similar to results in Fig.\ref{fig:Q_PA}, the adaptive BWA in the heuristic scheme brings the improvement in terms of QL compared to the greedy one while the joint user association, RB assignment, and power control optimization in the reference one further reduce mean QL. 
Additionally, by incorporating the interference-aware design, the SF-RA algorithm achieves an additional QL reduction of about $3-3.8$~MB.
However, compared to the proposed algorithms, the QL outcome gap is still significant. Regarding the proposed algorithms, a similar trend compared to results in Fig~\ref{fig:Q_PA} can be recognized. Particularly, the QL outcomes linearly decrease as LSat power budget $p_{\sf{LSat}}^{\sf{max}}$ increases, while the decreasing trend slightly reduces with $p_{\sf{LSat}}^{\sf{max}}=36 \rightarrow 40$ dBm. This slight degradation indicates that the LSat power budget starts to be redundant for QL minimization.

\begin{figure}[!t]
    \captionsetup{font=footnotesize}
    \centering
    \includegraphics[width=1\linewidth]{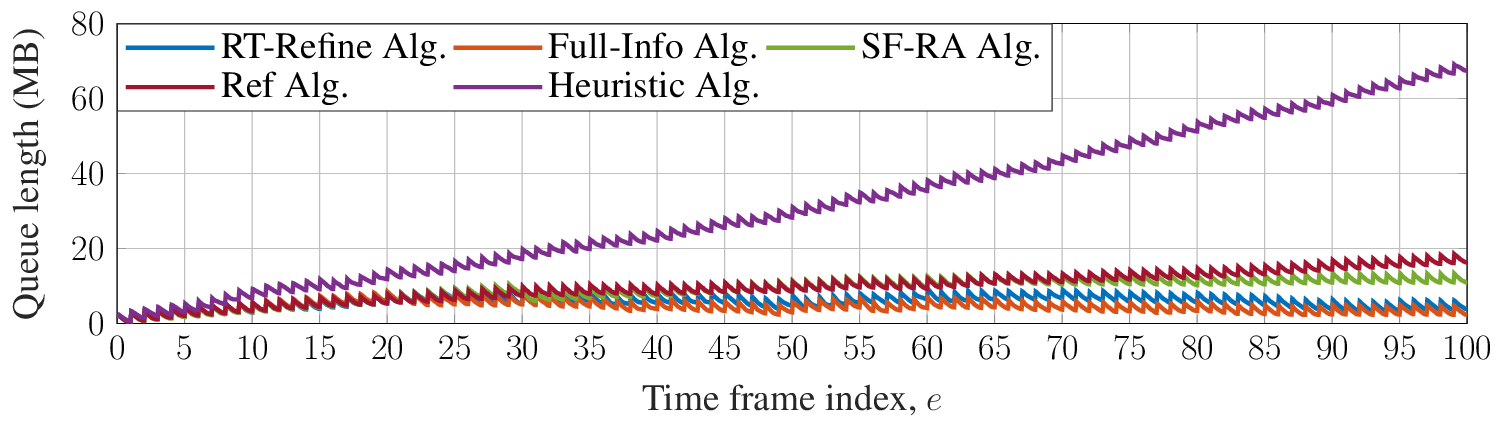}
    \vspace{-6mm}
    \caption{QL over time of examined algorithms.}
    \label{fig:Qtotal_frame} \vspace{-3mm}
\end{figure}
For insights into QL minimization performance, Fig.~\ref{fig:Qtotal_frame} illustrates the QL evolution of FIA, RT-Refine, the heuristic method, the reference, and the SF-RA algorithms. The heuristic benchmark exhibits a rapidly and almost linearly increasing QL, resulting in a significantly higher average QL than the proposed and reference schemes. 
Although the reference and SF-RA algorithms achieve the lower QL, it remains unstable as its QL continues to grow over time. In contrast, the proposed algorithms not only provide substantially better QL performance but also maintain the QL consistently below approximately $10$~MB across frames. These results demonstrate the superiority of the proposed solutions in terms of both QL minimization and long-term stability.

\begin{figure}[!t]
   \captionsetup{font=footnotesize}
    \centering
     \begin{subfigure}{1\linewidth}
         \centering
        \includegraphics[width=1\linewidth]{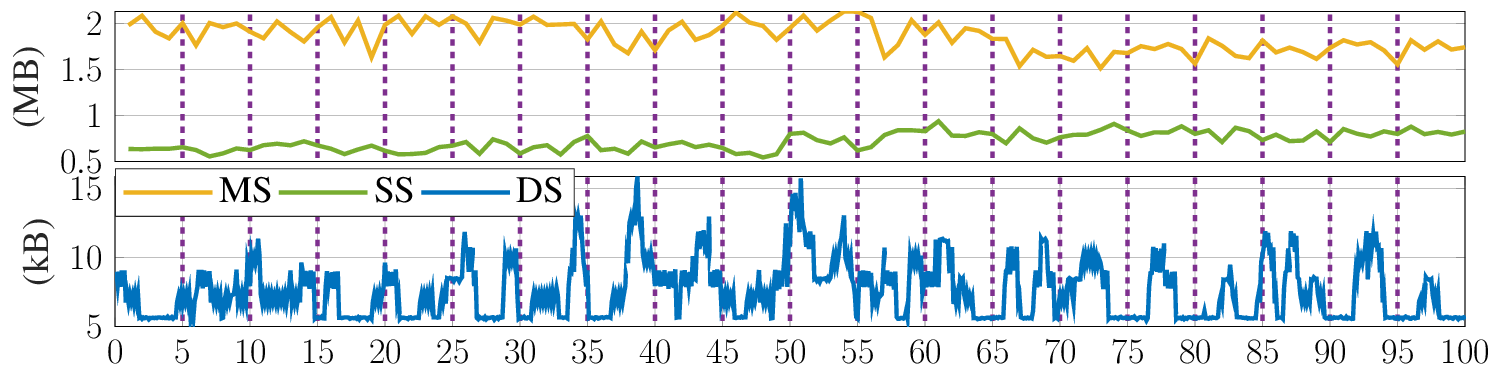}
        \vspace{-2mm}
        \caption{Arrival data of services.}
        \label{fig:lambda_frame}
    \end{subfigure}
    \begin{subfigure}{1\linewidth}
        \centering
        \includegraphics[width=1\linewidth]{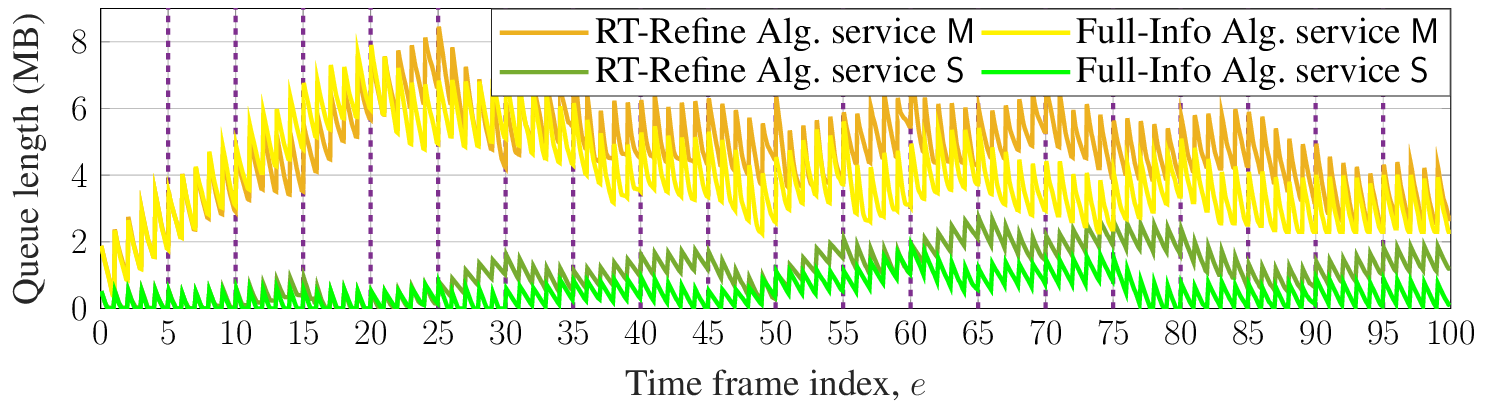}
        \vspace{-2mm}
        \caption{QL over time.}
        \label{fig:Q_frame}
    \end{subfigure}
    \begin{subfigure}{1\linewidth}
        \centering
        \includegraphics[width=1\linewidth]{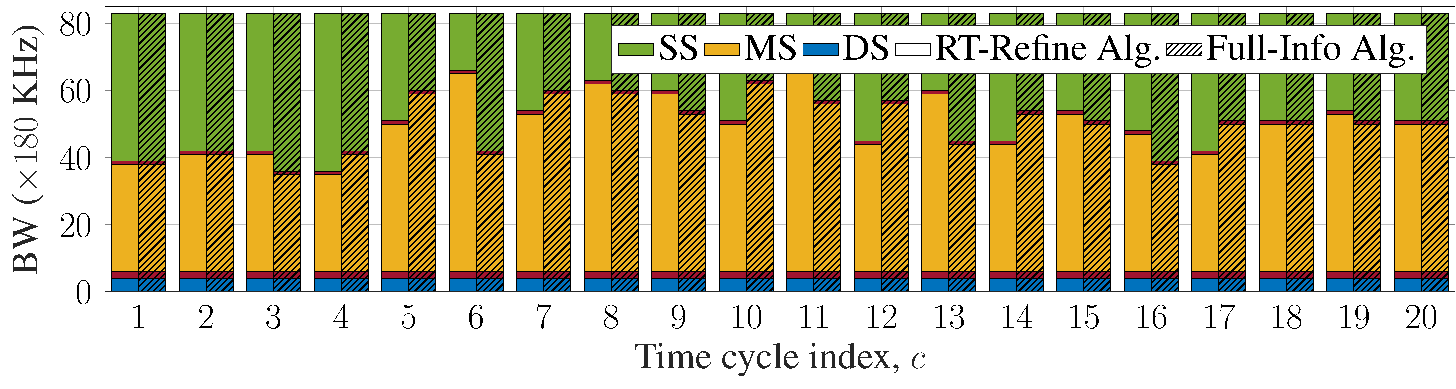}
        \vspace{-2mm}
        \caption{BW allocation over time.}
        \label{fig:BW_cycle}
    \end{subfigure}
    \vspace{-6mm}
    \caption{Arrival data, QL, and BWA over time.}
    \label{fig:result_time} \vspace{-3mm}
\end{figure}
For more insight into the QL over time, Fig.~\ref{fig:result_time} shows the arrival data of services, QL and BWA outcomes over time of the RT-Refine algorithm and FIA. Although the arrival data of $\sf{D}$ services is quite small, i.e., only about $5-10$~kB/frame, it is not completely served by using benchmarks since the interference is ignored. However, the proposed algorithms can satisfy $\sf{D}$ service requirements by allocating only $1$ SC for BWP $\sf{D}$. Besides, due to the larger varying total arrival data size of the $\sf{M}$ and $\sf{S}$ services, the BWAs for these services change over cycles for QL minimization. One can see that, at each cycle, the BW is allocated adaptively based on the remaining QL in the previous cycle and the arrival data. Additionally, compared with the FIA benchmark using full actual information, the RT-Refine algorithm, which is proposed for implementation purposes, can provide a comparable performance in terms of QL minimization over time.

\vspace{-3mm}

\section{Conclusion} \label{sec: conclusion}
\vspace{-1mm}
This work studied the novel DT-aided joint spectrum sharing, traffic steering, and resource allocation for ISTNs over time. The optimization problems aim to minimize the system congestion under system and service constraints. To assist RM, a digital-twin system is developed for information prediction. Based on this, we proposed a DT-JointRA algorithm implementable at the CNC to address the optimization problem. For practical implementation, a re-optimization phase is built on top of the DT-JointRA algorithm to refine its solution with the information in real-world systems. Through numerical results with the actual 3D map and traffic information, the proposed algorithms show the superiority in terms of congestion minimization, compared to benchmarks. Furthermore, the simulation results emphasize the practicality and effectiveness of the proposed algorithms in terms of adaptation ability.

The proposed presolve--refinement strategy is applicable to other DT-based systems, where problems can be presolved using predicted information to support long-term decision-making and satisfy execution time horizon, and resulting solutions are subsequently refined based on real-world feedback to better align with actual system conditions. Furthermore, utilizing the DT-based solution as an initial point is particularly beneficial for accelerating the refinement convergence.

\vspace{-3mm}
\appendices
\section{Proof of Proposition~\ref{pro: rate cons}} \label{app: rate cons}
Consider function $f_{\sf{r}}(x,y)=\log_2(1+\frac{x}{y+a}), x,y \geq 0, a > 0$, constraint $f_{\sf{r}}(x,y) \geq z, \; z\geq 0$ is approximated as \cite{Hung_TCOM25}

\begin{align}\label{eq: apx rate} 
\begin{cases}
    \log_2(x+y+a) \geq z + u / \ln(2), \\
    y + a \leq f_{\sf{exp}}^{(i)}(u) \triangleq \exp(u^{(i)})(u - u^{(i)} + 1), 
\end{cases}
\end{align}
where $u$ and $u^{(i)}$ are the slack variable and its feasible point at iteration $i$, respectively.
Apply approximation in $\eqref{eq: apx rate}$ to log component in rate functions $R_{\ell,k}^{{\sf{M}},[v_{\sf{M}},n_{\sf{M}}]}$ and $R_{0,k}^{{\sf{M}},[v_{\sf{M}},n_{\sf{M}}]}$ with setting $x$, $y$, $a$, and $u$ by the desired received power, the total interference terms, noise power, and the corresponding $\eta$, constraints $(C17.1)$ and $(C18)$ are convexified by $(\tilde{C}17a)$, $(\tilde{C}17b)$ and $(\tilde{C}18)$, respectively. Specifically, $R_{0,k}^{{{\sf{S}}},[v_{\sf{S}},n_{\sf{S}}]}(\boldsymbol{P}_{\!c})$ is concave, hence, constraint $(\tilde{C}20c)$ is convex.

\vspace{-3mm}
\section{Proof of Proposition~\ref{pro: SINR cons}}\label{app: SINR cons}
Constraint $(\bar{C}10)$ is equivalently transformed as
\beqn \label{eq: relax C10}
    \log_2(1+\gamma_{\ell,k}^{{\sf{D}},[v_{\sf{D}},n_{\sf{D}}]}(\boldsymbol{P}_{\!c})) \geq \mathcal{F}_{\! \sf{apx}}^{(i)}\!(p_{\ell,k}^{{\sf{D}},[v_{\sf{D}},n_{\sf{D}}]}) \log_2(1 + \gamma_{0}^{\sf{D}}), 
\eeqn
$\forall (\ell ,k,v_{\sf{D}},n_{\sf{D}})$. Apply approximation \eqref{eq: apx rate} to the log component by setting $x=p_{\ell,k}^{[v_{\sf{D}},n_{\sf{D}}]} h_{\ell,k}^{[v_{\sf{D}},n_{\sf{D}}]}$, $y=\Psi_{\ell,k}^{[v_{\sf{D}},n_{\sf{D}}]}$, $a=\sigma_{{\sf{D}},k}^{2}$, and $u=\eta_{\ell,k}^{[v_{\sf{D}},n_{\sf{D}}]}$, we obtain $(\tilde{C}10a)$ and $(\tilde{C}10b)$.

\vspace{-2mm}
\section{Proof of Proposition~\ref{pro: DS rate cons}} \label{app: DS rate cons}
\vspace{-1mm}
Let $\zeta_{\ell,k,[s]}$ be the upper bound of the RB number allocated to ${\sf{AP}}_{\ell}-{\sf{UE}}_{k}^{\sf{D}}$ link in SF $s$, constraint $(C17.2)$ is rewritten as
\begin{IEEEeqnarray}{ll} \label{eq: DS rate relax}
    w_{\sf{D}} \!\!\!\! \scaleobj{0.8}{\sum_{ \forall (v_{\sf{D}},n_{\sf{D}}) } } \!\!\! \log_2{(1 + \gamma_{\ell,k}^{{\sf{D}},[v_{\sf{D}},n_{\sf{D}}]}(\boldsymbol{P}_{\!c})}) - \chi_{\sf{D}} \scaleobj{0.8}{\sqrt{\zeta_{\ell,k,[s]}}} \geq r_{\ell,k,[s]}^{{\sf{SF,D}}}, \hspace{6mm} \subnum \label{eq: DS rate relax a} \\
    \zeta_{\ell,k,[s]} \geq \scaleobj{0.8}{\sum\nolimits_{{\forall (v_{\sf{D}},n_{\sf{D}}) } }} \Vert p_{\ell,k}^{[v_{\sf{D}},n_{\sf{D}}]} \Vert_{0}. \subnum \label{eq: DS rate relax b}
\end{IEEEeqnarray}
Apply approximation \eqref{eq: apx norm0} to $\ell_{0}$-norm components, \eqref{eq: DS rate relax b} is convexified as $(\tilde{C}17d)$. 
In \eqref{eq: DS rate relax a}, the log component can be convexified as in proposition~\ref{pro: DS rate cons} while the square root component is convexified as follows. Apply SCA technique, function $\mathcal{F}_{\sf{sqrt}}(x)=\sqrt{x}$ is approximated at iteration $i$ as 
\begin{equation} \label{eq: apx sqrt}
    \mathcal{F}_{\sf{sqrt}}(x) \leq {x}/{(2 \scaleobj{0.8}{\sqrt{x^{(i)}}})} + \scaleobj{0.8}{\sqrt{x^{(i)}}}/2 := \mathcal{F}_{\sf{sqrt}}^{(i)}(x).
\end{equation}
Use approximation \eqref{eq: apx sqrt} for $\mathtt{sqrt}$ term in $\eqref{eq: DS rate relax a}$, we obtain the convex form as $(\tilde{C}17c)$.

\vspace{-4mm}
\bibliographystyle{IEEEtran}
\bibliography{Journal}
\end{document}